\documentclass[a4paper,onecolumn,11pt,unpublished]{quantumarticle}
\pdfoutput=1

\usepackage{tikz_commands}

\allowdisplaybreaks

\usepackage[margin=1cm]{geometry}
\renewcommand{\headwidth}{\textwidth}

\usepackage{thm-restate}
\usepackage{thmtools,nameref}
\usepackage[noabbrev]{cleveref}

\declaretheorem[name=Proposition,
	refname={proposition,propositions},
	Refname={Proposition,Propositions}]{prop}

\declaretheorem[name=Lemma,
	refname={lemma,lemmas},
	Refname={Lemma,Lemmas},
	numberlike=prop]{lemma}

\declaretheorem[name=Definition,
	refname={definition,definitions},
	Refname={Definition,Definitions},
	numberlike=prop]{definition}

\crefalias{figure}{figure}

\usepackage{framed}
\newenvironment{claim}{\underline{Claim:} }{}
\newenvironment{claimproof}{\begin{leftbar}\noindent}{\end{leftbar}}

\begin{document}

\title{Outer approximations of classical multi-network correlations}

\author{Victor Gitton}
\affiliation{Institute for Theoretical Physics, ETH Zürich, Switzerland}
\email{vgitton@ethz.ch}
\maketitle

\begin{abstract}

We propose a framework, named the postselected inflation framework, to obtain converging outer approximations of the sets of probability distributions that are compatible with classical multi-network scenarios. 
Here, a network is a bilayer directed acyclic graph with a layer of sources of classical randomness, a layer of agents, and edges specifying the connectivity between the agents and the sources.
%
A multi-network scenario is a list of such networks, together with a specification of subsets of agents using the same strategy.
%
%
%
%
%
We furthermore show that the postselected inflation framework is mathematically equivalent to the standard inflation framework: in that respect, our results allow to gain further insights into the convergence proof of the inflation hierarchy of Navascuès and Wolfe, and extend it to the case of multi-network scenarios.

\end{abstract}

\tableofcontents

\newpage
\section{Introduction}

\paragraph{Causal compatibility.}

The problem of causal compatibility is the problem of deciding whether a given probability distribution can arise from a given causal structure.
The study of causal compatibility in quantum mechanics can be traced back to Bell's theorem \cite{bell_einstein_1964}. In modern language, this result can be understood as a proof that there exist outcome distributions compatible with a certain quantum causal model (Alice and Bob sharing an entangled quantum state and having classical inputs) yet incompatible with the corresponding classical causal model (where the entangled quantum state is replaced with shared classical randomness). This result, apart from its fundamental implications for possible theories of natures, turns out to be crucial for quantum cryptography \cite{ekert_quantum_1991,pirandola_advances_2020}.

Beyond the setting of Bell's theorem, there are a number of reasons to be interested in causal compatibility in greater generality \cite{tavakoli_bell_2021}. Here, we are concerned with causal structures called networks that feature a number of independent, unobserved parameters, which we refer to as ``sources'', that may influence a number of observed parameters, which we refer to as the outcomes of ``agents'': this naming convention reflects the quantum information mindset where the observed parameters would be the outcomes of measurements that human agents would carry out in a lab.
Such networks come in different flavors depending on the type of physical systems that implement the unobserved parameters: these are typically taken to be classical (e.g., bit strings sent out to the agents), quantum (featuring in particular quantum systems whose states may be entangled with respect to the different agent labs they are sent out to) or merely non-signaling (i.e., physical systems that are less constrained than quantum systems, whose internal description is not specified, but that nonetheless verify certain non-signaling conditions).
Some important results from the perspective of quantum information theory are the demonstration of non-locality without inputs \cite{renou_genuine_2019}, of full network non-locality (where all sources have to be non-classical) \cite{pozas-kerstjens_full_2022}, as well as the necessary existence of $N$-partite ``entanglement'' in any non-signaling theory of nature \cite{coiteux-roy_any_2021}.

In this work, we wish to investigate the structure of achievable outcome distributions in classical multi-network scenarios. These can be understood as networks of classical observers sharing maximally entropic classical sources of randomness in a specific arrangement, allowing for some observers to be using the same strategy (i.e., responding identically to the same set of inputs).
Accounting for such same-strategy constraints can be seen, from an operational standpoint, as a ``markovianity'' constraint on the agents. For instance, if the agents are in fact memory-less black boxes that can be plugged in various positions of the network, then such same-strategy constraints would arise.
In fact, a distribution that is causally incompatible with a given same-strategy multi-network scenario but is causally compatible with the corresponding ``any-strategy'' multi-network scenario is a distribution in which these memory effects are non-trivial.
Alternatively, the same-strategy constraints can be desirable in applications: this may be the case of a network in which each of many agents choose among two strategies depending on the correlations they wish to achieve with the rest of the network.
While we focus on the theory behind classical networks, such same-strategy constraints were already investigated in the non-signaling case \cite{bancal_non-local_2021}.


\paragraph{Methods.}

The problem of causal compatibility is typically hard to solve analytically. If one is interested in inner approximations of the set of achievable distributions of a given causal structure, then one can resort to an analytical or numerical sampling of the underlying search space --- namely, the space of all agent strategies and all source behaviors within a given theory.
A notably efficient tool in this direction is the neural network oracle for classical causal compatibility introduced in \cite{krivachy_neural_2020}.

In order to obtain guarantees that a fair amount of the search space has been sampled, it is crucial to have access to tractable outer approximations of the set of outcome distributions. A useful review for that purpose is that of \cite{tavakoli_bell_2021}. One basic analytical tool is Finner's inequality \cite{renou_limits_2019}, that provides simple bounds on achievable correlations in non-signaling networks.
In the case of classical networks, one can analytically certify the infeasibility of certain distributions based on so-called rigidity arguments \cite{renou_network_2020}. 
A general possibility to further study classical, quantum or non-signaling networks is to use entropy-based outer approximations of the feasible correlations \cite{chaves_informationtheoretic_2015,chaves_entropic_2016,weilenmann_analysing_2017}. Additionally, for quantum or classical networks, one may also use the semidefinite-programming-based relaxation of \cite{pozas-kerstjens_bounding_2019}, that is based on building a positive semidefinite correlation matrix augmented with scalar operators that enable the incorporation of conditional independence relations, if the network at hand features such conditional independences.

The only method known to date of generating outer approximations that converge to the actual set of feasible correlations in classical networks is the technique of inflation. Inflation is a general technique that can come in three flavors for classical, quantum and non-signaling networks.
The general idea is simple: given a certain network and an outcome distribution for the agents of this network, one could in principle have access to several copies of sources and agents, wire them in various ways, and then obtain a sensible probability distribution that should verify certain compatibility conditions with respect to the original outcome distribution.
In the classical case, the inflation technique --- which we will call ``fanout inflation'' --- was originally introduced in \cite{wolfe_inflation_2019}, and later proven to converge asymptotically in \cite{navascues_inflation_2020}.
The case of quantum and non-signaling inflation is different to the classical case: while classical information can be freely cloned, this does not hold for quantum and non-signaling information. Hence, in classical inflation, one makes use of ``fanout'' inflation graphs with explicit cloning of classical information, while in the quantum and non-signaling cases, the inflation graphs are in that sense ``non-fanout''. 
A description of quantum inflation can be found in \cite{wolfe_quantum_2021}. Some recent developments regarding the potential convergence of quantum inflation can be found in \cite{ligthart_convergent_2021}, but there remains a ``rank constraint'' loophole to be addressed.
The case of non-signaling inflation was discussed initially in \cite{wolfe_inflation_2019}, and then further developed in e.g.\ \cite{gisin_constraints_2020,coiteux-roy_any_2021}.
From a practical perspective, inflation is typically handled as a linear program in the classical and non-signaling case, while it takes the form of a semidefinite program in the quantum case.

\paragraph{Objectives.}

The proof of convergence of the classical inflation \cite{navascues_inflation_2020} is a rather surprising result that deserves some attention.
The primary objective of the present manuscript is to gain additional insights as to how the proof works.
This desire eventually yielded the postselected inflation formulation, which can be understood as an equivalent formulation of fanout inflation: the equivalence holds for the outer approximations that these two schemes can generate, but also in terms of the linear programs that one would solve in either formulation.
Interestingly, in the postselected inflation formulation, the convergence of the outer approximations is rather straightforward. On the other hand, the fact that the postselected inflation scheme yields outer approximations of the relevant set of outcome distributions is non-trivial.
The situation is the opposite with the fanout inflation formulation, which clearly yields outer approximation, but whose convergence is rather hidden.
In that sense, we strongly encourage the interested reader to gain familiarity with both formulations, as they complement each other with respect to the intuition that one gains from knowing about them.
A possible structure for this manuscript would have been to start introducing fanout inflation, and then work our way towards the postselected inflation formalism.
Instead, we choose to temporarily pretend, for the sake of pedagogy, that fanout inflation does not exist, and motivate and prove the soundness and convergence of the postselected inflation formalism from the bottom up.
This brings the opportunity to prove the convergence of certain inflation hierarchies in the contexts of classical multi-network correlations involving same-strategy constraints.
At last, we will show explicitly some working examples of the correspondence between fanout inflation and postselected inflation.

\paragraph{Outline.}

The Correlated Sleeper example is introduced first in section \cref{sec:intro sc} as a motivation to the kind of problems that the inflation framework can solve. 
We then introduce the relevant tensor network notation in section \ref{sec:tensor notation}. 
The multi-network scenarios and causal compatibility problem that we will consider are introduced in section \ref{sec:causal compat}. 
The next section \ref{sec:outer approx} is there to motivate the postselected inflation framework, and to intuitively show that the resulting scheme is convergent.
We collect certain basic formal results regarding the postselected inflation scheme in section \ref{sec:post-selected inflation} --- namely, that the scheme generates outer approximations that are increasingly tight and that eventually converge. 
We apply the postselected inflation formalism to the Correlated Sleeper in section \ref{sec:applications} as a concrete example.
The correspondence between postselected and fanout inflation is made explicit in section \ref{sec:fanout inflation}.

\newpage
\section{The Correlated Sleeper}
\label{sec:intro sc}

Let us introduce the Correlated Sleeper task: in the rest of the article, we will use it as the main example to which we will apply our framework. This task involves an agent $A$ that will be subject to two rounds of interrogation. In each of these rounds, $A$ has to output a number, say, either $1$ or $2$. If the outputs in the first and second round are equal, $A$ wins a prize, and $A$'s objective is to maximize her average probability of winning.

To succeed in this task, $A$ has access to two inputs, referred to as her left and right inputs.
These inputs each provide $A$ with a real number between $0$ and $1$. One of the two inputs will contain the same number (drawn uniformly at random between $0$ and $1$) across the two rounds --- this is the ``faithful'' input. The other input will contain two independently drawn numbers during the two rounds. However, $A$ does not know which of her two inputs is the
faithful one: this is decided according to the toss of a fair coin to which $A$ does not have access. The situation is summarized in \cref{fig:protocol_flow}. On top of her inputs, $A$ may also use some local randomness (for instance, she may flip a coin to decide which input to trust).

\begin{figure}[h!]
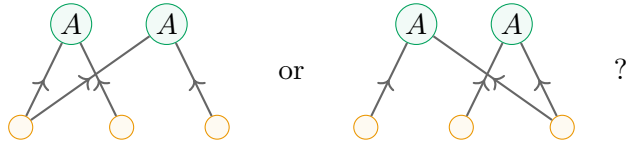

\centering
\begin{equation*}
\centertikz{
\node[agent] (a1) {$A$};
\node[agent] (a2) [right=20pt of a1] {$A$};
\node[source] (alpha) [below left=30pt and 10pt of a1] {};
\drawabsleg (alpha) -- (a1);
\drawabsleg (alpha) -- (a2);
\node[source] (beta1) [below right=30pt and 10pt of a1] {};
\drawabsleg (beta1) -- (a1);
\node[source] (beta2) [below right=30pt and 10pt of a2] {};
\drawabsleg (beta2) -- (a2);
}
\quad
\textup{ or }
\quad
\centertikz{
\node[agent] (a1) {$A$};
\node[agent] (a2) [right=20pt of a1] {$A$};
\node[source] (alpha1) [below left=30pt and 10pt of a1] {};
\drawabsleg (alpha1) -- (a1);
\node[source] (alpha2) [below left=30pt and 10pt of a2] {};
\drawabsleg (alpha2) -- (a2);
\node[source] (beta) [below right=30pt and 10pt of a2] {};
\drawabsleg (beta) -- (a1);
\drawabsleg (beta) -- (a2);
}
\quad\textup{?}
\end{equation*}
\caption{In the Correlated Sleeper protocol, effectively, the agent $A$ has a uniform prior over the above two networks. By assumption, the agent $A$ has one unique strategy independent of her location in the networks. Each $\protect\centertikz{\protect\node[source] {};}$ node denotes an independent source of randomness that generates a uniform number in $[0,1]$; the $\protect\centertikz{\protect\drawabsleg (0,0) -- (20pt,0);}$ edges specify to which source each instance of the agent $\protect\centertikz{\protect\node[agent] {$A$};}$ has access; and the bottom-left and bottom-right inputs of the $\protect\centertikz{\protect\node[agent] {$A$};}$ node are distinguishable by $A$.}
\label{fig:protocol_flow}
\end{figure}

Now, there are some restrictions at play. The first restriction is that the agent $A$ will have no memories of the first round during the second round, and there is no information available to $A$ that allows her to distinguish between the two rounds. This is in analogy with the setting of the Sleeping Beauty paradox \cite{elga_self-locating_2000}, in which the Sleeping Beauty is woken up multiple times without knowing how many times she has been woken up earlier on.
This justifies the fact that $A$ has to use the same strategy during the two rounds. This strategy is most generally captured by the set of conditional probabilities $p_A(a|\alpha,\beta)$ of $A$ giving the outcome $a$ upon seeing that her left (resp.\ right) input number is $\alpha$ (resp.\ $\beta$).
Alternatively, and perhaps more realistically, the agent $A$ can be thought of as a computer that has been programmed to use the strategy $p_A$ and is rebooted between the two rounds.
The second restriction is a marginal constraint on the strategy that $A$ may use: on average over the first round, her outcome distribution must be uniform over the two values $1$ and $2$. This is expressed as: for all $a\in\{1,2\}$,\footnote{We will typically leave the integration domains implicit in this work. Here, the integral runs over $[0,1]^{\times 2}$.}
\begin{equation}
\label{eq:marginal constraint}
\int\dd\alpha\dd\beta\, p_A(a|\alpha,\beta) = \frac{1}{2}.
\end{equation}

The question is then: what is the maximal probability $\aopt$, optimized over the strategy $p_\alice$, of $\alice$ winning the game? For instance, a viable strategy could be the following: choose
\begin{equation}
\label{eq:app example strat}
p_{\alice}(\outputa|\alpha,\beta) = \left\{\begin{aligned}
\kdelta{\outputa}{1}, \textup{ if } \alpha < 0.5, \\
\kdelta{\outputa}{2}, \textup{ if } \alpha \geq 0.5,
\end{aligned}
\right.
\end{equation}
corresponding to $\alice$ only looking at her left input $\alpha$.
The single-round output of $\alice$ is indeed uniform over $\{1,2\}$. Then, if we are in the case where the left input
is the faithful one, it will hold that (denoting with $\outputa_i$ the output of $A$ in the round $i \in \{1,2\}$) $\outputa_1 = \outputa_2$ with certainty. In the case where the right input
is the faithful one, then $\alice$ is effectively choosing $\outputa_1$ and $\outputa_2$ independently and uniformly over $\{1,2\}$, so that $\outputa_1 = \outputa_2$ occurs with probability $1/2$. On average, the probability of success is thus of $3/4$, which means that $\aopt$ is at least $3/4$. 

Let us denote the outcome distribution of $A$ when the left (resp.\ right) input was the faithful one as $p_1$ (resp.\ $p_2$) (a probability distribution over the set $\{1,2\}\times\{1,2\}$). We can thus write $\aopt$ as the following optimization problem:
\begin{subequations}
\label{eq:sleeping correlations plain}
\begin{align}
\aopt\ :=\ &\sup_{p_{\alice}} \ 
\ \frac{1}{2}\sum_{\outputa\in\{1,2\}} \left(p_1(\outputa,\outputa) + p_2(\outputa,\outputa)\right), \hspace{4cm}\\
&\ \textup{s.t. }
\forall \alpha,\beta\in[0,1], \forall \outputa,\outputa_1,\outputa_2 \in \{1,2\}:
\end{align}
\begin{align}
p_{\alice}(\outputa|\alpha,\beta) &\geq 0,
\quad
\sum_{\outputa'\in\{1,2\}} p_{\alice}(\outputa'|\alpha,\beta) = 1,\label{eq:sleeping correlations norm} \\
p_1(\outputa_1,\outputa_2) &= 
\int\dd\alpha \dd{\beta_1}\dd{\beta_2}
p_\alice(\outputa_1|\alpha,\beta_1)p_\alice(\outputa_2|\alpha,\beta_2),\\
p_2(\outputa_1,\outputa_2) &= 
\int\dd{\alpha_1}\dd{\alpha_2}\dd\beta
p_\alice(\outputa_1|\alpha_1,\beta)p_\alice(\outputa_2|\alpha_2,\beta),\\
\frac{1}{2} &= \int\dd\alpha\dd\beta p_\alice(\outputa|\alpha,\beta) \label{eq:sleeping correlations marginal}.
\end{align}
\end{subequations}
As we will show in \cref{prop:aopt}, it turns out that $p^* = 3/4$, so that the strategy of \eqref{eq:app example strat} is in fact optimal. 
It is quite likely that this result can be obtained with a more straightforward proof and in some greater generality (e.g., allowing the agent $A$ to use a different strategy during each round), but this nonetheless gives us the opportunity to see a working example of our framework at a minimal computational cost.

\newpage
\section{Tensor notation}
\label{sec:tensor notation}

In this section, we introduce the tensor notation that we will be using to present our results --- it is merely a specialized tensor network notation that is reviewed more generally in e.g.\ \cite{bridgeman_hand-waving_2017}.

\subsection{Probability tensors}


We will use probability tensors to represent conditional probability distributions. The input legs are drawn at the bottom of the boxes, while the outputs are above, allowing to think of the diagrams as a time-ordered transmission of information from bottom to top. These probability tensors can be thought of as functions from several sets, one set per leg, to the interval $[0,1]$ such that upon summation over the outputs of the top legs, one obtains $1$ to achieve the desired normalization.

\paragraph{Examples.}

A classical agent such as $\alice$ in the Correlated Sleeper task (see section \ref{sec:intro sc}) with two inputs and one output uses a conditional probability distribution that we draw as $\twostrat{\alice}{}{}{}$ which, upon evaluation, gives
\begin{equation}
\twostrat{\alice}{\outputa}{\alpha}{\beta} = p_\alice(\outputa|\alpha,\beta).
\end{equation}
The normalization can be written as
\begin{equation}
\forall \alpha,\beta \lst \sum_{\outputa'} \twostrat{\alice}{\outputa'}{\alpha}{\beta} = 1.
\end{equation}
Analogously, an outcome distribution over two outcomes, e.g.\ $p_1$ in \eqref{eq:sleeping correlations plain}, corresponds to a tensor $\atargetp{1}{}{}$ which evaluates to
\begin{equation}
\atargetp{1}{\outputa_1}{\outputa_2} = p_1(\outputa_1,\outputa_2).
\end{equation}
We will make extensive use of the tensors $\isource{\dsunif}{}$, whose output leg has domain $\{1,\dots,\ninf\}$ and which evaluates to, for all $\dsvalalpha\in\{1,\dots,\ninf\}$,
\begin{equation}
\isource{\dsunif}{\dsvalalpha} := \frac{1}{\ninf}.
\end{equation}
We will also make use of, in a certain sense, the limit case $\ninf \rightarrow \infty$, which we define as the probability tensor $\isource{\maxunif}{}$ that represents the uniform probability density over the unit interval $[0,1]$. That is, the output leg has continuous domain $[0,1]$, and $\isource{\maxunif}{\alpha} := 1$ for all $\alpha\in[0,1]$.

\paragraph{Tensor domains.}

In principle, one should always specify the domain of a tensor leg index. Here, we will leave this implicit, as it should be relatively clear from the context and is anyway typically irrelevant for the general constructions that we describe.


\subsection{Composition rules}

There are several ways to combine the above tensors together, which we clarify in this section.

\paragraph{Scalar multiplication.}

Drawing two tensors next to each others simply implies the scalar multiplication of the probabilities. For instance,
\begin{equation}
\isource{\csalpha}{\csvalalpha}\isource{\csbeta}{\csvalbeta}  = \left(\isource{\csalpha}{\csvalalpha}\right)\cdot \left(\isource{\csbeta}{\csvalbeta}\right).
\end{equation}
One can think of such disconnected tensors as representing parallel, independent processes.

\paragraph{Contractions.}

One can contract the input leg of a tensor with the output leg of another tensor, provided that they share the same domain. This contraction, indicated graphically by the corresponding connection, implies a summation or integral over the corresponding argument. In the context of an integral, we will always tacitly assume that we are dealing with a Riemann integral, allowing us to approximate e.g.\ $\isource{\maxunif}{}$ sources with limits as $\ninf \rightarrow \infty$ of $\isource{\dsunif}{}$ sources --- this will be used in particular in \cref{app:det strat sc}. For instance, the marginal constraint of \eqref{eq:marginal constraint} can be written as
\begin{equation}
\frac{1}{2} =
\centertikz{
\node[tensornode] (alice) at (0,0) {$\alice$};
\node[voidnode] (out) [above=\outcomevspace of alice] {\indexstyle{\outputa}};
\node[tensornode] (beta) [below left=15pt and -3pt of alice] {$\maxunif$};
\node[tensornode] (gamma) [below right=15pt and -3pt of alice] {$\maxunif$};
\drawleg (beta.north) -- \AnchorOneTwo{alice}{south};
\drawleg (gamma.north) -- \AnchorTwoTwo{alice}{south};
\drawoutcomeleg (alice.north) -- (out.south);
}
=
\int \dd \alpha \dd \beta
\isource{\maxunif}{\alpha} \isource{\maxunif}{\beta}\twostrat{\alice}{\outputa}{\alpha}{\beta}
=
\int \dd \alpha \dd \beta p_\alice(\outputa|\alpha,\beta).
\end{equation}
We can actually be even more compact by writing equality between tensors with open legs, which corresponds to component-wise equality. The output legs of either side of an equality have to be matched from left to right. For instance, the constraint \eqref{eq:sleeping correlations marginal} can be written as
\begin{equation}
\isource{\macrodsunif{2}}{} = 
\centertikz{
\node[tensornode] (alice) at (0,0) {$\alice$};
\node[voidnode] (out) [above=\outcomevspace of alice] {\indexstyle{}};
\node[tensornode] (beta) [below left=15pt and -3pt of alice] {$\maxunif$};
\node[tensornode] (gamma) [below right=15pt and -3pt of alice] {$\maxunif$};
\drawleg (beta.north) -- \AnchorOneTwo{alice}{south};
\drawleg (gamma.north) -- \AnchorTwoTwo{alice}{south};
\drawoutcomeleg (alice.north) -- (out.south);
}.
\end{equation}
We will occasionally draw tensor contractions using a dashed leg such as $\centertikz{\drawdashedleg (0,0) -- (20pt,0);}$ to better distinguish overlapping legs.

\subsection{Special tensors}

\label{sec:special tensors}

\paragraph{Deterministic tensors.}

A deterministic tensor is defined as a probability tensor which, upon evaluation of all input and output legs, yields either 0 or 1. They are represented by double-edged boxes, e.g.\ $\scdetstrat$. Such deterministic tensors are in one-to-one correspondence with functions from the joint values of the input legs to the joint values of the output legs, e.g., if $\alice$ uses a deterministic strategy $\scdetstrat$ in the Correlated Sleeper task, then there must exist a function $f : [0,1] \times [0,1] \to \{1,2\}$ such that, for all $\outputa \in \{1,2\}$, for all $\alpha,\beta\in[0,1]$,
\begin{equation}
\scdetstratargs = \ddelta{\outputa}{f(\alpha,\beta)}.
\end{equation}
We will prefer using the deterministic probability tensors over the functions such as $f$ in this work.

\paragraph{Marginal node.}

Another useful node is the marginal node, $\centertikz{\node[margnode] {  };}$, which takes in arbitrary inputs, has no outputs, and always evaluates to $1$. This implies, through the contraction rule, that placing this node on an output leg amounts to marginalizing over this leg:
\begin{equation}
\centertikz{
\node[tensornode] (p) {$p$};
\node[margnode] (a) [above left=11pt and 0pt of p] { };
\node[voidnode] (b) [above=\outcomevspace of p] {\indexstyle{\outputb}};
\node[voidnode] (c) [above right=\outcomevspace and \outcomehspacetargetp of p] {\indexstyle{\outputc}};
\drawoutcomeleg \AnchorOneThree{p}{north} -- (a.south);
\drawoutcomeleg \AnchorTwoThree{p}{north} -- (b.south);
\drawoutcomeleg \AnchorThreeThree{p}{north} -- (c.south);
}
=
\sum_\outputa 
\centertikz{
\node[margnode] (m) {};
\node[voidnode] (a) [below=\outcomevspace of m] {\indexstyle{\outputa}};
\drawleg (a.north) -- (m);
}
\targetp{\outputa}{\outputb}{\outputc}
=
\sum_\outputa \targetp{\outputa}{\outputb}{\outputc}.
\end{equation}

\paragraph{Fanout nodes.}

Since we deal with the transmission of classical information, there is a special tensor which we will use quite often, namely the \emph{fanout node}. It has one input and arbitrarily many outputs, all within the same domain, and gives probability one if and only if each output is equal to the input. For instance,
\begin{align}
\centertikz{
\node[voidnode] (in) at (0,0) {\indexstyle{\csvalalpha_0}};
\node[copynode] (copy) [above=10pt of in] { };
\node[voidnode] (out1) [above left=15pt and 5pt of copy] {\indexstyle{\csvalalpha_1}};
\node[voidnode] (out2) [above=15pt of copy] {\indexstyle{\csvalalpha_2}};
\node[voidnode] (out3) [above right=15pt and 5pt of copy] {\indexstyle{\csvalalpha_3}};
\drawleg (in.north) -- (copy);
\drawleg (copy) -- (out1.south);
\drawleg (copy) -- (out2.south);
\drawleg (copy) -- (out3.south);
}
&:=
\ \ddelta{\csvalalpha_1}{\csvalalpha_0}\ddelta{\csvalalpha_2}{\csvalalpha_0}\ddelta{\csvalalpha_3}{\csvalalpha_0}.
\end{align}
Here, $\delta$ denotes either a Dirac delta functional in the physicist's notation or a Kronecker delta tensor, depending on whether the tensor leg domain is in the integers or in the reals, which are the main two options here.

\paragraph{Bundle nodes.}

It will be useful to think of special types of legs which represent tuples of legs. For instance, suppose an agent $\alice$ receives a number of inputs $\csvalalpha_1, \dots,\csvalalpha_n$ from sources $\csalpha_1,\dots\csalpha_n$:
\begin{equation}
\label{eq:tempreprmanyins}
\centertikz{
\node[tensornode] (alice) at (0,0) {$\alice$};
\node[voidnode] (out) [above=\outcomevspace of alice] { };
\node[voidnode] (small int) [below=\smalldotsvspace of alice] {\indexstyle{\dots}};
\node[voidnode] (dots) [below=18pt of alice] {$\dots$};
\node[tensornode] (in1) [left=\bigdotshspace of dots] {$\csalpha_1$};
\node[tensornode] (inn) [right=\bigdotshspace of dots] {$\csalpha_n$};
\drawleg (in1.north) -- \AnchorOneTwo{alice}{south};
\drawleg (inn.north) -- \AnchorTwoTwo{alice}{south};
\drawoutcomeleg (alice.north) -- (out.south);
}.
\end{equation}
It will be convenient to have a prescription for the notation in this sort of situation. We can achieve this by denoting
\begin{equation}
\centertikz{
\node[voidnode] (dots) at (0,0) {$\dots$};
\node[tensornode] (in1) [left=\bigdotshspace of dots] {$\csalpha_1$};
\node[tensornode] (inn) [right=\bigdotshspace of dots] {$\csalpha_n$};
\node[voidnode] (out1) [above=\outcomevspace of in1] {\indexstyle{\csvalalpha_1}};
\node[voidnode] (outn) [above=\outcomevspace of inn] {\indexstyle{\csvalalpha_n}};
\node[voidnode] (sdots) [above=10pt of dots] {\indexstyle{\dots}};
\drawleg (in1) -- (out1);
\drawleg (inn) -- (outn);
}
=
\centertikz{
\node[tuplenode] (tuple) at (0,0) { };
\node[voidnode] (dots) [below=18pt of tuple] {$\dots$};
\node[tensornode] (in1) [left=\bigdotshspace of dots] {$\csalpha_1$};
\node[tensornode] (inn) [right=\bigdotshspace of dots] {$\csalpha_n$};
\node[voidnode] (sdots) [below=\smalldotsvspace of tuple] {\indexstyle{\dots}};
\node[voidnode] (outs) [above=10pt of tuple] {\indexstyle{\vec\csvalalpha}};
\drawleg (in1.north) -- \AnchorOneTwo{tuple}{south};
\drawleg (inn.north) -- \AnchorTwoTwo{tuple}{south};
\drawtupleleg (tuple.north) -- (outs.south);
},
\end{equation}
where the leg style $\centertikz{\drawtupleleg (0,0) -- (0,10pt);}$ indicates a tuple of values normally carried by $\centertikz{\drawleg (0,0) -- (0,10pt);}$ legs, and where $\vec\csvalalpha = (\csvalalpha_1,\dots,\csvalalpha_n)$. The special \emph{bundle node} $\centertikz{\node[tuplenode] at (0,0) { };}$ is responsible of bundling all its input legs into the outgoing tuple of values.
Formally speaking, we can define this $\centertikz{\node[tuplenode] at (0,0) { };}$ tensor as
\begin{equation}
\centertikz{
\node[tuplenode] (tuple) at (0,0) { };
\node[voidnode] (dots) [below=18pt of tuple] {\indexstyle{\dots}};
\node[voidnode] (in1) [left=2pt of dots] {\indexstyle{\csvalalpha_1}};
\node[voidnode] (inn) [right=2pt of dots] {\indexstyle{\csvalalpha_n}};
\node[voidnode] (sdots) [below=2pt of tuple] {\indexstyle{\dots}};
\node[voidnode] (outs) [above=10pt of tuple] {\indexstyle{\vec\csvalbeta = (\csvalbeta_1,\dots,\csvalbeta_n)}};
\drawleg (in1.north) -- \AnchorOneTwo{tuple}{south};
\drawleg (inn.north) -- \AnchorTwoTwo{tuple}{south};
\drawtupleleg (tuple.north) -- (outs.south);
}
:=
\ddelta{\csvalbeta_1}{\csvalalpha_1} \dots \ddelta{\csvalbeta_n}{\csvalalpha_n}.
\end{equation}
The diagram of equation \eqref{eq:tempreprmanyins} now becomes
\begin{equation}
\centertikz{
\node[tensornode] (alice) at (0,0) {$\alice$};
\node[voidnode] (out) [above=\outcomevspace of alice] { };
\node[voidnode] (small int) [below=\smalldotsvspace of alice] {\indexstyle{\dots}};
\node[voidnode] (dots) [below=18pt of alice] {$\dots$};
\node[tensornode] (in1) [left=\bigdotshspace of dots] {$\csalpha_1$};
\node[tensornode] (inn) [right=\bigdotshspace of dots] {$\csalpha_n$};
\drawleg (in1.north) -- \AnchorOneTwo{alice}{south};
\drawleg (inn.north) -- \AnchorTwoTwo{alice}{south};
\drawoutcomeleg (alice.north) -- (out.south);
}
=
\centertikz{
\node[tuplenode] (tuple) at (0,0) { };
\node[voidnode] (dots) [below=18pt of tuple] {$\dots$};
\node[tensornode] (in1) [left=\bigdotshspace of dots] {$\csalpha_1$};
\node[tensornode] (inn) [right=\bigdotshspace of dots] {$\csalpha_n$};
\node[voidnode] (sdots) [below=2pt of tuple] {\indexstyle{\dots}};
\node[tensornode] (alice) [above=10pt of tuple] {$\alice$};
\node[voidnode] (a) [above=\outcomevspace of alice] { };
\drawleg (in1.north) -- \AnchorOneTwo{tuple}{south};
\drawleg (inn.north) -- \AnchorTwoTwo{tuple}{south};
\drawtupleleg (tuple.north) -- (outs.south);
\drawoutcomeleg (alice.north) -- (a.south);
}.
\end{equation}

\paragraph{Selector nodes.}

The above construction allows to capture conditional tensor contractions once we introduce the \emph{selector node}. The selector node, $\centertikz{\node[selectnode] { };}$, takes a tuple of legs as its bottom input, and receive another input, a discrete one, call it $\dsvalalpha$, on the side. The output is then the $\dsvalalpha$-th component of the input tuple. Formally, this reads
\begin{equation}
\centertikz{
\node[selectnode] (s) at (0,0) { };
\node[voidnode] (y) [above=15pt of s] {\indexstyle{\csvalbeta}};
\node[voidnode] (vecx) [below=15pt of s] {\indexstyle{\vec\csvalalpha=(\csvalalpha_1,\dots,\csvalalpha_n)}};
\node[voidnode] (i) [left=15pt of s] {\indexstyle{\dsvalalpha}};
\drawleg (s.north) -- (y.south);
\drawtupleleg (vecx.north) -- (s.south);
\drawoutcomeleg (i.east) -- (s.west);
}
:=
\ddelta{\csvalbeta}{\csvalalpha_\dsvalalpha}.
\end{equation}
One can make this even more complete by allowing the selection to pick an ordered subset of $k \leq n$ of the input legs, so that one can write:
\begin{equation}
\centertikz{
\node[selectnode] (s) at (0,0) { };
\node[voidnode] (vecy) [above=20pt of s] {\indexstyle{\vec\csvalbeta=(\csvalbeta_1,\dots,\csvalbeta_k)}};
\node[voidnode] (vecx) [below=20pt of s] {\indexstyle{\vec\csvalalpha=(\csvalalpha_1,\dots,\csvalalpha_n)}};
\node[voidnode] (veci) [left=20pt of s] {\indexstyle{\vec\dsvalalpha=(\dsvalalpha_1,\dots,\dsvalalpha_k)}};
\drawtupleleg (s.north) -- (vecy.south);
\drawtupleleg (vecx.north) -- (s.south);
\drawtupleoutcomeleg (veci.east) -- (s.west);
}
:=
\ddelta{\csvalbeta_1}{\csvalalpha_{\dsvalalpha_1}}\dots \ddelta{\csvalbeta_k}{\csvalalpha_{\dsvalalpha_k}}.
\end{equation}

\subsection{Postselection}
\label{sec:post-selection}

We will make use of postselection over tensors with $\{\false,\true\}$-valued output. We will always postselect on the output $\true$.
For instance, the tensor
$
\centertikz{
\node[tensornode] (f) {$\psname$};
\node[voidnode] (o) [above=\outcomevspace of f] { };
\node[voidnode] (i) [below=\outcomevspace of f] { };
\drawoutcomeleg (f.north) -- (o.south);
\drawleg (i.north) -- (f.south);
}
$
induces the postselection:
\begin{equation}
\centertikz{
\node[psnode] (f) {$\psname$};
\node[tensornode] (x) [below left=28pt and 3pt of f] {$\csalpha$};
\node[voidnode] (o) [left=25pt of f] {\indexstyle{\csvalalpha}};
\node[copynode] (copy) [above=\outcomevspace of x] { };
\drawleg (copy) -- (f.south);
\drawleg (x.north) -- (copy);
\drawleg (copy) -- (o.south);
}
:=
\left(
\centertikz{
\node[tensornode] (f) {$\psname$};
\node[voidnode] (of) [above=\outcomevspace of f] {\indexstyle{\true}};
\node[tensornode] (x) [below left=20pt and 3pt of f] {$\csalpha$};
\node[margnode] (o) [left=25pt of f] { };
\node[copynode] (copy) [above=\outcomevspace of x] { };
\drawoutcomeleg (f.north) -- (of.south);
\drawleg (copy) -- (f.south);
\drawleg (x.north) -- (copy);
\drawleg (copy) -- (o.south);
}
\right)^{-1}
\cdot
\centertikz{
\node[tensornode] (f) {$\psname$};
\node[voidnode] (of) [above=\outcomevspace of f] {\indexstyle{\true}};
\node[tensornode] (x) [below left=20pt and 3pt of f] {$\csalpha$};
\node[voidnode] (o) [left=25pt of f] {\indexstyle{\csvalalpha}};
\node[copynode] (copy) [above=\outcomevspace of x] { };
\drawoutcomeleg (f.north) -- (of.south);
\drawleg (copy) -- (f.south);
\drawleg (x.north) -- (copy);
\drawleg (copy) -- (o.south);
},
\end{equation}
assuming that the postselection has a chance of succeeding, that is, assuming that
\begin{equation}
\centertikz{
\node[tensornode] (f) {$\psname$};
\node[voidnode] (of) [above=\outcomevspace of f] {\indexstyle{\true}};
\node[tensornode] (x) [below left=20pt and 3pt of f] {$\csalpha$};
\node[margnode] (o) [left=25pt of f] { };
\node[copynode] (copy) [above=\outcomevspace of x] { };
\drawoutcomeleg (f.north) -- (of.south);
\drawleg (copy) -- (f.south);
\drawleg (x.north) -- (copy);
\drawleg (copy) -- (o.south);
}
\neq 
0.
\end{equation}

%
%

\subsection{Correlated Sleeper in tensor notation}

We may now rewrite the optimization problem of \eqref{eq:sleeping correlations plain} in tensor notation. Notice that the non-negativity and normalization constraint of equation \eqref{eq:sleeping correlations norm} are now omitted because they are implied by $\twostrat{\alice}{}{}{}$ being denoted as a probability tensor. This yields the compact form:
\begin{subequations}
\label{eq:sleeping correlations tensor}
\begin{align}
\aopt\ =\ &\sup_{\scalebox{0.9}{$\twostrat{\alice}{}{}{}$}} \ 
\ \frac{1}{2}\sum_{\outputa\in\{1,2\}} \left(\atargetp{1}{\outputa}{\outputa} + \atargetp{2}{\outputa}{\outputa}\right) \hspace{4cm}
\end{align}
%
%
\begin{equation}
\textup{s.t. } \atargetp{1}{}{} = \acaseone{\alice}{}{}{0},\quad 
\atargetp{2}{}{} = \acasetwo{\alice}{}{}{0},\quad
\isource{\macrodsunif{2}}{} =
\amargconstraint{\alice}{}{0}.
\end{equation}
\end{subequations}
Let us in fact take the opportunity to further simplify this problem with the following proposition, which allows us to restrict the optimization to deterministic strategies for $\alice$. The proof is given in \app~\ref{app:det strat sc}; it primarily relies on the Riemann integrability assumption over $\alice$'s strategy.
\begin{restatable}{prop}{PropDetApp}
\label{prop:det app}
It holds that one can restrict the optimization variable of \eqref{eq:sleeping correlations tensor}, namely, the probability tensor $\twostrat{\alice}{}{}{}$, to range over the \emph{deterministic} probability tensors only:
\begin{subequations}
\label{eq:sc det tensor}
\begin{align}
\label{eq:sc det tensor obj}
\aopt\ =\ &\sup_{\scalebox{0.9}{$\scdetstrat$}} \ 
\ \frac{1}{2}\sum_{\outputa\in\{1,2\}} \left(\atargetp{1}{\outputa}{\outputa} + \atargetp{2}{\outputa}{\outputa}\right) \hspace{4cm}
\end{align}
%
%
\begin{equation}
\label{eq:sc det tensor param}
\textup{s.t. } \atargetp{1}{}{} = \acaseone{\alice}{}{}{1},\quad 
\atargetp{2}{}{} = \acasetwo{\alice}{}{}{1},\quad
\isource{\macrodsunif{2}}{} =
\amargconstraint{\alice}{}{1}.
\end{equation}
\end{subequations}
\end{restatable}

\newpage
\section{Causal compatibility}
\label{sec:causal compat}

We now turn to the problem of causal compatibility.
It is worth mentioning that in the work of \cite{navascues_inflation_2020}, one can find an excellent introduction to the notion of causal unpacking, which describes the tools that one can use to translate the problem of causal compatibility with a causal structure into a related problem of causal compatibility with another \emph{simpler} causal structure.
In the case of a classical causal structure, where all the nodes (agent, sources etc.) have an associated probability tensor, one can actually unpack this causal structure (featuring e.g.\ direct causal influence between observed agents, several layers of unobserved sources interconnected in arbitrary ways, measurement settings à la CHSH, etc.) into a bi-layer causal structure with no inputs.
However, we will be considering multi-network scenarios with some agents using the same strategies, and it is now unclear whether bilayer structures are most general in this extended case. Let us nonetheless restrict our attention to these cases, since this is an interesting generalization of the work of \cite{navascues_inflation_2020}.

\subsection{Network scenarios}
\label{sec:networks}

\paragraph{Single-network scenarios.}
The most general causal structure that we shall consider will be called a network. The network consists of several ingredients. There are three integer parameters: $\pcount$ labels the number of strategies that may be used by the agents, $\npcount$ labels the number of agents in the network (we can assume in the case of a single-network scenario that $\pcount \leq \npcount$), and $\scount$ labels the number of sources (sometimes called latent nodes in the literature) that exist in the network. There are now two maps to specify: one is the strategy assignment map, $\pmap : \{1,\dots,\npcount\} \to \{1,\dots,\pcount\}$, which says that the agent $\npindex \in \{1,\dots,\npcount\}$ must use the strategy $\pmap(\npindex) \in \{ 1,\dots,\pcount\}$. The other is the connectivity map,
\begin{equation}
\cmap  : \{1,\dots,\npcount\} \to \powerset{\{1,\dots,\scount\}} := \{ (1), (2,3), (3,2), (1,4,5), (1,2,\scount-2,\scount),\dots\},
\end{equation}
where $\powerset{\{1,\dots,\scount\}}$ is the set of all sequences coming from the subsets of $\{1,\dots,\scount\}$. This map $\cmap$ specifies that the agent $\npindex$ receives the sources $\cmap(\npindex)$, in this order, as inputs to their strategy. We implicitly assume, for consistency, that whenever two agents $\npindex \neq \npindex'$ are using the same strategy $\pmap(\npindex) = \pmap(\npindex')$, it must be that $\cmap(\npindex)$ and $\cmap(\npindex')$ are sequences of equal length, since the strategy $\pmap(\npindex)$ has a well-defined number of inputs (occasionally denoted ``$\nin_{\pmap(\npindex)}$'').

For instance, consider the bilocal network (also known as the ``three-on-a-line'' network), represented graphically in \cref{fig:bilocal net}, where $\npcount = 3$ agents share $\scount = 2$ sources, so that the agent $\npindex = 2$ has access to the two sources $\sindex = 1,2$, while the agent $\npindex = 1$ has only access to the source $\sindex = 1$ and the agent $\npindex = 3$ has access to the $\sindex = 2$ source. If the three agents are allowed to use $\pcount = 3$ arbitrary strategies, then this network will be specified as
\begin{subequations}
\label{eq:bilocal topo}
\begin{align}
\bilocaltopo &= (\pcount = 3, \npcount = 3, \scount = 2, \pmap, \cmap), \\
\pmap(1) &= 1, \quad \pmap(2) = 2, \quad \pmap(3) = 3, &\netnote{3 different strategies} \\
\cmap(1) &= (1), \quad \cmap(2) = (1,2), \quad \cmap(3) = (2). &\netnote{connectivity of the bilocal network}
\end{align}
\end{subequations}
\begin{figure}[h!]
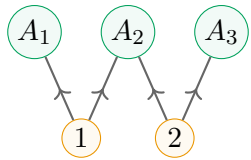

\centering
\begin{equation*}
\centertikz{
\node[agent] (a) at (0,0) {$A_1$};
\node[agent] (b) at (35pt,0) {$A_2$};
\node[agent] (c) at (70pt,0) {$A_3$};
\node[source] (alpha) at (17.5pt,-40pt) {1};
\node[source] (beta) at (52.5pt,-40pt) {2};
\drawabsleg (alpha) -- (a);
\drawabsleg (alpha) -- (b);
\drawabsleg (beta) -- (b);
\drawabsleg (beta) -- (c);
}
\end{equation*}
\caption{The graph corresponding to the $\bilocaltopo$ single-network scenario. The labeling of the sources and of the different agent strategies is chosen to match the parametrization of \eqref{eq:bilocal topo} --- more conventionally, one would prefer to think of $A_1$ as Alice, $A_2$ as Bob and $A_3$ as Charlie.}
\label{fig:bilocal net}
\end{figure}

\paragraph{Multi-network scenarios.}

We will also be interested in cases where several networks are involved, with the set of available strategies being globally shared across these networks: such scenarios will be called \emph{multi-network scenarios}.
While a multi-network scenario can always be embedded into a single-network scenario whose associated graph features several disconnected components, the inflation framework is most conveniently applied to the multi-network scenario formulation.

To fix the notation, consider a number $\topocount$ of networks. Each network $\topoindex \in \{1,\dots,\topocount\}$ will have its own number of agents $\npcount_\topoindex$, number of sources $\scount_\topoindex$, strategy assignment map $\pmap_\topoindex : \{1,\dots,\npcount_\topoindex\} \to \{1,\dots,\pcount\}$, and connectivity map $\cmap_\topoindex : \{1,\dots,\npcount_\topoindex\} \to \powerset{\{1,\dots,\scount_\topoindex\}}$, but the number of strategy $\pcount$ does not depend on $\topoindex$. 
There is a consistency condition that is implicitly assumed: for all $\topoindex, \topoindex' \in \{1,\dots,\topocount\}$, for any pair of agents $\npindex \in \{1,\dots,\npcount_\topoindex\}$ and $\npindex'\in\{1,\dots,\npcount_{\topoindex'}\}$ that use the same strategy, i.e., $\pmap_\topoindex(\npindex) = \pmap_{\topoindex'}(\npindex')$, it must be that these two agents receive the same number of inputs, i.e., $\cmap_\topoindex(\npindex)$ and $\cmap_{\topoindex'}(\npindex')$ must be sequences of the same length, since the agent strategy $\pmap_\topoindex(\npindex)$ has a well-defined number of inputs $\nin_{\pmap_\topoindex(\npindex)}$.
Such a multi-network scenario will be denoted as $\genmultopo$ (the sequence notation is to emphasize the fact that we pick a specific ordering of the networks), or more explicitly as
\begin{equation}
\label{eq:def genmultopoexpl}
\genmultopoexpl.
\end{equation}

\paragraph{The Correlated Sleeper's multi-network scenario.}

In the case of the Correlated Sleeper, there are three relevant networks, whose graphs are represented in \cref{fig:sc configs}: the first network is the one where the two instances of $\alice$ are connected through the left input, which corresponds to the network 
\begin{subequations}
\label{eq:sc config 1}
\begin{align}
\scnetone &= (\pcount = 1,\npcount_1 = 2,\scount_1 = 3,\pmap_1,\cmap_1), &\netnote{one strategy, two agents, three sources}\\
\pmap_1(1) &= 1, \quad \pmap_1(2) = 1, &\netnote{agents use same strategy}\\
\cmap_1(1) &= (1,2), \quad \cmap_1(2) = (1,3). &\netnote{agents' first inputs connected to first source}
\end{align}
\end{subequations}
The second network is the one where the two $\alice$'s are connected through the right input, i.e.,
\begin{subequations}
\label{eq:sc config 2}
\begin{align}
\scnettwo &= (\pcount = 1,\npcount_2 = 2,\scount_2 = 3,\pmap_2,\cmap_2), \\
\pmap_2(1) &= 1, \quad  \pmap_2(2) = 1, \\
\cmap_2(1) &= (1,3), \quad \cmap_2(2) = (2,3). &\netnote{agents' second inputs connected to third source}
\end{align}
\end{subequations}
The last network is the one that allows us to express the marginal constraint, where we only look at one isolated agent $\alice$:
\begin{subequations}
\label{eq:sc config 3}
\begin{align}
\scnetthree &= (\pcount = 1,\npcount_3 = 1,\scount_3 = 2,\pmap_3,\cmap_3), \\
\pmap_3(1) &= 1, \\
\cmap_3(1) &= (1,2). &\netnote{the agent inputs are two i.i.d.\ sources}
\end{align}
\end{subequations}

\begin{figure}[h!]
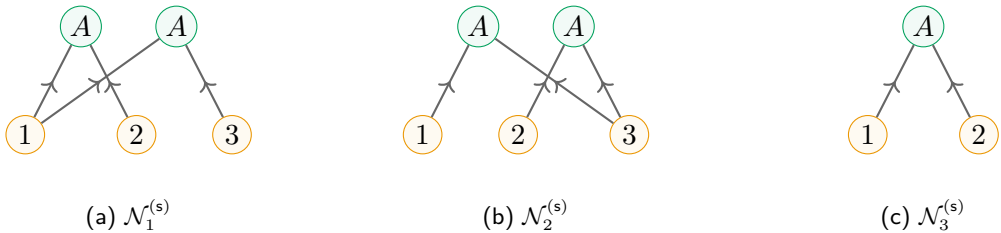

\centering
	\begin{subfigure}{0.3\textwidth}
	\centering
	\begin{equation*}
	\centertikz{
	\node[agent] (a1) {$A$};
	\node[agent] (a2) [right=20pt of a1] {$A$};
	\node[source] (alpha) [below left=30pt and 10pt of a1] {$1$};
	\drawabsleg (alpha) -- (a1);
	\drawabsleg (alpha) -- (a2);
	\node[source] (beta1) [below right=30pt and 10pt of a1] {$2$};
	\drawabsleg (beta1) -- (a1);
	\node[source] (beta2) [below right=30pt and 10pt of a2] {$3$};
	\drawabsleg (beta2) -- (a2); 
	}
	\end{equation*}
	\caption{$\scnetone$}
	\end{subfigure}
	\begin{subfigure}{0.3\textwidth}
	\centering
	\begin{equation*}
	\centertikz{
	\node[agent] (a1) {$A$};
	\node[agent] (a2) [right=20pt of a1] {$A$};
	\node[source] (alpha1) [below left=30pt and 10pt of a1] {$1$};
	\drawabsleg (alpha1) -- (a1);
	\node[source] (alpha2) [below left=30pt and 10pt of a2] {$2$};
	\drawabsleg (alpha2) -- (a2);
	\node[source] (beta) [below right=30pt and 10pt of a2] {$3$};
	\drawabsleg (beta) -- (a1);
	\drawabsleg (beta) -- (a2);
	}
	\end{equation*}
	\caption{$\scnettwo$}
	\end{subfigure}
	\begin{subfigure}{0.3\textwidth}
	\centering
	\begin{equation*}
	\centertikz{
	\node[agent] (a) {$A$};
	\node[source] (alpha) [below left=30pt and 10pt of a] {$1$};
	\node[source] (beta) [below right=30pt and 10pt of a] {$2$};
	\drawabsleg (alpha) -- (a);
	\drawabsleg (beta) -- (a);
	}
	\end{equation*}
	\caption{$\scnetthree$}
	\end{subfigure}
\caption{The Correlated Sleeper's multi-network scenario.}
\label{fig:sc configs}
\end{figure}

\subsection{Causal compatibility}

\paragraph{A note on deterministic strategies.}

Our framework deals best with deterministic agent strategies as basic primitives. This may sound restrictive, but fundamentally speaking, it is not: any non-deterministic strategy can be achieved with a deterministic strategy upon giving each agent access to an additional local randomness source, which can be captured by an appropriate update of the network scenario.
From a computational perspective, this explicit addition of additional sources can turn out to be costly --- we will return to this aspect in \cref{sec:remarks about implementation}. However, this addition of local sources is not always necessary: depending on the sources shared between the agents, local randomness can sometimes be extracted without adding additional local sources. This is for instance demonstrated in \cref{prop:det app} for the Correlated Sleeper. In the following bilocal network example, although we do not give an explicit construction, it is also the case that local randomness can be extracted from the shared sources.

\paragraph{Causal compatibility: example.}
We first present causal compatibility with the bilocal network before generalizing to arbitrary networks.
A probability tensor $\targetp{}{}{}$, which we will typically call an outcome distribution, is compatible with the bilocal network $\bilocaltopo$ (equation \eqref{eq:bilocal topo}), which we denote  as
\begin{equation}
\label{eq:outputdistribs bilocaltopo}
\targetp{}{}{} \in \outputdistribs{\bilocaltopo},
\end{equation}
if there exist deterministic probability tensors
\begin{subequations}
\label{eq:example bilocal}
\begin{equation}
\detonestrat{\alice}{}{}, \dettwostrat{\bob}{}{}{}, \detonestrat{\charlie}{}{},
\end{equation}
such that for all $\outputa, \outputb, \outputc$,
\begin{equation}
\bilocalnetwork{\csalpha_1}{\csalpha_2}{\alice}{\bob}{\charlie}{\outputa}{\outputb}{\outputc}
=
\targetp{\outputa}{\outputb}{\outputc}
.
\end{equation}
\end{subequations}

\paragraph{Causal compatibility: single-network scenarios.}

More generally, a probability tensor $\gtargetp$ is compatible with a single-network scenario $\gentopo = \gentopoexpl$, denoted
\begin{equation}
\gtargetp \in \outputdistribs{\gentopo},
\end{equation}
if there exist deterministic probability tensors
\begin{subequations}
\label{eq:gen causal compat}
\begin{equation}
\label{eq:gen causal compat prim}
\left\{\gonestratdet{\alice_\pindex}{}{}\right\}_{\pindex=1}^{\pcount}
\end{equation}
%
such that for all ${\outputa_1}, {\outputa_2}, \dots, {\outputa_\npcount}$:
\begin{equation}
\gnetwork{\outputa_1}{\outputa_2}{\outputa_\npcount}{0}
=
\gtargetpargs{\outputa_1}{\outputa_2}{\outputa_\npcount}.
\end{equation}
\end{subequations}
%
The output domains of the strategies are some finite subsets of the integers (this is anyway the only possibility from a computational perspective) that are also left implicit here.
Recall the special selector tensor $\centertikz{\node[selectnode] { };}$, allowing us to parametrize the sources that each agent has access to, and the bundle tensor $\centertikz{\node[tuplenode] {};}$, collecting the source outputs in a single vector, both introduced more precisely in section \ref{sec:special tensors}.

\paragraph{Causal compatibility: multi-network scenarios.}
\label{sec:causal compat multiple configs}

%
We now introduce causal compatibility for multi-network scenarios: this will be the problem formulation that we shall use in the rest of this work. In fact, the rest of this work will be concerned with computationally tractable supersets of the following $\outputdistribs{\genmultopo}$.

\begin{definition}
\label{def:causal compat}
Consider a multi-network scenario
\begin{equation}
\label{eq:genmultopoexpl}
\genmultopoexpl
\end{equation}
(this notation is explained in \cref{sec:networks}).
We say that a sequence of outcome distributions $\genmultopotargetps$, where the outcome distribution $\macrogtargetp{\targetpname_\topoindex}$ must have $\npcount_\topoindex$ output legs,\footnote{Strictly speaking, we should also specify the domain of the outputs of each agent (i.e., how many different outcomes they may output) to have a well-defined set $\outputdistribs{\genmultopo}$, but we leave this dependence implicit.}
is causally compatible with the multi-network scenario $\genmultopo$, denoted
\begin{equation}
\label{eq:output distribs multiple configs}
\genmultopotargetps \in \outputdistribs{\genmultopo},
\end{equation}
if there exist deterministic probability tensors
\begin{subequations}
\label{eq:causal compat multiple config}
\begin{equation}
\left\{\gonestratdet{\alice_\pindex}{}{}\right\}_{\pindex=1}^{\pcount}
\end{equation}
such that, for all $\topoindex=1,\dots,\topocount$,
\begin{equation}
\label{eq:causal compat multiple config condition}
\gnetwork{}{}{}{1}
=
\macrogtargetp{\targetpname_\topoindex}.
\end{equation}
\end{subequations}
\end{definition}

In equation \eqref{eq:causal compat multiple config condition}, we used the compact notation where equality of two tensors with open legs simply corresponds to component-wise equality. 
Furthermore, we use ``$\dots$'' in the same sense that $\{1,\dots,4\} = \{1,2,3,4\}$, but with the sources, the number of terms omitted is not very explicit, so we write `` $\overset{(\scount_\topoindex - 3)}{\dots}$ '' to indicate that we omitted $\scount_\topoindex - 3$ sources and drew the 3 remaining one explicitly. 
Note that the tensors $\centertikz{\node[selectnode] { };}$ and $\centertikz{\node[tuplenode] {};}$ were introduced in section \ref{sec:special tensors}.

\paragraph{Causal compatibility: Correlated Sleeper.}

In this notation, thanks to \cref{prop:det app} and the multi-network scenario of equations \eqref{eq:sc config 1}-\eqref{eq:sc config 3}, we can rewrite the feasible region of the optimization problem \eqref{eq:sc det tensor} as a causal compatibility problem:
\begin{subequations}
\label{eq:sc causal compat formulation}
\begin{align}
p^*\ =\ &\sup_{\scalebox{0.8}{$\atargetp{1}{}{},\atargetp{2}{}{}$}} \ 
\ \frac{1}{2}\sum_{\outputa\in\{1,2\}} \left(\atargetp{1}{\outputa}{\outputa} + \atargetp{2}{\outputa}{\outputa}\right) \\
&\textup{s.t. }
\left(\atargetp{1}{}{},\atargetp{2}{}{},\isource{\macrodsunif{2}}{}\right)
\in\outputdistribs{\scnetone,\scnettwo,\scnetthree}. \label{eq:sc causal compat formulation 2}
\end{align}
\end{subequations}
We emphasize the fact that the marginal constraint induced from \eqref{eq:marginal constraint} on $\atargetp{1}{}{}$ and $\atargetp{2}{}{}$ is indeed contained in the above equation \eqref{eq:sc causal compat formulation 2}.

\subsection{Further generalizations}

\paragraph{Complex source behavior.}

We are currently allowing all the sources to be ``maximally entropic'', so that any other source distribution can be obtained by the agents upon applying the relevant postprocessing. One thing that our framework can not deal with (currently, at least --- it is unclear whether this can be nicely incorporated in) is the possibility to constrain the sources to a specific type of distribution. This could either be a network-wide constraint, e.g.\ restrict all sources to be uniform over a fixed number of values, or context-dependent constraints, allowing to capture e.g.\ the performance of a strategy faced with different source distributions.

\paragraph{Partial constraints, optimization.}

Our framework can relatively straightforwardly deal with partial constraints over the network probabilities, as well as optimizing polynomials of the network probabilities. These ideas and techniques are explored more systematically in \cite{navascues_inflation_2020}, and can be adapted to the present framework easily.
For instance, in the problem \eqref{eq:causal compat multiple config}, one may not know the full statistics $\macrogtargetp{\targetpname_\topoindex}$ for all $\topoindex$, but perhaps only an average value for $\topoindex=1$, the probability of one event (one tuple of outcomes) only for $\topoindex=2$, a lower bound of a certain polynomial over the probabilities of the events for $\topoindex=3$, etc.
We shall not attempt to parameterize these sorts of problems in order to remain somewhat concise,
%
but we will deal with the explicit example of the Correlated Sleeper in \cref{sec:applications}.

\newpage
\section{Postselected inflation: motivation}
\label{sec:outer approx}

In this section, we give an intuitive motivation for the postselected inflation outer approximations in the context of a single-network scenario, namely, the bilocal network $\bilocaltopo$. In the next \cref{sec:post-selected inflation}, we will state general proofs of soundness of this approach, before explicitly applying these techniques in \cref{sec:applications}. We will return to the correspondence with the usual fanout inflation formalism in \cref{sec:fanout inflation}.

\subsection{Convexification of the causal compatibility problem}
\label{sec:convexification}

Deciding the causal compatibility of a distribution $\gtargetp$ with a single-network scenario is generally hard. 
There are two reasons behind this: one is that the problem in its standard formulation as in equation \eqref{eq:gen causal compat} is not convex, in the sense that a convex combination of the solution tensors of equation \eqref{eq:gen causal compat prim} cannot be used as a new solution of the problem \eqref{eq:gen causal compat}.
In fact, allowing for convex combinations of these tensors is equivalent to sending the output of a source $\isource{\hv}{}$ to all the agents.
%
%
%
Let us make the corresponding causal compatibility problem explicit in the case of the bilocal network (this is to be compared with equations \eqref{eq:outputdistribs bilocaltopo}-\eqref{eq:example bilocal}): we let $\itempSR$ (for ``global randomness'') be the set of all distributions $\targetp{}{}{}$
for which there exist
\begin{subequations}
\begin{equation}
\dettwostrat{\alice}{}{}{}, \threedstrat{\bob}{}{}{}{}, \dettwostrat{\charlie}{}{}{}, \isource{\hv}{}
\end{equation}
such that (the style difference between the dashed and solid edges is there to guide the eye but implies the same operation of tensor contraction)
\begin{equation}
\bilocalnetworkhv{\csalpha}{\csbeta}{\alice}{\bob}{\charlie}{}{}{}{\hv} = \targetp{}{}{}.
\end{equation}
\end{subequations}
This modified causal compatibility is too permissive: it holds that $\outputdistribs{\bilocaltopo} \subsetneq \itempSR$, which is a general feature of allowing global randomness.
One needs to think of something else to obtain a causal compatibility problem that is both convex, i.e., which allows for a global randomness source $\isource{\hv}{}$, and that yields a good outer approximation of the set $\outputdistribs{\bilocaltopo}$. 
The trick is the following: define the set $\itempL$ of all distributions $\targetp{}{}{}$ for which there exist
\begin{subequations}
\label{eq:convex game}
\begin{equation}
\dettwostrat{\alice}{}{}{}, \threedstrat{\bob}{}{}{}{}, \dettwostrat{\charlie}{}{}{}, \isource{\hv}{}
\end{equation}
such that
\begin{equation}
\label{eq:doubledbilocalnetworkhv}
\doubledbilocalnetworkhv
=
\targetp{}{}{}\targetp{}{}{}.
\end{equation}
\end{subequations}
%
Indeed, it holds that $\outputdistribs{\bilocaltopo} = \itempL$.
The fact that $\outputdistribs{\bilocaltopo} \subseteq \itempL$ is trivial --- the agents may simply discard the input from $\isource{\hv}{}$.
%
For the other direction, if the strategies $\dettwostrat{\alice}{}{}{}, \threedstrat{\bob}{}{}{}{}, \dettwostrat{\charlie}{}{}{}$ are making a non-trivial use of the global randomness from $\isource{\hv}{}$, then, they will do so on both ends of the diagram, and will thus necessarily become correlated, whereas the constraint of \eqref{eq:doubledbilocalnetworkhv} imposes the two halves of the diagram to be uncorrelated. 
%
Formally speaking, this follows from our ``main lemma'', whose proof is given in \app~\ref{app:post-selected inflation proofs}, along with the definition of the relevant norms.

\begin{restatable}[Main lemma]{lemma}{LemmaMainLemma}
\label{lem:average twonorm}
For any $k\in\mathbb{N}$,
for any probability tensors 
\begin{equation}
\isource{\hv}{}, \isource{\targetpname}{} \in \mathbb{R}^k, \left\{\onestrat{\targetptildename}{}{\hvval} \in \mathbb{R}^k\right\}_{\hvval},
\end{equation}
it holds that
\begin{equation}
\int\dd\hvval \isource{\hv}{\hvval} \twonorm{\isource{\targetpname}{} - \onestrat{\targetptildename}{}{\hvval}}^2
\leq 
3\onenorm{\isource{\targetpname}{}\isource{\targetpname}{} - \cpdoubleq{}{}}.
\end{equation}
\end{restatable}

In our case, we can apply \cref{lem:average twonorm} to equation \eqref{eq:doubledbilocalnetworkhv} to read out that we must have
\begin{equation}
\int\dd\hvval \isource{\hv}{\hvval} \twonorm{\targetp{}{}{} - 
\centertikz{
\node[detnode] (alice) at (-1,0) {$\alice$};
\node[detnode] (bob) at (0,0) {$\bob$};
\node[detnode] (charlie) at (1,0) {$\charlie$};
\node[voidnode] (a) [above=\outcomevspace of alice] {\indexstyle{}};
\node[voidnode] (b) [above=\outcomevspace of bob] {\indexstyle{}};
\node[voidnode] (c) [above=\outcomevspace of charlie] {\indexstyle{}};
\node[copynode] (alphacopy) at (-0.52,-30pt) { };
\node[tensornode] (alpha) [below=\outcomevspace of alphacopy] {$\maxunif$};
\node[copynode] (betacopy) at (0.52,-30pt) { };
\node[tensornode] (beta) [below=\outcomevspace of betacopy] {$\maxunif$};
\drawleg (alpha.north) -- (alphacopy);
\drawleg (alphacopy) to [out=150,in=270] ($0.75*(alice.south west) + 0.25*(alice.south east)$);
\drawleg (alphacopy) to [out=30,in=270] ($0.8*(bob.south west) + 0.2*(bob.south east)$);
\drawleg (beta.north) -- (betacopy);
\drawleg (betacopy) to [out=150,in=270] ($0.5*(bob.south west) + 0.5*(bob.south east)$);
\drawleg (betacopy) to [out=30,in=270] ($0.75*(charlie.south west) + 0.25*(charlie.south east)$);
\drawoutcomeleg (alice.north) -- (a.south);
\drawoutcomeleg (bob.north) -- (b.south);
\drawoutcomeleg (charlie.north) -- (c.south);
\node[voidnode] (hva) [below right=7pt and -4pt of alice] {\indexstyle{\hvval}};
\drawleg  (hva.north) -- \AnchorTwoTwo{alice}{south};
\node[voidnode] (hvb) [below right=7pt and -4pt of bob] {\indexstyle{\hvval}};
\drawleg  (hvb.north) -- \AnchorThreeThree{bob}{south};
\node[voidnode] (hvc) [below right=7pt and -4pt of charlie] {\indexstyle{\hvval}};
\drawleg  (hvc.north) -- \AnchorTwoTwo{charlie}{south};
}
}^2 = 0.
\end{equation}
This implies in particular that there exists a value $\hvval_0$ such that
\begin{equation}
\centertikz{
\node[detnode] (alice) at (-1,0) {$\alice$};
\node[detnode] (bob) at (0,0) {$\bob$};
\node[detnode] (charlie) at (1,0) {$\charlie$};
\node[voidnode] (a) [above=\outcomevspace of alice] {\indexstyle{}};
\node[voidnode] (b) [above=\outcomevspace of bob] {\indexstyle{}};
\node[voidnode] (c) [above=\outcomevspace of charlie] {\indexstyle{}};
\node[copynode] (alphacopy) at (-0.52,-30pt) { };
\node[tensornode] (alpha) [below=\outcomevspace of alphacopy] {$\maxunif$};
\node[copynode] (betacopy) at (0.52,-30pt) { };
\node[tensornode] (beta) [below=\outcomevspace of betacopy] {$\maxunif$};
\drawleg (alpha.north) -- (alphacopy);
\drawleg (alphacopy) to [out=150,in=270] ($0.75*(alice.south west) + 0.25*(alice.south east)$);
\drawleg (alphacopy) to [out=30,in=270] ($0.8*(bob.south west) + 0.2*(bob.south east)$);
\drawleg (beta.north) -- (betacopy);
\drawleg (betacopy) to [out=150,in=270] ($0.5*(bob.south west) + 0.5*(bob.south east)$);
\drawleg (betacopy) to [out=30,in=270] ($0.75*(charlie.south west) + 0.25*(charlie.south east)$);
\drawoutcomeleg (alice.north) -- (a.south);
\drawoutcomeleg (bob.north) -- (b.south);
\drawoutcomeleg (charlie.north) -- (c.south);
\node[voidnode] (hva) [below right=7pt and -4pt of alice] {\indexstyle{\hvval_0}};
\drawleg  (hva.north) -- \AnchorTwoTwo{alice}{south};
\node[voidnode] (hvb) [below right=7pt and -4pt of bob] {\indexstyle{\hvval_0}};
\drawleg  (hvb.north) -- \AnchorThreeThree{bob}{south};
\node[voidnode] (hvc) [below right=7pt and -4pt of charlie] {\indexstyle{\hvval_0}};
\drawleg  (hvc.north) -- \AnchorTwoTwo{charlie}{south};
}
= \targetp{}{}{},
\end{equation} 
which proves that $\itempL \subseteq \outputdistribs{\bilocaltopo}$.

Now, it is important to notice how we have effectively turned the causal compatibility problem into a convex one.
For each output of the source $\isource{\hv}{}$, the agents of \eqref{eq:convex game} are using some strategies in the original network scenario. Thus, solving for the problem \eqref{eq:convex game} is equivalent to optimizing the distribution $\isource{\hv}{}$ whose output domain is the set of possible tuples of strategies, and the only constraint in place is that of \eqref{eq:doubledbilocalnetworkhv}, which is linear in $\isource{\hv}{}$ for fixed $\targetp{}{}{}$.\footnote{Note that in \eqref{eq:doubledbilocalnetworkhv}, one may assume without loss of generality that $\hv$ outputs at most $\sim k^2$ distinct values, given $k$ outcomes for $\targetp{}{}{}$ --- this is the content of Carathéodory's theorem, see e.g.\ theorem 4.3.2 in \cite{ConvexAnalysis}.}
However, this convexification of the original problem is trivial so far: one has to enumerate all possible tuples of strategies, and there are infinitely many of them due to the fact that the sources output an unbounded number of values to which the agents may react differently.
%
%
To tackle this, we move on to finding a way to restrict the output of the sources to take very few values, e.g., 2 or 3 each, while still having a chance of certifying that $\targetp{}{}{} \notin \outputdistribs{\bilocaltopo}$.

\subsection{Restricting the output cardinality of the sources}

Still focusing on the example of the bilocal network, let us now attempt to add to the problem of \eqref{eq:convex game} the constraint that the sources may only take $\ninf\in\mathbb{N}, \ninf \geq 2$ different values (potentially, say, only 2 or 3 values --- in fact, we could also use a different cardinality for each individual source, but we do not make this option explicit here for simplicity).
All sources will thus be distributed as the uniform distribution $\dsource{\dsunif}{}$ over $\ninf$ values.
We let $\itempRestricted$ be the set of all $\targetp{}{}{}$ for which there exist
\begin{subequations}
\label{eq:convex game dsource}
\begin{equation}
\twodstrat{\alice}{}{}{}, \threedstrat{\bob}{}{}{}{}, \twodstrat{\charlie}{}{}{}, \isource{\hv}{},
\end{equation}
such that 
\begin{equation}
\label{eq:convex game dsource indep}
\doubleddiscretebilocalnetworkhv = 
\targetp{}{}{}\targetp{}{}{}.
\end{equation}
\end{subequations}
This problem can be formulated without loss of generality as a linear program over $\isource{\hv}{}$, whose outputs can be taken to range over the tuple of deterministic strategies that the agents should use, and there are now finitely many such strategies.
It is apparent that the above set $\itempRestricted$ taken in the limit of arbitrarily large $\ninf$ will coincide with $\itempL$ as in equation \eqref{eq:convex game}, which by the arguments of \cref{sec:convexification} equals to $\outputdistribs{\bilocaltopo}$, so the convergence of the above scheme is under control. However, we have restricted the possibilities for the agents by going from \eqref{eq:convex game} to \eqref{eq:convex game dsource}, so it is clear that $\itempRestricted \subseteq \outputdistribs{\bilocaltopo}$:
%
we need to give the agents more possibilities to obtain outer approximations of $\outputdistribs{\bilocaltopo}$.

\subsection{Adding possibilities through postselection} 

So far, the agents cannot really use the source $\isource{\hv}{}$ because of the independence condition of \eqref{eq:convex game dsource indep} and the argument surrounding \cref{lem:average twonorm}. To bypass \cref{lem:average twonorm}, let us add additional correlations between the two halves of the left-hand side diagram of \eqref{eq:convex game dsource indep} besides the source $\isource{\hv}{}$.
More precisely, let us add correlations between the four i.i.d.\ sources $\isource{\dsunif}{}$. We could in principle denote this by a new source with four output legs, but it is more adequate to in fact add correlations in the form of postselection on the outputs of the sources $\isource{\dsunif}{}$. 
Let us for now leave the postselection arbitrary: we denote it with $\pstensor$, meaning that we postselect on the output $\true$ of a tensor $\pstensorbase{}{}{}$ which has outputs $\true$ and $\false$ --- see also \cref{sec:post-selection} for more details on postselection.
We say that a distribution $\targetp{}{}{}$ is compatible with a \findname{}, denoted for now
\begin{equation}
\targetp{}{}{} \in \inftemp,
\end{equation}
if there exist
\begin{subequations}
\label{eq:gen inf bilocal}
\begin{equation}
\label{eq:gen inf bilocal tensors}
\twodstrat{\alice}{}{}{}, \threedstrat{\bob}{}{}{}{}, \twodstrat{\charlie}{}{}{}, \isource{\hv}{},
\end{equation}
such that
\begin{equation}
\label{eq:gen inf bilocal indep}
\geninfbilocal = \targetp{}{}{}\targetp{}{}{}.
\end{equation}
\end{subequations}
\Cref{lem:average twonorm} no longer applies in this case, since the left-hand side of \eqref{eq:gen inf bilocal indep} no longer has the form of a mixture of i.i.d.\ distributions.
%
%
The convergence of the above problem is still under control: if the postselection $\pstensor$ becomes closer and closer to being trivial while $\ninf$ becomes larger and larger, the set $\inftemp$ will converge to the set $\outputdistribs{\bilocaltopo}$ as characterized in equation \eqref{eq:convex game}.
However, compared with the original causal compatibility problem \eqref{eq:convex game}, we took two contradicting steps, and the sets $\inftemp$ and $\outputdistribs{\bilocaltopo}$ seem incomparable:
on the one hand, we let the agents share some additional correlations through the postselection $\pstensor$,
but we also restricted the cardinality of the sources $\isource{\dsunif}{}$.
Ideally, we would like a postselection $\pstensor$ such that, overall, the ``extra possibilities'' given by the postselection win over the restriction of the source cardinalities, so that we obtain an outer approximation $\outputdistribs{\bilocaltopo} \subseteq \inftemp$. 

\subsection{Fixing the postselection}
\label{sec:fixing the postselection}

It turns out that the two main criteria that the postselection $\pstensor$ should fulfill are the following. The main one is that only distinct values for the inputs of $\pstensor$ should pass the postselection: this will guarantee that $\outputdistribs{\bilocaltopo} \subseteq \inftemp$. The second important feature is that $\pstensor$ should be as unimportant as possible; that is, the impact of introducing the postselection in the network should be as low as possible. This will guarantee that $\inftemp$ is a good outer approximation of $\outputdistribs{\bilocaltopo}$. For these reasons, we make the choice (the label ``2'' refers to the numbers of inputs of the tensor)
\begin{equation}
\label{eq:psdiff choice}
\pstensor = \psdiff,
\end{equation}
where the right hand-side denotes the postselection over the outcome $\true$ of the tensor $\psdiffbase{}{}{}$ defined, for all $b \in\{\false,\true\}$, for all $\dsvalunif_1,\dsvalunif_2 \in \{1,\dots,\ninf\}$, through
\begin{equation}
\psdiffbase{\psval}{\dsvalunif_1}{\dsvalunif_2} := \delta_{\psval,\false}\delta_{\dsvalunif_1,\dsvalunif_2}
+ \delta_{\psval,\true}(1 - \delta_{\dsvalunif_1,\dsvalunif_2}).
\end{equation}
The effect of this postselection on the sources is the following:
\begin{equation}
\forall i,j \in\{1,\dots,\ninf\}\st \centertikz{
\node[tensornode] (u1) at (0,0) {$\dsunif$};
\node[copynode] (copy1) [above=\outcomevspace of u1] {};
\drawleg (u1.north) -- (copy1);
\node[tensornode] (u2) at (50pt,0) {$\dsunif$};
\node[copynode] (copy2) [above=\outcomevspace of u2] {};
\drawleg (u2.north) -- (copy2);
\node[psnode] (ps) at (25pt,40pt) {\psdiffname{2}};
\drawleg (copy1) -- (ps.south west);
\drawleg (copy2) -- (ps.south east);
\node[voidnode] (o1) at (-10pt, 40pt) {\indexstyle{i}};
\node[voidnode] (o2) at (60pt, 40pt) {\indexstyle{j}};
\foreach \x in {1,2}
	\drawleg (copy\x) -- (o\x);
}
= \frac{1}{\ninf(\ninf-1)}(1-\delta_{ij}).
\end{equation}
The set $\inftemp$ with this choice of postselection is denoted as $\infneq$. It is almost obvious that we will have, roughly speaking, that $\lim_{\ninf\to\infty}\infneq = \outputdistribs{\bilocaltopo}$: as $\ninf\rightarrow\infty$, the postselection strategy $\psdiff$ has almost no effect (the inputs are anyway not equal to one another with high probability), so that $\lim_{\ninf\to\infty}\infneq$ is essentially equal to $\lim_{\ninf\to\infty}\itempRestricted$, which is itself equal to the set $\outputdistribs{\bilocaltopo}$ as characterized in equation \eqref{eq:convex game}.
This claim will be made general and formal in \cref{th:convergence}.

Perhaps more surprising is the fact that, although $\inftemp$ and $\outputdistribs{\bilocaltopo}$ seemed incomparable, we have $\outputdistribs{\bilocaltopo} \subseteq \infneq$, i.e., this specific postselection guarantees that we obtain an outer approximation as desired. 
Let us prove this fact in the following lemma, which is, at least conceptually, a corollary of the more general \cref{th:certifying ps} that we will give in the next section \ref{sec:post-selected inflation} --- however, an explicit proof in this simple context captures the general proof idea.

\begin{lemma}
\label{lem:simple proof}
It holds that
\begin{equation}
\outputdistribs{\bilocaltopo} \subseteq \mathcal I_{\neq}^{(\ninf = 2)}(\bilocaltopo).
\end{equation}
\end{lemma}
\begin{proof}
Let $\targetp{}{}{} \in \outputdistribs{\bilocaltopo}$ so that we have probability tensors $\detonestrat{\alice_0}{}{}$, $\dettwostrat{\bob_0}{}{}{}$, $\detonestrat{\charlie_0}{}{}$ such that
\begin{equation}
\bilocalnetwork{\csalpha_1}{\csalpha_2}{\alice_0}{\bob_0}{\charlie_0}{}{}{}
=
\targetp{}{}{}.
\end{equation}
Let us choose the tensors of \eqref{eq:gen inf bilocal tensors} to be as follows. First off, the source $\isource{\hv}{}$ will actually be sending a tuple of four values, sampled from four independent $\isource{\maxunif}{}$ sources:
\begin{equation}
\isource{\hv}{(\csvalalpha_1,\csvalbeta_1,\csvalalpha_2,\csvalbeta_2)} = \isource{\maxunif}{\csvalalpha_1}\isource{\maxunif}{\csvalbeta_1}\isource{\maxunif}{\csvalalpha_2}\isource{\maxunif}{\csvalbeta_2}.
\end{equation}
The strategy $\dettwostrat{\alice}{}{}{}$ will consist in using the input $\dsvalalpha$ coming from the $\isource{\macrodsunif{2}}{}$ source of \eqref{eq:gen inf bilocal indep} to choose to use the value $\csvalalpha_\dsvalalpha$ coming from the source $\isource{\hv}{}$ as the input to the original strategy $\detonestrat{\alice_0}{}{}$. The other two strategies are analogous: we let
\begin{subequations}
\begin{align}
\dettwostrat{\alice}{\outputa}{\dsvalalpha}{(\csvalalpha_1,\csvalbeta_1,\csvalalpha_2,\csvalbeta_2)} &:= \detonestrat{\alice_0}{\outputa}{\csvalalpha_\dsvalalpha}, \\
\threedstrat{\bob}{\outputb}{\dsvalalpha}{\dsvalbeta}{(\csvalalpha_1,\csvalbeta_1,\csvalalpha_2,\csvalbeta_2)} &:= \dettwostrat{\bob_0}{\outputb}{\csvalalpha_\dsvalalpha}{\csvalbeta_\dsvalbeta}, \\
\dettwostrat{\charlie}{\outputc}{\dsvalbeta}{(\csvalalpha_1,\csvalbeta_1,\csvalalpha_2,\csvalbeta_2)} &:= \detonestrat{\charlie_0}{\outputc}{\csvalbeta_\dsvalbeta}.
\end{align}
\end{subequations}
We can now show that \eqref{eq:gen inf bilocal indep} is indeed verified: for all $a,b,c,\tilde a,\tilde b,\tilde c$,
\begin{subequations}
\begin{equation}
\centertikz{
\node[detnode] (alice) at (-1,0) {$\alice$};
\node[detnode] (bob) at (0,0) {$\bob$};
\node[detnode] (charlie) at (1,0) {$\charlie$};
\node[voidnode] (a) [above=\outcomevspace of alice] {\indexstyle{\outputa}};
\node[voidnode] (b) [above=\outcomevspace of bob] {\indexstyle{\outputb}};
\node[voidnode] (c) [above=\outcomevspace of charlie] {\indexstyle{\outputc}};
\node[copynode] (alphacopy) at (-0.6,\psbilocalsourcepos) { };
\node[tensornode] (alpha) [below=\outcomevspace of alphacopy] {$\macrodsunif{2}$};
\node[copynode] (betacopy) at (0.5,\psbilocalsourcepos) { };
\node[tensornode] (beta) [below=\outcomevspace of betacopy] {$\macrodsunif{2}$};
\drawdsource (alpha.north) -- (alphacopy);
\drawdsource (alphacopy) -- \AnchorOneTwo{alice}{south};
\drawdsource (alphacopy) -- \AnchorOneThree{bob}{south};
\drawdsource (beta.north) -- (betacopy);
\drawdsource (betacopy) -- \AnchorTwoThree{bob}{south};
\drawdsource (betacopy) -- \AnchorOneTwo{charlie}{south};
\drawoutcomeleg (alice.north) -- (a.south);
\drawoutcomeleg (bob.north) -- (b.south);
\drawoutcomeleg (charlie.north) -- (c.south);
%
\node[detnode] (alice2) at (2,0) {$\alice$};
\node[detnode] (bob2) at (3,0) {$\bob$};
\node[detnode] (charlie2) at (4,0) {$\charlie$};
\node[voidnode] (a2) [above=\outcomevspace of alice2] {\indexstyle{\tilde\outputa}};
\node[voidnode] (b2) [above=\outcomevspace of bob2] {\indexstyle{\tilde\outputb}};
\node[voidnode] (c2) [above=\outcomevspace of charlie2] {\indexstyle{\tilde\outputc}};
\node[copynode] (alphacopy2) at (2.5,\psbilocalsourcepos) { };
\node[tensornode] (alpha2) [below=\outcomevspace of alphacopy2] {$\macrodsunif{2}$};
\node[copynode] (betacopy2) at (3.6,\psbilocalsourcepos) { };
\node[tensornode] (beta2) [below=\outcomevspace of betacopy2] {$\macrodsunif{2}$};
\drawdsource (alpha2.north) -- (alphacopy2);
\drawdsource (alphacopy2) -- \AnchorOneTwo{alice2}{south};
\drawdsource (alphacopy2) -- \AnchorOneThree{bob2}{south};
\drawdsource (beta2.north) -- (betacopy2);
\drawdsource (betacopy2) -- \AnchorTwoThree{bob2}{south};
\drawdsource (betacopy2) -- \AnchorOneTwo{charlie2}{south};
\drawoutcomeleg (alice2.north) -- (a2.south);
\drawoutcomeleg (bob2.north) -- (b2.south);
\drawoutcomeleg (charlie2.north) -- (c2.south);
%
%
\node[tensornode] (hv) at (1.5,-3.4) {$\hv$};
\node[copynode] (hvcopy) [above=8pt of hv] { };
\drawleg (hv.north) -- (hvcopy.south);
\drawdashedleg (hvcopy) -- ($(alpha.west) - (11pt,9pt)$) -- \AnchorTwoTwo{alice}{south};
\drawdashedleg (hvcopy) -- ($(beta.west) - (9pt,9pt)$) -- \AnchorThreeThree{bob}{south};
\drawdashedleg (hvcopy) -- ($(beta.east) + (6pt,-8pt)$) -- \AnchorTwoTwo{charlie}{south};
\drawdashedleg (hvcopy) -- ($(alpha2.west) - (6pt,8pt)$) -- \AnchorTwoTwo{alice2}{south};
\drawdashedleg (hvcopy) -- ($(beta2.west) - (6pt,15pt)$) -- \AnchorThreeThree{bob2}{south};
\drawdashedleg (hvcopy) -- ($(beta2.east) + (6pt,-12pt)$) -- \AnchorTwoTwo{charlie2}{south};
%
%
\node[psnode] (f1) [above left=140pt and 5pt of hv] {\psdiffname{2}};
\node[psnode] (f2) [above right=140pt and 5pt of hv] {\psdiffname{2}};
\drawdashedleg (alphacopy) to [out=90,in=220] (f1.south west);
\drawdashedleg (alphacopy2) to [out=75,in=335] (f1.south east);
\drawdashedleg (betacopy) to [out=105,in=215] (f2.south west);
\drawdashedleg (betacopy2) to [out=90,in=330] (f2.south east);
}
\end{equation}%
\begin{align}
&=
\frac{1}{4}
\sum_{\substack{\dsvalalpha_1,\dsvalalpha_2,\dsvalbeta_1,\dsvalbeta_2\in\{1,2\} \\ \dsvalalpha_1 \neq \dsvalalpha_2, \dsvalbeta_1 \neq \dsvalbeta_2}}
\int \dd{\csvalalpha_1}\dd{\csvalalpha_2}\dd{\csvalbeta_1}\dd{\csvalbeta_2}
\detonestrat{\alice_0}{\outputa}{\csvalalpha_{\dsvalalpha_1}}
\dettwostrat{\bob_0}{\outputb}{\csvalalpha_{\dsvalalpha_1}}{\csvalbeta_{\dsvalbeta_1}}
\detonestrat{\charlie_0}{\outputc}{\csvalbeta_{\dsvalbeta_1}}
\detonestrat{\alice_0}{\tilde\outputa}{\csvalalpha_{\dsvalalpha_2}}
\dettwostrat{\bob_0}{\tilde\outputb}{\csvalalpha_{\dsvalalpha_2}}{\csvalbeta_{\dsvalbeta_2}}
\detonestrat{\charlie_0}{\tilde\outputc}{\csvalbeta_{\dsvalbeta_2}} \\
&= \frac{1}{4}
\sum_{\substack{\dsvalalpha_1,\dsvalalpha_2,\dsvalbeta_1,\dsvalbeta_2\in\{1,2\} \\ \dsvalalpha_1 \neq \dsvalalpha_2, \dsvalbeta_1 \neq \dsvalbeta_2}}
\int \textup{d}{\csvalalpha_{\dsvalalpha_1}}\textup{d}{\csvalbeta_{\dsvalbeta_1}}
\left(
\detonestrat{\alice_0}{\outputa}{\csvalalpha_{\dsvalalpha_1}}
\dettwostrat{\bob_0}{\outputb}{\csvalalpha_{\dsvalalpha_1}}{\csvalbeta_{\dsvalbeta_1}}
\detonestrat{\charlie_0}{\outputc}{\csvalbeta_{\dsvalbeta_1}}
\right)
\int\textup{d}{\csvalalpha_{\dsvalalpha_2}}\textup{d}{\csvalbeta_{\dsvalbeta_2}}
\left(
\detonestrat{\alice_0}{\tilde\outputa}{\csvalalpha_{\dsvalalpha_2}}
\dettwostrat{\bob_0}{\tilde\outputb}{\csvalalpha_{\dsvalalpha_2}}{\csvalbeta_{\dsvalbeta_2}}
\detonestrat{\charlie_0}{\tilde\outputc}{\csvalbeta_{\dsvalbeta_2}}
\right) \\
&= \frac{1}{4}
\sum_{\substack{\dsvalalpha_1,\dsvalalpha_2,\dsvalbeta_1,\dsvalbeta_2\in\{1,2\} \\ \dsvalalpha_1 \neq \dsvalalpha_2, \dsvalbeta_1 \neq \dsvalbeta_2}}
\bilocalnetwork{\csalpha_1}{\csalpha_2}{\alice_0}{\bob_0}{\charlie_0}{\outputa}{\outputb}{\outputc}
\bilocalnetwork{\csalpha_1}{\csalpha_2}{\alice_0}{\bob_0}{\charlie_0}{\tilde\outputa}{\tilde\outputb}{\tilde\outputc} \\
&= \targetp{\outputa}{\outputb}{\outputc}\targetp{\tilde\outputa}{\tilde\outputb}{\tilde\outputc},
\end{align}
\end{subequations}
so that indeed $\targetp{}{}{}\in\mathcal I_{\neq}^{(\ninf = 2)}(\bilocaltopo)$.
\end{proof}

Note that this proof idea would work regardless of the precise behavior of the $\isource{\macrodsunif{2}}{}$ and $\psdiff$ tensors, as long as the postselection enables $\dsvalalpha_1 \neq \dsvalalpha_2$ and $\dsvalbeta_1 \neq \dsvalbeta_2$. 



\newpage
\section{Postselected inflation: formal aspects}
\label{sec:post-selected inflation}

In this section, we give a general description of the sort of outer approximations that we consider, and we prove a number of results that characterize these.
The previous section \ref{sec:outer approx} gave some intuition for the idea of this construction.

\subsection{Definition}

Consider a multi-network scenario $\genmultopo$ as in equation
\eqref{eq:def genmultopoexpl}.
%
The set $\outputdistribs{\genmultopo}$ was defined in \cref{def:causal compat}.
Let us define the postselected inflation set $\geninfset{\genmultopo}{\ninf}{\ninfcons}$ --- parametrized by two integers $\ninf$ (number of output values for the discretized sources) and $\ninfcons$ (order of the tensor product constraints) --- which is meant to be the outer approximation of the set $\outputdistribs{\{\gentopo_\topoindex\}_{\topoindex=1}^\topocount}.$ 
In the following definition, for all integers $k$, the tensor $\gpstensor{k}$ denotes the postselection of the $\true$ outcome of the probability tensor $\gpstensorbase{k}$ (see section \ref{sec:post-selection} for the definition of postselection). This probability tensor $\gpstensorbase{k}$ has $k$ input legs, each with domain $\{1,\dots,\ninf\}$ in this context, and is defined, for all $b\in\{\false,\true\}$, $\dsvalalpha_1,\dots,\dsvalalpha_{k}\in\{1,\dots,\ninf\}$, as
\begin{equation}
\gpstensorbaseargs{k}{b}{\dsvalalpha_1}{\dsvalalpha_2}{\dsvalalpha_k} = \left\{\begin{aligned}
&\delta_{b,\true}, \textup{ if all the values }\dsvalalpha_1,\dsvalalpha_2,\dots,\dsvalalpha_k\textup{ are pairwise distinct,} \\
&\delta_{b,\false} \textup{ else.}
\end{aligned}\right.
\end{equation}
%
Now, the following diagrammatic constraint \eqref{eq:geninfcondition} is quite large, so let us
describe how to obtain it: for each network $\topoindex \in \{1,\dots,\topocount\}$,
\begin{myitem}
\item Take the corresponding tensor contraction, as in equation \eqref{eq:causal compat multiple config condition}.
\item Replace each source $\isource{\maxunif}{}$ by a uniform source $\isource{\dsunif}{}$ over $\ninf$ values.
\item Add an additional input to all the strategy tensors, e.g.\ $\gonestratdet{\alice_1}{}{} \rightarrow \geninfstrat{\alice_1}{}{}$.
\item Duplicate $\ninfcons$ times (take the $\ninfcons$-fold tensor product of) both the network tensor and the target distribution $\macrogtargetp{\targetpname_\topoindex}$ in equation \eqref{eq:causal compat multiple config condition}.
\item Connect all the agent's additional inputs to a global randomness source $\isource{\hv}{}$. Importantly, the tensor $\isource{\hv}{}$ is the same for all $\topoindex = 1,\dots,\topocount$.
\item 
At last, postselect on the $\scount_\topoindex \cdot \ninfcons$ sources $\isource{\dsunif}{}$ having different values. This is the step which requires $\ninf \geq \scount_\topoindex\cdot\ninfcons$.
\end{myitem}
This yields the following definition (recall that the style difference between the dashed and solid edges is there to guide the eye but implies the same operation of tensor contraction).

\newpage
\begin{definition}[Postselected inflation]
\label{def:postselected inflation set}
Consider a multi-network scenario $\genmultopoexpl$. Let $\ninfcons$ be an integer parameter, and let $\ninf$ be any integer such that
\begin{equation}
\label{eq:ps m n condition}
\ninf \geq \max_{\topoindex\in\{1,\dots,\topocount\}} \scount_\topoindex\cdot \ninfcons. 
\end{equation}
A list of distributions $\genmultopotargetps$ belongs to the set $\geninfset{\genmultopo}{\ninf}{\ninfcons}$ if and only if there exist probability tensors
\begin{subequations}
\label{eq:prob gen inf set}
\begin{equation}
\label{eq:prob gen inf set primitives}
\left\{\geninfstrat{\alice_\pindex}\right\}_{\pindex=1}^{\pcount}, \isource{\hv}{},
\end{equation}
such that, for all $\topoindex = 1,\dots,\topocount$, it holds that
\begin{multline}
\label{eq:geninfcondition}
\genpsinfcondition \\
{ } \\
{} \\
= 
\macrogtargetptwoargs{\targetpname_\topoindex} \macrogtargetptwoargs{\targetpname_\topoindex} \overset{(\ninfcons-3)}{\dots} \macrogtargetptwoargs{\targetpname_\topoindex}.\hspace{2cm}
\end{multline}
\end{subequations}
\end{definition}


\subsection{Remarks about the implementation}
\label{sec:remarks about implementation}

The idea behind the above \cref{def:postselected inflation set} is that testing for $\genmultopotargetps \in \geninfset{\genmultopo}{\ninf}{\ninfcons}$ can be easily formulated as a linear program over $\isource{\hv}{}$ only. Indeed, it is clear from \eqref{eq:geninfcondition} that one can assume without loss of generality that the output domain of $\isource{\hv}{}$ is the set of all tuples of $\pcount$ deterministic strategies such that each agent $\npindex$ then chooses the $\pmap(\npindex)$-th element of this tuple to use as their strategy.
This output domain is quite large. For each $\pindex \in \{1,\dots,\pcount\}$, we denote with $\nout_\pindex$ the number of outcomes of the strategy $\pindex$,
and with $\nin_\pindex$ the number of input legs of this strategy.\footnote{This number $\nin_\pindex$ would be the length of the list $\cmap_\topoindex(\npindex)$ for any $\topoindex$ and $\npindex$ such that $\pmap_\topoindex(\npindex) = \pindex$.} This strategy will then receive from the sources $\isource{\dsunif}{}$ a tuple of $\nin_\pindex$ values each in the set $\{1,\dots,\ninf\}$, and there are $\ninf^{\nin_\pindex}$ such tuples. Since for each such tuple, the deterministic strategy may give $\nout_\pindex$ different outcomes, there are thus $\exp\big(\ln[\nout_\pindex]\ninf^{\nin_\pindex}\big)$ such deterministic strategies. Overall, the distribution $\isource{\hv}{}$ should thus have an output (indexing the possible strategies) with cardinality
\begin{equation}
\exp\left(\sum_{\pindex=1}^\pcount 
\ln[\nout_\pindex]
\ninf^{\nin_\pindex}
\right),
\end{equation}
which is quite a large number as $\ninf$ grows. 

For this reason, it is crucial to be able to formulate versions of \cref{def:postselected inflation set} with $\ninf$ being a small integer. Depending on the number $\scount_\topoindex$ of sources in each network $\topoindex$, this may be problematic with respect to the postselection constraint of equation \eqref{eq:ps m n condition}. It is thus useful to keep in mind that although \cref{def:postselected inflation set} is fully general for our purposes, one should in practice investigate the network's structure to see if it is possible to have less postselection than in \eqref{eq:geninfcondition}. 
Additionally, as already alluded to in \cref{sec:fixing the postselection}, having less postselection is also helpful for the purpose of letting the outer approximation generated by the postselected inflation procedure be as tight as possible.
In this work, two examples of enforcing less postselection than in \cref{def:postselected inflation set} are given: the case of a single-network scenario ($\topocount = 1$) where all agents use independent strategies ($\pcount = \npcount$) --- see \cref{sec:fixing the postselection} as well as \cref{sec:triangle network three strats} --- and the case of the Correlated Sleeper --- see \cref{sec:sc outer approx}.
Furthermore, in the latter, we show how letting the parameter $\ninfcons$ depend on $\topoindex$ can be convenient with respect to these requirements. It is also possible to think of taking different value of $\ninf$ for each source $\isource{\dsunif}{}$ appearing in the left-hand side of \eqref{eq:geninfcondition}.

\subsection{Results}

We now prove a number of results regarding the postselected inflation scheme. As will be discussed in \cref{sec:fanout inflation}, these results extend those of \cite{navascues_inflation_2020} to the case of our multi-network scenarios with subsets of agents using the same strategy. The proofs are gathered in \app~\ref{app:post-selected inflation proofs}.
The first result establishes that the postselected inflation scheme indeed allows to certify causal incompatibility.
The proof is the generalization of that of \cref{lem:simple proof} and \cref{lem:app inf feasible}.

\begin{restatable}[Certification]{theorem}{ThCertifyingPs}
\label{th:certifying ps}
Let $\genmultopo$, $\ninf$, $\ninfcons$ be as in \cref{def:postselected inflation set}.
It holds that
\begin{equation}
\label{eq:causal compat certified}
\outputdistribs{\genmultopo} 
\subseteq
\geninfset{\genmultopo}{\ninf}{\ninfcons}{}.
\end{equation}
\end{restatable}
%
The next theorem establishes that increasing the parameters $\ninf$ and $\ninfcons$ that appear in the postselected inflation problem will make the outer approximation $\geninfset{\genmultopo}{\ninf}{\ninfcons}$ smaller and smaller, while remaining an outer approximation of $\outputdistribs{\genmultopo}$ thanks to \cref{th:certifying ps}.

\begin{restatable}[Hierarchy]{theorem}{ThInclusionRelations}
\label{th:inclusion relations}
For all $\genmultopoexpl$ and $\ninf, \ninfcons$  ($\ninf \geq \max_\topoindex \scount_\topoindex \cdot \ninfcons$) as well as $\ninf', \ninfcons'$ ($\ninf' \geq \max_\topoindex \scount_\topoindex \cdot \ninfcons'$) such that
\begin{equation}
\ninf \geq \ninf' \text{ and } \ninfcons \geq \ninfcons',
\end{equation}
it holds that
\begin{equation}
\label{eq:general inclusion}
\geninfset{\genmultopo}{\ninf}{\ninfcons}
\subseteq
\geninfset{\genmultopo}{\ninf'}{\ninfcons'}.
\end{equation}
\end{restatable}

To discuss convergence, we will make use of the $p$-norms on $\mathbb{R}^k$ (defined explicitly in \cref{def:p norms} in \app~\ref{app:post-selected inflation proofs}). We will take the 1-norm as the operationally most meaningful one, since it is related to an operational measure of distinguishability of two distributions.
The following theorem makes precise the fact that $\geninfset{\genmultopo}{\ninf}{\ninfcons} \overset{\ninf\rightarrow\infty}{\longrightarrow} \outputdistribs{\genmultopo}$. The fact that increasing the parameter $\ninfcons$ to its maximal value, $\left\lfloor \ninf / \max_\topoindex \scount_\topoindex \right\rfloor$, does not improve the below convergence rate suggests that the convergence rate is actually better than the one we give here.
%

\begin{restatable}[Convergence]{theorem}{ThConvergence}
\label{th:convergence}
Let $\genmultopoexpl$, $\ninf$ and $\ninfcons$ be as in \cref{def:postselected inflation set}. We assume that $\ninfcons \geq 2$. Then, for any list of distributions
\begin{equation}
\left(\macrogtargetp{\targetpname_\topoindex} \in \mathbb{R}^{d_\topoindex}\right)_{\topoindex=1}^\topocount
\in
\geninfset{\genmultopo}{\ninf}{\ninfcons}
\end{equation}
where we made explicit the total number of outcomes $d_\topoindex$ of each distribution,\footnote{To tie the notation together: we have that $d_\topoindex = \prod_{\npindex=1}^{\npcount_\topoindex} \nout_{\nu_\topoindex(\npindex)}$.}
it holds that
\begin{align}
\label{eq:convergence rate}
\inf_{
\scalebox{0.95}{$
\scriptsize\big(\macrogtargetp{\targetptildename_\topoindex}\big)_{\topoindex=1}^\topocount\in\outputdistribs{\genmultopo}
$}
}
\frac{1}{\topocount}\sum_{\topoindex=1}^\topocount
\onenorm{\macrogtargetp{\targetpname_\topoindex} - \macrogtargetp{\targetptildename_\topoindex}}
\ 
&\leq
\ 
\frac{
\sqrt{12\topocount^{-1} \sum_{\topoindex=1}^\topocount d_\topoindex \scount_\topoindex\left(\scount_\topoindex - \frac{1}{2}\right) }
}
{\sqrt \ninf}
+ \bigo{\frac{1}{\ninf\sqrt{\ninf}}}.
\end{align}
\end{restatable}

A corollary of \cref{th:certifying ps} and \cref{th:convergence} is the following:

\begin{restatable}{corollary}{CorollaryConvergence}
\label{corollary:convergence}
Consider some $\genmultopoexpl$ and $\ninfcons \geq 2$, and let $\ninf_0 := \max_{\topoindex} \scount_\topoindex \cdot \ninfcons$.
In the topology induced by the metric
\begin{equation}
\metric\left[
\genmultopotargetps,
\Big(\macrogtargetp{\targetptildename_\topoindex}\Big)_{\topoindex=1}^\topocount
\right]
:=
\frac{1}{\topocount}\sum_{\topoindex=1}^\topocount
\onenorm{\macrogtargetp{\targetpname_\topoindex} - \macrogtargetp{\targetptildename_\topoindex}},
\end{equation}
and denoting $\closure{X}$ the closure of a set $X$ in this topology, it holds that
\begin{equation}
\outputdistribs{\genmultopo} \subseteq \bigcap_{\ninf = \ninf_0}^{\infty} \geninfset{\genmultopo}{\ninf}{\ninfcons} \subseteq \closure{\outputdistribs{\genmultopo}}.
\end{equation}
\end{restatable}

\newpage
\section{Application: Correlated Sleeper}
\label{sec:applications}

In this section, we apply the postselected inflation formalism introduced in sections \ref{sec:outer approx} and \ref{sec:post-selected inflation} to the Correlated Sleeper protocol, which was introduced in section \ref{sec:intro sc}. This allows to demonstrate the use of the formalism in a simple example.

\subsection{Feasible region}

\paragraph{Parametrization.} 

It will be useful to parametrize the achievable distributions in the Correlated Sleeper protocol as follows. The symmetry of the protocol and the marginal constraint \eqref{eq:marginal constraint} first give the following lemma. See \app~\ref{sec:sc parametrization} for the proof.
\begin{restatable}{lemma}{LemmaScLambdas}
\label{lem:sc lambdas}
Let $\atargetp{1}{}{}$ and $\atargetp{2}{}{}$ be as in \eqref{eq:sc det tensor param}. 
Define for all $\topoindex \in \{1,2\}$
\begin{equation}
	\lambda_\topoindex := \atargetp{\topoindex}{1}{1}.
\end{equation}
It then holds that
\begin{subequations}
\begin{align}
\atargetp{\topoindex}{2}{2} &= \lambda_{\topoindex}, \\
\atargetp{\topoindex}{1}{2} = \atargetp{\topoindex}{2}{1} &= \frac{1}{2} - \lambda_\topoindex,
%
\end{align}
\end{subequations}
with
\begin{equation}
	\lambda_\topoindex \in \left[0,\frac{1}{2}\right].
\end{equation}
\end{restatable}
It is thus sufficient for us to specify the value of the pair $(\lambda_1,\lambda_2)$ to characterize the output behavior of a given strategy used by the agent $\alice$ in the Correlated Sleeper protocol.
In this parametrization, the objective function of the optimization problem \eqref{eq:sc det tensor} is just $\lambda_1 + \lambda_2$.
For instance, the strategy of \eqref{eq:app example strat}, where $\alice$ only looks at her first input, leads to $(\lambda_1 = 1/2, \lambda_2 = 1/4)$ corresponding to a score of $3/4$, while a completely mixed behavior (where $\alice$ ignores her inputs and outputs a random bit) leads to $(\lambda_1 = 1/4, \lambda_2 = 1/4)$, corresponding to a score of $1/2$. Generally speaking, if $(\lambda_1,\lambda_2)$ is feasible with a certain strategy, then using the same strategy while exchanging the role of the two inputs that $\alice$ receives shows that also $(\lambda_2,\lambda_1)$ is feasible.
Furthermore, we prove the following lemma in \app~\ref{sec:sc parametrization}.
\begin{restatable}{lemma}{LemmaBoundLambdas}
	\label{lem:bound lambdas}
	Let $(\lambda_1,\lambda_2)$ be as in \cref{lem:sc lambdas}. It holds for all $\topoindex\in\{1,2\}$ that
	\begin{equation}
		\lambda_\topoindex \geq \frac{1}{4}.
	\end{equation}
\end{restatable}
This bound can be equivalently stated as $\atargetp{\topoindex}{1}{2} \leq 1/4$, which means that the best chance that $\alice$ has to obtain \emph{distinct} outcomes in the two round is to adopt a completely mixed behavior --- any deviation from this will inevitably lead to increased correlations.

\paragraph{Multi-network scenario.}

Recall that in the formulation \eqref{eq:sc causal compat formulation}, we defined the feasible region of the Correlated Sleeper protocol with the causal compatibility problem
\begin{equation}
\left(\atargetp{1}{}{},\atargetp{2}{}{},\isource{\macrodsunif{2}}{}\right)
\in\outputdistribs{\scnetone,\scnettwo,\scnetthree},
\end{equation}
where the relevant networks were defined in equations \eqref{eq:sc config 1}-\eqref{eq:sc config 3}. 
Here, we want to add an extra network that will look redundant at first, but that will eventually yield non-trivial feasibility constraints once considered from the perspective of the postselected inflation test.
This network, whose graph is shown in \cref{fig:sc config 4}, is one with three disconnected components where, in the first, the two $\alice$'s are connected with their left input, in the second the two $\alice$'s are connected with their right input, and the third is an isolated $\alice$:
\begin{subequations}
\begin{align}
\scnetfour &= (\pcount = 1,\npcount_4 = 5,\scount_4 = 8,\pmap_4,\cmap_4), \\
\forall \npindex \in \{1,\dots,5\},\ \pmap_4(\npindex) &= 1, &\netnote{only one strategy}\\
\cmap_4(1) &= (1,2),\ \cmap_4(2) = (1,3), &\netnote{first input connected}\\
\cmap_4(3) &= (4,6),\ \cmap_4(4) = (5,6), &\netnote{second input connected}\\
\cmap_4(5) &= (7,8). &\netnote{isolated agent}
\end{align}
\end{subequations}
In this network, we want the distribution $\atargetp{1}{}{} \atargetp{2}{}{} \isource{\macrodsunif{2}}{}$, with this output ordering, to be feasible.
\begin{figure}[h!]
\centering
\begin{equation*}
\centertikz{
\node[agent] (a1) {$A$};
\node[agent] (a2) [right=20pt of a1] {$A$};
\node[source] (alpha) [below left=30pt and 10pt of a1] {$1$};
\drawabsleg (alpha) -- (a1);
\drawabsleg (alpha) -- (a2);
\node[source] (beta1) [below right=30pt and 10pt of a1] {$2$};
\drawabsleg (beta1) -- (a1);
\node[source] (beta2) [below right=30pt and 10pt of a2] {$3$};
\drawabsleg (beta2) -- (a2);
}
\quad
\centertikz{
\node[agent] (a1) {$A$};
\node[agent] (a2) [right=20pt of a1] {$A$};
\node[source] (alpha1) [below left=30pt and 10pt of a1] {$4$};
\drawabsleg (alpha1) -- (a1);
\node[source] (alpha2) [below left=30pt and 10pt of a2] {$5$};
\drawabsleg (alpha2) -- (a2);
\node[source] (beta) [below right=30pt and 10pt of a2] {$6$};
\drawabsleg (beta) -- (a1);
\drawabsleg (beta) -- (a2);
}
\quad
\centertikz{
\node[agent] (a) {$A$};
\node[source] (alpha) [below left=30pt and 10pt of a] {$7$};
\node[source] (beta) [below right=30pt and 10pt of a] {$8$};
\drawabsleg (alpha) -- (a);
\drawabsleg (beta) -- (a);
}
\end{equation*}
\caption{The network $\scnetfour$.}
\label{fig:sc config 4}
\end{figure}
%

We are now ready to define the feasible region $\appfeasible$
in the form that we will feed in to the postselected inflation:
\begin{equation}
\label{eq:app feasible}
\appfeasible := \left\{
 \Big(\atargetp{1}{}{},\atargetp{2}{}{}\Big)\ \middle|\ 
\left(\atargetp{1}{}{},\ \atargetp{2}{}{},\ \isource{\macrodsunif{2}}{},\ \atargetp{1}{}{}\atargetp{2}{}{}\isource{\macrodsunif{2}}{}\right) \in \outputdistribs{\left(\gentopo_\topoindex^{(\textup{s})}\right)_{\topoindex=1}^4}\right\}. 
\end{equation}
Furthermore, the subset of $[1/4,1/2]^{\times 2}$ that corresponds to the parametrization of $\appfeasible$ according to \cref{lem:sc lambdas,lem:bound lambdas} will be denoted $\appfeasibleparam$.

\subsection{Outer approximation of the feasible region}
\label{sec:sc outer approx}

The outer approximation of the feasible region will be constructed thanks to a postselected inflation feasibility problem, but with two twists with respect to \cref{def:postselected inflation set}:
\begin{itemize}
\item We will not use the same parameter $\ninfcons$ across the four networks: this is to allow $\ninf$ to be not too large. Essentially, we will take the reasonable value of $\ninf = 4$ and, for each network, we will let $\ninfcons$ be as large as possible given the postselection and $\ninf = 4$. We indicate these values of $\ninfcons$ in the equations \eqref{eq:app inf set}.
\item We will not postselect on \emph{all} the sources taking different values as in \cref{def:postselected inflation set}, but rather introduce a more minimal postselection scheme which allows, according to the arguments of \cref{sec:fixing the postselection}, to have a tighter outer approximation. In that sense, the mention of ``inflation'' in the equations \eqref{eq:app inf set} is to be understood as a small generalization of \cref{def:postselected inflation set}.
\end{itemize}

Recall that the style difference between the dashed and solid edges is there to guide the eye but implies the same operation of tensor contraction.
%
\begin{definition}
\label{def:app inf set}
We let our outer approximation of the feasible region $\appfeasible$ be the set $\appinffeasible$ defined as
\begin{subequations}
\label{eq:app inf set}
\begin{align}
\appinffeasible := \Bigg\{ \Big(\atargetp{1}{}{},\atargetp{2}{}{}\Big)
\ \Bigg|\ 
\exists \threedstrat{\alice}{}{}{}{}, \isource{\hv}{} \textup{ s.t. }& \\ \nonumber\\
\label{eq:app inf set first eq}
\centertikz{
\node (a0) {};
\foreach \x/\y in {1/0,2/1,3/2,4/3}
{
	\node[detnode] (a\x) [right=\appplayeroffset of a\y] {$\alice$};
	\node[voidnode] (o\x) [above=\outcomevspace of a\x] {};
	\drawleg (a\x.north) -- (o\x.south);
	\node[tensornode] (beta\x) [below right=\appbetapos of a\x] {$\macrodsunif{4}$};
	\node[copynode] (copybeta\x) [above=\outcomevspace of beta\x] {};
	\drawleg (beta\x.north) -- (copybeta\x);
	\drawleg (copybeta\x) -- \AnchorTwoThree{a\x}{south};
}
\foreach \x in {1,3} 
{
	\node[tensornode] (alpha\x) [below left=\appalphapos of a\x] {$\macrodsunif{4}$};
	\node[copynode] (copyalpha\x) [above=\outcomevspace of alpha\x] {};
	\drawleg (alpha\x.north) -- (copyalpha\x);
}
\foreach \x/\y in {1/1,1/2,3/3,3/4}
	\drawleg (copyalpha\x) -- \AnchorOneThree{a\y}{south};
\node[psnode] (psbeta) [below left=\appbetaps of a1] {\psdiffname{4}};
\foreach \x/\y/\z in {1/1.00/0.00,
					  2/0.67/0.33,
  					  3/0.33/0.67,
  					  4/0.00/1.00}
	\drawdashedleg (copybeta\x) to [out=170,in=340] ($\y*(psbeta.south west)+\z*(psbeta.south east)$);
\node[psnode] (psalpha) [below left=\appalphaps of a1] {\psdiffname{2}};
\drawdashedleg (copyalpha1) to [out=170,in=340] (psalpha.south west);
\drawdashedleg (copyalpha3) to [out=170,in=340] (psalpha.south east);
\node[tensornode] (hv) [below right=\apphvpos of a4] {$\hv$};
\node[copynode] (copyhv) [above=\outcomevspace of hv] {};
\drawleg (hv.north) -- (copyhv);
\foreach \x in {1,2,3,4}
	\drawdashedleg (copyhv) to [out=180,in=340] ($0.2*(a\x.south west)+0.8*(a\x.south east)$);
}
&=
\atargetp{1}{}{}\atargetp{1}{}{} \quad\netnotesmall{$\ninfcons=2$ inf.\  of $\scnetone$},\\ \nonumber\\
\centertikz{
\node (a0) {};
\foreach \x/\y in {1/0,2/1,3/2,4/3}
{
	\node[detnode] (a\x) [right=\appplayeroffset of a\y] {$\alice$};
	\node[voidnode] (o\x) [above=\outcomevspace of a\x] {};
	\drawleg (a\x.north) -- (o\x.south);
	\node[tensornode] (alpha\x) [below left=\appalphapos of a\x] {$\macrodsunif{4}$};
	\node[copynode] (copyalpha\x) [above=\outcomevspace of alpha\x] {};
	\drawleg (alpha\x.north) -- (copyalpha\x);
	\drawleg (copyalpha\x) -- \AnchorOneThree{a\x}{south};
}
\foreach \x in {2,4}
{
	\node[tensornode] (beta\x) [below right=\appbetapos of a\x] {$\macrodsunif{4}$};
	\node[copynode] (copybeta\x) [above=\outcomevspace of beta\x] {};
	\drawleg (beta\x.north) -- (copybeta\x);
}
\drawleg (copybeta2) -- \AnchorTwoThree{a2}{south};
\drawleg (copybeta4) -- \AnchorTwoThree{a4}{south};
\drawleg (copybeta2) to [out=180,in=270] ($0.5*(a1.south west)+0.5*(a1.south east)$);
\drawleg (copybeta4) to [out=180,in=270] ($0.5*(a3.south west)+0.5*(a3.south east)$);
\node[psnode] (psalpha) [below left=\appalphaps of a1] {\psdiffname{4}};
\foreach \x/\y/\z in {1/1.00/0.00,
					  2/0.67/0.33,
  					  3/0.33/0.67,
  					  4/0.00/1.00}
	\drawdashedleg (copyalpha\x) to [out=170,in=340] ($\y*(psalpha.south west)+\z*(psalpha.south east)$);
\node[psnode] (psbeta) [below left=\appbetaps of a1] {\psdiffname{2}};
\drawdashedleg (copybeta2) to [out=170,in=340] (psbeta.south west);
\drawdashedleg (copybeta4) to [out=170,in=340] (psbeta.south east);
\node[tensornode] (hv) [below right=\apphvpos of a4] {$\hv$};
\node[copynode] (copyhv) [above=\outcomevspace of hv] {};
\drawleg (hv.north) -- (copyhv);
\foreach \x in {1,2,3,4}
	\drawdashedleg (copyhv) to [out=180,in=340] ($0.2*(a\x.south west)+0.8*(a\x.south east)$);
}
&=
\atargetp{2}{}{}\atargetp{2}{}{} \quad\netnotesmall{$\ninfcons=2$ inf.\  of $\scnettwo$}, \\ \nonumber\\
\centertikz{
\node (a0) {};
\foreach \x/\y in {1/0,2/1,3/2,4/3}
{
	\node[detnode] (a\x) [right=\appplayeroffset of a\y] {$\alice$};
	\node[voidnode] (o\x) [above=\outcomevspace of a\x] {};
	\drawleg (a\x.north) -- (o\x.south);
	\node[tensornode] (alpha\x) [below left=\appalphapos of a\x] {$\macrodsunif{4}$};
	\node[copynode] (copyalpha\x) [above=\outcomevspace of alpha\x] {};
	\drawleg (copyalpha\x) -- (alpha\x.north);
	\drawleg (copyalpha\x) -- \AnchorOneThree{a\x}{south};
	\node[tensornode] (beta\x) [below right=\appbetapos of a\x] {$\macrodsunif{4}$};
	\node[copynode] (copybeta\x) [above=\outcomevspace of beta\x] {};
	\drawleg (copybeta\x) -- (beta\x.north);
	\drawleg (copybeta\x) -- \AnchorTwoThree{a\x}{south};
}
\node[psnode] (psbeta) [below left=\appbetaps of a1] {\psdiffname{4}};
\foreach \x/\y/\z in {1/1.00/0.00,
					  2/0.67/0.33,
  					  3/0.33/0.67,
  					  4/0.00/1.00}
	\drawdashedleg (copybeta\x) to [out=170,in=340] ($\y*(psbeta.south west)+\z*(psbeta.south east)$);
\node[psnode] (psalpha) [below left=\appalphaps of a1] {\psdiffname{4}};
\foreach \x/\y/\z in {1/1.00/0.00,
					  2/0.67/0.33,
					  3/0.33/0.67,
					  4/0.00/1.00}
	\drawdashedleg (copyalpha\x) to [out=170,in=340] ($\y*(psalpha.south west)+\z*(psalpha.south east)$);
\node[tensornode] (hv) [below right=\apphvpos of a4] {$\hv$};
\node[copynode] (copyhv) [above=\outcomevspace of hv] {};
\drawleg (hv.north) -- (copyhv);
\foreach \x in {1,2,3,4}
	\drawdashedleg (copyhv) to [out=180,in=340] ($0.2*(a\x.south west)+0.8*(a\x.south east)$);
}
&= \isource{\macrodsunif{2}}{}\isource{\macrodsunif{2}}{}\isource{\macrodsunif{2}}{}\isource{\macrodsunif{2}}{} \quad\netnotesmall{$\ninfcons=4$ inf.\  of $\scnetthree$},\\ \nonumber\\
\centertikz{
\node (a0) {};
\foreach \x/\y in {1/0,2/1,3/2,4/3,5/4}
{
	\node[detnode] (a\x) [right=\appplayeroffset of a\y] {$\alice$};
	\node[voidnode] (o\x) [above=\outcomevspace of a\x] {};
	\drawleg (a\x.north) -- (o\x.south);
}
\foreach \x in {1,3,4,5}
{
	\node[tensornode] (alpha\x) [below left=\appalphapos of a\x] {$\macrodsunif{4}$};
	\node[copynode] (copyalpha\x) [above=\outcomevspace of alpha\x] {};
	\drawleg (copyalpha\x) -- (alpha\x.north);
	\drawleg (copyalpha\x) -- \AnchorOneThree{a\x}{south};
}
\drawleg (copyalpha1) -- \AnchorOneThree{a2}{south};
\foreach \x in {1,2,4,5}
{
	\node[tensornode] (beta\x) [below right=\appbetapos of a\x] {$\macrodsunif{4}$};
	\node[copynode] (copybeta\x) [above=\outcomevspace of beta\x] {};
	\drawleg (copybeta\x) -- (beta\x.north);
	\drawleg (copybeta\x) -- \AnchorTwoThree{a\x}{south};
}
\drawleg (copybeta4) to [out=180,in=270] ($0.5*(a3.south west)+0.5*(a3.south east)$);
\node[psnode] (psalpha) [below left=\appalphaps of a1] {\psdiffname{4}};
\foreach \x/\y/\z in {1/1.00/0.00,
					  3/0.67/0.33,
  					  4/0.33/0.67,
  					  5/0.00/1.00}
	\drawdashedleg (copyalpha\x) to [out=170,in=340] ($\y*(psalpha.south west)+\z*(psalpha.south east)$);
\node[psnode] (psbeta) [below left=\appbetaps of a1] {\psdiffname{4}};
\foreach \x/\y/\z in {1/1.00/0.00,
					  2/0.67/0.33,
  					  4/0.33/0.67,
  					  5/0.00/1.00}
	\drawdashedleg (copybeta\x) to [out=170,in=340] ($\y*(psbeta.south west)+\z*(psbeta.south east)$);
\node[tensornode] (hv) [below right=\apphvpos of a5] {$\hv$};
\node[copynode] (copyhv) [above=\outcomevspace of hv] {};
\drawleg (hv.north) -- (copyhv);
\foreach \x in {1,2,3,4,5}
	\drawdashedleg (copyhv) to [out=180,in=340] ($0.2*(a\x.south west)+0.8*(a\x.south east)$);	
}
&=
\atargetp{1}{}{}\atargetp{2}{}{}\isource{\macrodsunif{2}}{}\quad\netnotesmall{$\ninfcons=1$ inf.\ of $\scnetfour$}\Bigg\}. \label{eq:app inf set last eq}
\end{align}
\end{subequations}
The corresponding subset of $[0,1/2]^{\times 2}$ in the parametrization of \cref{lem:sc lambdas} is denoted $\appinffeasibleparam$.
\end{definition}

Since we used a slightly different postselection compared to \cref{def:postselected inflation set}, let us briefly state that we still have an outer approximation of the feasible region:
\begin{lemma}
\label{lem:app inf feasible}
It holds that
\begin{equation}
\appfeasible \subseteq \appinffeasible,
\end{equation}
or, in the parametrization of \cref{lem:sc lambdas}, that $\appfeasibleparam\subseteq\appinffeasibleparam$.
\end{lemma}
\begin{proof}
The proof is analogous to those of \cref{lem:simple proof} and \cref{th:certifying ps}.
If we let $\dettwostrat{\alice_0}{}{}{}$ be the strategy that establishes that $\Big(\atargetp{1}{}{},\atargetp{2}{}{}\Big) \in \appfeasible$ (as in equation \eqref{eq:sc det tensor param}), then the following tensors (defined for all $\csvalalpha_1,\dots,\csvalalpha_4,\csvalbeta_1,\dots,\csvalbeta_4 \in [0,1]$, $\outputa \in \{1,2\}$, $\dsvalalpha,\dsvalbeta \in \{1,\dots,4\}$)
\begin{subequations}
\label{eq:app inf winning inf}
\begin{align}
\isource{\hv}{(\csvalalpha_1,\csvalalpha_2,\csvalalpha_3,\csvalalpha_4,\csvalbeta_1,\csvalbeta_2,\csvalbeta_3,\csvalbeta_4)} &:= \isource{\maxunif}{\csvalalpha_1}\isource{\maxunif}{\csvalalpha_2}\isource{\maxunif}{\csvalalpha_3}\isource{\maxunif}{\csvalalpha_4}
\isource{\maxunif}{\csvalbeta_1}\isource{\maxunif}{\csvalbeta_2}\isource{\maxunif}{\csvalbeta_3}\isource{\maxunif}{\csvalbeta_4},\\
\threedstrat{\alice}{\outputa}{\dsvalalpha}{\dsvalbeta}{(\csvalalpha_1,\dots,\csvalbeta_4)} &:=
\dettwostrat{\alice_0}{\outputa}{\csvalalpha_\dsvalalpha}{\csvalbeta_\dsvalbeta}
\end{align}
\end{subequations}
will verify equations \eqref{eq:app inf set first eq}-\eqref{eq:app inf set last eq}, thus proving that $\Big(\atargetp{1}{}{},\atargetp{2}{}{}\Big) \in \appinffeasible$.
\end{proof}

\subsection{Details about the implementation}
\label{sec:details about the implementation}

\paragraph{Semi-explicit linear program.}

Let us give some additional details about the formulation of the question $(\lambda_1,\lambda_2) \in \appinffeasibleparam$ as a linear program. We can assume without loss of generality that the source $\isource{\hv}{}$ appearing in equations \eqref{eq:app inf set} has as outputs the possible deterministic strategies that $\alice$ will then use. We can conveniently represent these deterministic strategies as $4\times4$ matrices with entries $1$ or $2$, corresponding to $\alice$ outputting the value of the matrix at position $\dsvalalpha,\dsvalbeta$ upon receiving the inputs $\dsvalalpha,\dsvalbeta$ from the sources $\isource{\macrodsunif{4}}{}$. Let us denote such matrices as $\appmat$ and the set thereof as $\appmatset$ (containing a priori $2^{16} = 65 536$ elements).
We can thus rewrite the integral over the outputs of $\hv$ as a finite sum where we make the first couple terms explicit:
\begin{equation}
\label{eq:app hvval}
\int\dd\hvval \isource{\hv}{\hvval} = \sum_{\appmat \in \appmatset} \isource{\hv}{\appmat} = 
\appsource{\hv}{\appfirstmatrix} + \appsource{\hv}{\appsecondmatrix} + \dots
\end{equation}
Testing for $(\lambda_1,\lambda_2) \in\appinffeasibleparam$ thus amounts to optimizing the coefficients $\isource{\hv}{\appmat}$ so that they fulfill the conditions \eqref{eq:app inf set first eq}-\eqref{eq:app inf set last eq}. 
For instance, the $1,1,1,1$ component of equation \eqref{eq:app inf set first eq} looks like:
\begin{multline}
\sum_{\appmat\in\appmatset} \isource{\hv}{\appmat}\cdot \frac{1}{4\cdot3}\cdot \frac{1}{4!} \sum_{\substack{
\dsvalalpha_1,\dsvalalpha_2,\dsvalbeta_1,\dots,\dsvalbeta_4 \in \{1,\dots,4\} \\
\dsvalalpha_1,\dsvalalpha_2 \allneq \\
\dsvalbeta_1,\dots,\dsvalbeta_4 \allneq
}}
\delta(\appmat_{\dsvalalpha_1,\dsvalbeta_1} = 1)
\delta(\appmat_{\dsvalalpha_1,\dsvalbeta_2} = 1)
\delta(\appmat_{\dsvalalpha_2,\dsvalbeta_3} = 1)
\delta(\appmat_{\dsvalalpha_2,\dsvalbeta_4} = 1)
\\
=
1 \cdot \appsource{\hv}{\appfirstmatrix}
+ \frac{3}{4} \cdot \appsource{\hv}{\appsecondmatrix} + \ \dots \ 
\overset{!}{=} \atargetp{1}{1}{1}\atargetp{1}{1}{1} = \lambda_1^2.
\end{multline}

\paragraph{Symmetry reduction.}

We can in fact reduce the variable set by noting that any two $\appmat_1,\appmat_2 \in \appmatset$ where $\appmat_2$ can be obtained by permuting the rows and columns of $\appmat_1$ will yield the same behavior within any of equations \eqref{eq:app inf set first eq}-\eqref{eq:app inf set last eq}. We can thus restrict without loss of generality the source $\isource{\hv}{}$ to only have outputs in the subset $\appmatsetreduced \subset \appmatset$ of representatives of the orbits of $\appmatset$ under the group action induced by swapping rows and columns. This subset $\appmatsetreduced$ is found numerically to have cardinality 317.
We can furthermore remove certain redundant components of the equations \eqref{eq:app inf set first eq}-\eqref{eq:app inf set last eq}: indeed, we can see from the structure of the left-hand side tensor contractions that certain permutations of the outputs of the agents gives the very same constraints. For instance, in \eqref{eq:app inf set first eq}, we have that the component $(\outputa_1,\outputa_2,\outputa_3,\outputa_4)$ will give the same constraint as that of $(\outputa_2,\outputa_1,\outputa_3,\outputa_4)$ as well as $(\outputa_1,\outputa_2,\outputa_4,\outputa_3)$ and $(\outputa_3,\outputa_4,\outputa_1,\outputa_2)$. This only gives as useful constraints the $(1,1,1,1)$, $(1,1,1,2)$, $(1,2,1,2)$, $(1,1,2,2)$, $(1,2,2,2)$ and $(2,2,2,2)$ components of \eqref{eq:app inf set first eq}. However, one of these equations is trivially verified thanks to the normalization constraint on $\isource{\hv}{}$, so that we can for instance remove the $(2,2,2,2)$ component. 
We summarize the linear program taking into account these symmetry reductions in \app~\ref{sec:app linear program explicit}.

\begin{figure}[h]
\centering
\includegraphics{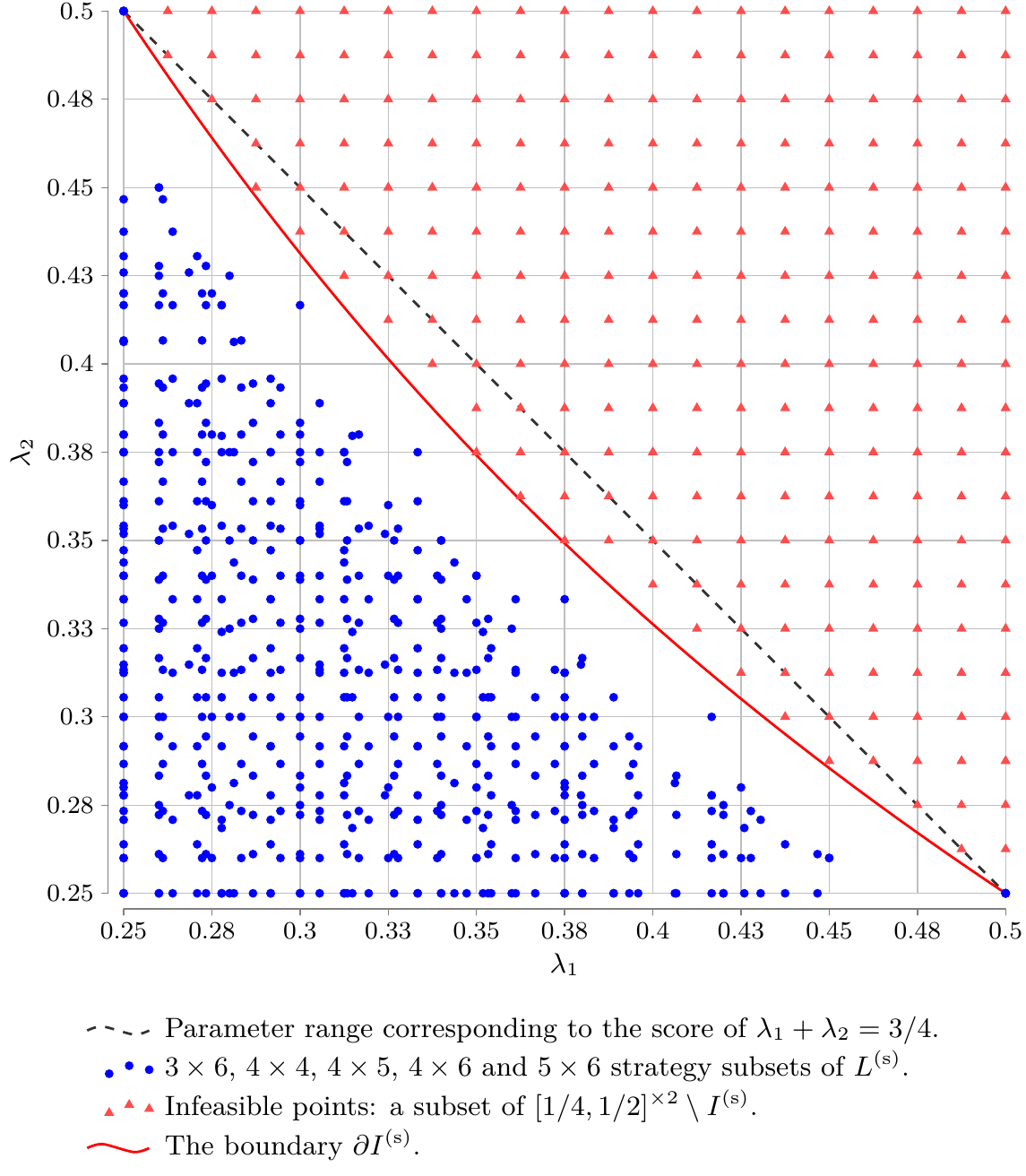}
\caption{Numerical results: we show a sample of points belonging to $\appfeasibleparam$, as well as a sample of points of $[1/4,1/2]^{\times 2} \setminus \appinffeasibleparam$, which are points that are guaranteed to lie outside of $\appfeasibleparam$ and thus correspond to infeasible correlations.}
\label{fig:main plot}
\end{figure}

\subsection{Numerical results: inner and outer approximations}

We explore the inside of $\appfeasibleparam$ by looking at deterministic strategies that process their continuous left (resp.\ right) input into a uniformly distributed discrete random variable over the set $\{1,...,n\}$ (resp.\ $\{1,...,m\}$), and refer to the corresponding set of strategies as ``the $n\times m$ strategies''. There are $2^{nm}$ such strategies without the marginal constraint \eqref{eq:marginal constraint}. To implement the marginal constraint, we first make sure that $nm$ is an even integer, and we then fill the strategy with $nm/2$ 1's and $nm/2$ 2's, for a total of $nm/2 \choose 2$ deterministic $n\times m$ strategies that verify the marginal constraint \eqref{eq:marginal constraint}. The resulting points, for reasonable values of $n$ and $m$, are shown in \cref{fig:main plot}: they populate the bottom left corner of the parameter space, corresponding to a score $\lambda_1 + \lambda_2$ less than $3/4$. The density of these points is relatively low --- it could well be at this stage that higher values of $n$ and $m$ would allow for better scores.


To certify that this is not the case, and gain further insights into the geometrical structure of the feasible region, we scan the $[1/4,1/2]^{\times 2}$ square by discretizing it into a uniform mesh, and keep in memory the points of the mesh that are found numerically to be incompatible with an inflation of the form of \cref{def:app inf set}, i.e., those points for which the corresponding linear program is found numerically to be infeasible: such points are guaranteed to lie outside of $\appfeasibleparam$. A sample of these points are shown in \cref{fig:main plot}. The infeasible region $[1/4,1/2]^{\times 2}\setminus\appinffeasibleparam$ seems to have a smooth shape: let us assume that this is the case. We can then track the boundary $\setboundary\appinffeasibleparam$ of the set $\appinffeasibleparam$ efficiently: for each fixed $\lambda_1 \in [1/4,1/2]$, we run a dichotomic search on $\lambda_2$ to find the threshold value between feasibility and infeasibility of the linear program corresponding to \cref{def:app inf set}. The resulting line (which is strictly speaking a dense mesh of data points) is shown in \cref{fig:main plot} as well.
We can already see from the data of the infeasible region, assuming that all our meshes were sufficiently fine-grained, that it looks like the best achievable score is $\lambda_1 + \lambda_2 = 3/4$ --- we draw the corresponding parameter range to guide the eye.
Additionally, in \app~\ref{sec:extended plot}, we provide an extended plot that shows the feasibility with respect to the inflation of \cref{def:app inf set} of more general distributions in the extended square $[0,1/2]^{\times 2}$: indeed, the inflation of \cref{def:app inf set} is not expected to, and does not, capture exactly the bound $\lambda_\topoindex \geq 1/4$ of \cref{lem:bound lambdas}.


\subsection{Solving the optimization}

We now present the linear program relaxation\footnote{Since \eqref{eq:app inf primal} is a feasible linear program, we know the maximum is achievable, so we replace the supremum with a maximum.} to the optimization \eqref{eq:sc det tensor} that we will use:
\begin{subequations}
\label{eq:app inf primal}
\begin{align}
\label{eq:app inf primal obj}
	\aprimal := \max_{\scalebox{0.9}{$\isource{\hv}{}, \threedstrat{\alice}{}{}{}{}{}$}}&
	\frac{1}{2}\sum_{\outputa\in\{1,2\}}\!\left(\rule{0pt}{60pt}\right.\!\!
\appprimalobjone{0}
	+\!\!\!\!
\appprimalobjtwo{0}\!\!\left.\rule{0pt}{60pt}\right)\!\!\! \\
\nonumber\\
	&\textup{s.t. } 
\appprimalconstraint{0}
	=
	\isource{\macrodsunif{2}}{}
	\isource{\macrodsunif{2}}{}. \label{eq:app inf primal constraint}
\end{align}
\end{subequations}
This is indeed a linear program according to the same logic as \cref{sec:details about the implementation}: one can assume without loss of generality that the source $\isource{\hv}{}$ tells the agent $A$ what to do, so that one can remove the optimization over $\threedstrat{\alice}{}{}{}{}$. The objective and constraint are then linear functions of $\isource{\hv}{}$.
Let us state that the linear program \eqref{eq:app inf primal} is indeed an upper bound to the original value $p^*$ of \eqref{eq:sc det tensor}:
\begin{lemma}
\label{lem:app inf opt}
It holds that
\begin{equation}
\aopt \leq \aprimal.
\end{equation}
\end{lemma}
\begin{proof}
The proof of \cref{lem:app inf feasible} can be used to show that any feasible value in the optimization problem defining $\aopt$ induces a feasible value for the optimization problem defining $\aprimal$. Indeed, consider an objective value $\tilde p$ (less than or equal to $\aopt$) that is feasible in \cref{eq:sc det tensor}, and let the probability tensor that achieves it be $\dettwostrat{\alice_0}{}{}{}$. Then, the probability tensors $\isource{\hv}{}$ and $\threedstrat{\alice}{}{}{}{}$ constructed from $\dettwostrat{\alice_0}{}{}{}$ as in equation \eqref{eq:app inf winning inf} are feasible in the linear program \eqref{eq:app inf primal} and yield the same objective value $\tilde p$.
\end{proof}

Now, let us in fact focus on the \emph{dual} problem to \eqref{eq:app inf primal}:
\begin{subequations}
\label{eq:app inf dual}
\begin{align}
\label{eq:app inf dual obj}
\adual := &\min_{z_{11},z_{12},z_{21},z_{22} \in\mathbb{R}} \frac{1}{4}\sum_{\outputa,\outputb\in\{1,2\}} z_{\outputa\outputb} \\
&\textup{s.t. } \forall \dettwostrat{\alice}{}{}{}: \nonumber
\end{align}
\vspace{-20pt}
\begin{multline}
\label{eq:app inf dual constraint}
\sum_{\substack{\outputa,\outputb\in\\\{1,2\}}}\!z_{\outputa\outputb}\!
\appprimalconstraint{1}\!\!\!\!\!\!\geq
\frac{1}{2}\!\sum_{\substack{\outputa\in\\\{1,2\}}}\!
\left(\rule{0pt}{60pt}\right.\!\!\!\!\!
\appprimalobjone{1}\!\!\!\!\!\!+\!\!\!\!\!\!\appprimalobjtwo{1}\!\!\left.\rule{0pt}{60pt}\right)\!.\!\!\!
\end{multline}
\end{subequations}
This dual problem is written in standard notation in \cref{sec:app linear program explicit}. The relation between the primal and the dual can be understood intuitively as follows. Consider an agent $\alice$ who made some choice of strategy $\threedstrat{\alice}{}{}{}{}$ and shared randomness $\isource{\hv}{}$ that are feasible in the primal \eqref{eq:app inf primal} (i.e., that verify \eqref{eq:app inf primal constraint}). This agent $\alice$ was told that she will receive an amount of money equal, in some units, to the objective of \eqref{eq:app inf primal obj}.
Now, suppose the organizer of the protocol, who does not know the choice of strategy of the agent $\alice$,
will actually run an alternative protocol where the agent $\alice$ is put in the network of \eqref{eq:app inf primal constraint}, and want to give $\alice$ some amount of money $z_{\outputa\outputb}$ (in the same units) when the two $\alice$'s outputs the outcomes $\outputa$ and $\outputb$ in this network. The organizer wants to ensure that no matter what $\alice$ is actually doing, the amounts $\{z_{\outputa\outputb}\}_{\outputa\outputb}$ are chosen fairly so that $\alice$ will receive at least the amount of money that she would have gotten in the original protocol --- this is what \eqref{eq:app inf dual constraint} captures. The organizer wants however to minimize the average cost of this alternative protocol, if performed on an honest $A$ that respects \eqref{eq:app inf primal constraint} --- this is what the minimization of \eqref{eq:app inf dual obj} captures. In particular, weak duality holds:
\begin{lemma}
\label{lem:app inf weak duality}
It holds that
\begin{equation}
\aprimal \leq \adual.
\end{equation}
\end{lemma}
\begin{proof}
Consider some $\isource{\hv}{},\threedstrat{\alice}{}{}{}{}$ that are feasible in the primal problem of $\aprimal$ (i.e., they verify \eqref{eq:app inf primal constraint} but do not necessarily achieve the maximum value of $\aprimal$), and some $\{z_{\outputa\outputb}\}_{\outputa\outputb}$ that are feasible in the dual problem (i.e., they verify \eqref{eq:app inf dual constraint} but do not necessarily achieve the minimum value of $\adual$). Then, using the non-negativity and normalization of $\isource{\hv}{}$ (recall that these constraints are always implicitly assumed when we draw such probability tensors), we can compute:
\begin{multline}
\frac{1}{2}\sum_{\outputa}\left(\rule{0pt}{60pt}\right.\appprimalobjone{0} + \appprimalobjtwo{0}\left.\rule{0pt}{60pt}\right) \\
\overset{\eqref{eq:app inf dual constraint}}{\leq}
\sum_{\outputa,\outputb} z_{\outputa\outputb} \appprimalconstraint{2}
\overset{\eqref{eq:app inf primal constraint}}{=}
\frac{1}{4}\sum_{\outputa,\outputb} z_{\outputa\outputb}.
\end{multline}
Maximizing (resp.\ minimizing) the first (resp.\ last) term of this inequality yields that indeed $\aprimal \leq \adual$.
\end{proof}

It is now easy to solve the optimization:
\begin{prop}
\label{prop:aopt}
It holds that
\begin{equation}
\aopt = \frac{3}{4}.
\end{equation}
\end{prop}
\begin{proof}
We saw in \cref{sec:intro sc} that $\aopt \geq 3/4$.
Combining \cref{lem:app inf opt,lem:app inf weak duality} together gives $\aopt \leq \adual$. The choice of
\begin{equation}
\label{eq:app dual choice}
z_{11} = z_{22} = 1,\quad z_{12} = z_{21} = \frac{1}{2}
\end{equation}
is feasible in the dual problem \eqref{eq:app inf dual}: indeed, for all $\dettwostrat{\alice}{}{}{}$, which one may enumerate since there are $2^{16}$ of them,\footnote{This follows from $\dettwostrat{\alice}{}{}{}$ having $2$ outcomes and $4\cdot 4$ different inputs. Alternatively, one may simply enumerate the inequivalent strategies belonging to $\appmatsetreduced$ --- see \cref{sec:details about the implementation}.} the condition of \eqref{eq:app inf dual constraint} is verified --- this can be checked numerically using rational arithmetic. The choice of \eqref{eq:app dual choice} corresponds to an objective value in \eqref{eq:app inf dual obj} of 3/4, which proves at first that $\adual \leq 3/4$, and then also $\aopt = \aprimal = \adual = 3/4$.
\end{proof}

\newpage
\section{The fanout inflation correspondence}
\label{sec:fanout inflation}

In this section, we show through three examples of the fact that to each postselected inflation corresponds a fanout inflation (in the usual sense of e.g.\ \cite{wolfe_inflation_2019,navascues_inflation_2020}).
This implies that the framework that we propose can be seen as a reformulation of the fanout inflation framework. Each formulation comes with its own insights: arguably, the principles underlying the fanout inflation frameworks are easy to explain from physical principles, clearly yield outer approximations, and the same principles allow to formulate quantum or non-signaling inflation schemes \cite{wolfe_quantum_2021,gisin_constraints_2020}.
The postselected inflation formulation has the benefit that establishing convergence is relatively intuitive as explained in \cref{sec:outer approx}. 
It also allows to prove straightforwardly the formal convergence of inflation in multi-network scenarios with subsets of agents using the same strategy. Now, for each multi-network scenario, given the specific proof of convergence formulated in the context of postselected inflation, a fanout inflation proof is of course readily available once one extract the correct hierarchy of fanout inflations.
In that sense, the postselected inflation framework can be seen as a convenient parametrization of fanout inflation in which the general rule to produce convergent hierarchies is clear.

\subsection{Correlated Sleeper}

Recall that in \cref{def:app inf set}, we introduced the set $\appinffeasible$. We claim that this set has an alternative definition in terms of a fanout inflation sketched in \cref{fig:connection sc graph}. 
\begin{figure}[h!]
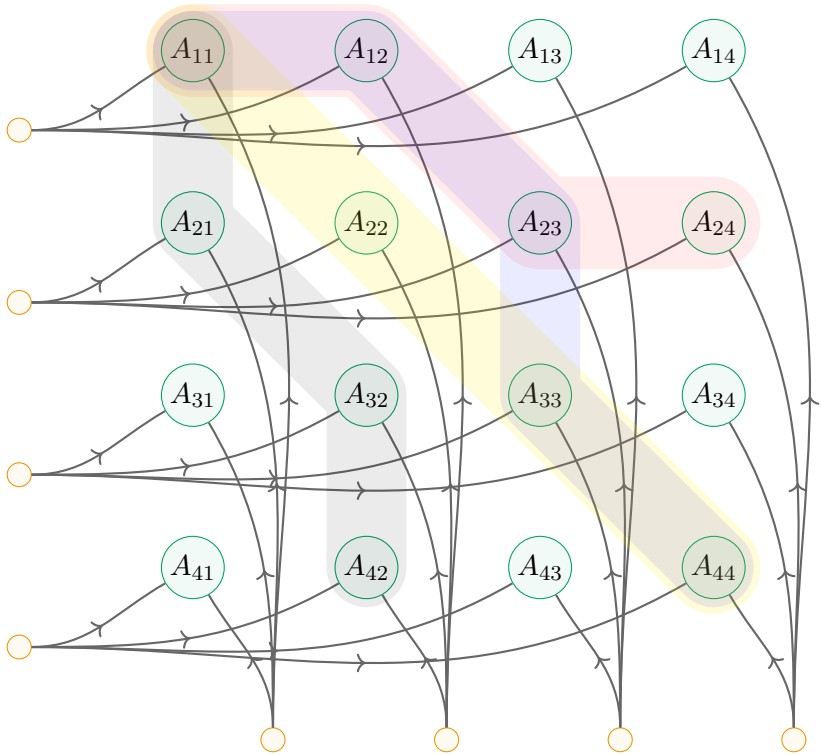

\centering
\begin{equation*}
\label{eq:connection sc graph}
\centertikz{
\foreach \x in {1,2,3,4}
{
	\foreach \y in {1,2,3,4}
		\node[agent] (a\x\y) at (\y*65pt,-\x*65pt) {$A_{\x\y}$};
}
\foreach \x in {1,2,3,4}
{
	\node[source] (alpha\x) at (0pt,-30pt-\x*65pt) {};
	\foreach \y in {1,2,3,4}
		\drawabsleg (alpha\x) to [out=0,in=210] (a\x\y);
}
\foreach \y in {1,2,3,4}
{
	\node[source] (beta\y) at (30pt+\y*65pt,-325pt) {};
	\foreach \x in {1,2,3,4}
		\drawabsleg (beta\y) to [out=90,in=300] (a\x\y);
}
\draw[opacity=0.08,cap=round,join=round,line width=35pt,draw=red] (a11.center) -- (a12.center) -- (a23.center) -- (a24.center);
\draw[opacity=0.08,cap=round,join=round,line width=30pt,draw=black] (a11.center) -- (a21.center) -- (a32.center) -- (a42.center);
\draw[opacity=0.15,cap=round,join=round,line width=35pt,draw=yellow] (a11.center) -- (a22.center) -- (a33.center) -- (a44.center);
\draw[opacity=0.08,cap=round,join=round,line width=30pt,draw=blue] (a11.center) -- (a12.center) -- (a23.center) -- (a33.center) -- (a44.center);
}
\end{equation*}
\caption{The fanout inflation graph of $\mathcal I^{(\textup{s})}_{\textup{alt}}$. We highlight with hyperedges the subsets of nodes whose marginal distribution is expressible in terms of the target distributions --- see equations \eqref{eq:inf q first condition}-\eqref{eq:inf q last condition}.}
\label{fig:connection sc graph}
\end{figure}

In the following definition, $\qinf$ is a probability distribution over $\{1,2\}^{\times 16}$, namely, over the possible outcome tuples of $16$ copies of the agent $A$. The copies of the agent $A$ are indexed in a matrix by two indices $i,j \in \{1,2,3,4\}$. We denote by $\qinf(\{A_{ij} = a_{ij}\}_{i,j=1}^4)$ the probability of an atomic event, while e.g.\ $\qinf(\{A_{11} = a_{11},A_{22} = a_{22}\})$ represents the marginal probability that the agent $A_{ii}$ outputs $a_{ii}$ for $i\in \{1,2\}$. We also let $S_4$ denote the group of permutations of $4$ elements.

\begin{lemma}
\label{lem:connection sc}
Define
\begin{subequations}
\begin{align}
\appinffeasible_{\textup{alt}} :=  \Bigg\{ \Big(\atargetp{1}{}{},\atargetp{2}{}{}\Big)
\ \Bigg|\ 
\exists &\text{ a probability distribution }\qinf \text{ over the outcomes of the agents }\centertikz{\node[agent] {$A_{ij}$};} \nonumber\\
&\text{ of the graph of }\cref{fig:connection sc graph}\text{ such that } 
\forall \sigma,\pi \in S_4, \forall \big\{a_{ij} \in \{1,2\}\big\}_{i,j=1}^{4}\st
\end{align}  
\begin{align}
q(\{A_{ij} = a_{ij}\}_{i,j=1}^4) &= q(\{A_{ij} = a_{\sigma(i)\pi(j)}\}_{i,j=1}^4), \label{eq:inf sym sc}\\
q(\{A_{11} = a_{11}, A_{12} = a_{12}, A_{23} = a_{23}, A_{24} = a_{24}\}) &= \atargetp{1}{a_{11}}{a_{12}}\atargetp{1}{a_{23}}{a_{24}}, \label{eq:inf q first condition}\\
q(\{A_{11} = a_{11}, A_{21} = a_{21}, A_{32} = a_{32}, A_{42} = a_{42}\}) &= \atargetp{2}{a_{11}}{a_{21}}\atargetp{2}{a_{32}}{a_{42}}, \\
q(\{A_{ii} = a_{ii}\}_{i=1}^4) &= \isource{\macrodsunif{2}}{a_{11}}\isource{\macrodsunif{2}}{a_{22}}\isource{\macrodsunif{2}}{a_{33}}\isource{\macrodsunif{2}}{a_{44}}, \\
q(\{A_{11} = a_{11},A_{12} = a_{12}, A_{23} = a_{23}, A_{33} = a_{33}, A_{44} = a_{44}\}) &= \atargetp{1}{a_{11}}{a_{12}}\atargetp{2}{a_{23}}{a_{33}}\isource{\macrodsunif{2}}{a_{44}} \label{eq:inf q last condition}
\Bigg\}.
\end{align}
\end{subequations}
Then, it holds that
\begin{equation}
\appinffeasible_{\textup{alt}} = \appinffeasible.
\end{equation}
\end{lemma}
\begin{proof}
$\appinffeasible_{\textup{alt}}\subseteq\appinffeasible$: for any $\big(\atargetp{1}{}{},\atargetp{2}{}{}\big) \in \appinffeasible_{\textup{alt}}$, consider the $\qinf$ distribution associated to it. Let us write $\qinf$ as a mixture of \emph{deterministic} probability distributions $\qinf^{(\lambda)}$ with convex weights $\isource{\hv}{\hvval}$:
\begin{equation}
\label{eq:correspondence qinf 1}
\qinf =: \sum_{\hvval} \isource{\hv}{\hvval} \qinf^{(\lambda)}.
\end{equation}
Define the deterministic\footnote{Indeed, the marginal of a deterministic distribution is still deterministic.} probability tensor, for all $a\in\{1,2\}$, $i,j\in\{1,2,3,4\}$, for all $\hvval$,
\begin{equation}
\label{eq:correspondence qinf 2}
\threedstrat{\alice}{a}{i}{j}{\hvval} := \qinf^{(\lambda)}(A_{ij} = a).
\end{equation}
This combination of $\threedstrat{\alice}{}{}{}{}$ and $\isource{\hv}{}$ solves the postselected inflation for $\big(\atargetp{1}{}{},\atargetp{2}{}{}\big)$: for instance, we can verify equation \eqref{eq:app inf set first eq} explicitly: for all $a_{11},a_{12},a_{23},a_{24} \in \{1,2\}$,
\begin{subequations}
\begin{align}
\label{eq:temp correspondence qinf 2}
&\centertikz{
\node (a0) {};
\foreach \x/\y/\out in {1/0/11,2/1/12,3/2/23,4/3/24}
{
	\node[detnode] (a\x) [right=\appplayeroffset of a\y] {$\alice$};
	\node[voidnode] (o\x) [above=\outcomevspace of a\x] {\indexstyle{a_{\out}}};
	\drawleg (a\x.north) -- (o\x.south);
	\node[tensornode] (beta\x) [below right=\appbetapos of a\x] {$\macrodsunif{4}$};
	\node[copynode] (copybeta\x) [above=\outcomevspace of beta\x] {};
	\drawleg (beta\x.north) -- (copybeta\x);
	\drawleg (copybeta\x) -- \AnchorTwoThree{a\x}{south};
}
\foreach \x in {1,3} 
{
	\node[tensornode] (alpha\x) [below left=\appalphapos of a\x] {$\macrodsunif{4}$};
	\node[copynode] (copyalpha\x) [above=\outcomevspace of alpha\x] {};
	\drawleg (alpha\x.north) -- (copyalpha\x);
}
\foreach \x/\y in {1/1,1/2,3/3,3/4}
	\drawleg (copyalpha\x) -- \AnchorOneThree{a\y}{south};
\node[psnode] (psbeta) [below left=\appbetaps of a1] {\psdiffname{4}};
\foreach \x/\y/\z in {1/1.00/0.00,
					  2/0.67/0.33,
  					  3/0.33/0.67,
  					  4/0.00/1.00}
	\drawdashedleg (copybeta\x) to [out=170,in=340] ($\y*(psbeta.south west)+\z*(psbeta.south east)$);
\node[psnode] (psalpha) [below left=\appalphaps of a1] {\psdiffname{2}};
\drawdashedleg (copyalpha1) to [out=170,in=340] (psalpha.south west);
\drawdashedleg (copyalpha3) to [out=170,in=340] (psalpha.south east);
\node[tensornode] (hv) [below right=\apphvpos of a4] {$\hv$};
\node[copynode] (copyhv) [above=\outcomevspace of hv] {};
\drawleg (hv.north) -- (copyhv);
\foreach \x in {1,2,3,4}
	\drawdashedleg (copyhv) to [out=180,in=340] ($0.2*(a\x.south west)+0.8*(a\x.south east)$);
} \\
= 
&\sum_{\hvval} \isource{\hv}{\hvval} \frac{1}{12\cdot 4!} 
\sum_{\substack{
\dsvalalpha_1,\dsvalalpha_2 \allneq \\
\dsvalbeta_1,\dots,\dsvalbeta_4 \allneq
}}
\threedstrat{\alice}{\outputa_{11}}{\dsvalalpha_1}{\dsvalbeta_1}{\hvval}
\threedstrat{\alice}{\outputa_{12}}{\dsvalalpha_1}{\dsvalbeta_2}{\hvval}
\threedstrat{\alice}{\outputa_{23}}{\dsvalalpha_2}{\dsvalbeta_3}{\hvval}
\threedstrat{\alice}{\outputa_{24}}{\dsvalalpha_2}{\dsvalbeta_4}{\hvval} \label{eq:temp qinf computation}\\
\overset{\textup{eqs. }\eqref{eq:correspondence qinf 1},\eqref{eq:correspondence qinf 2}}{=} 
&\frac{1}{12\cdot 4!} 
\sum_{\substack{
\dsvalalpha_1,\dsvalalpha_2 \allneq \\
\dsvalbeta_1,\dots,\dsvalbeta_4 \allneq
}}
\qinf(\{A_{i_1j_1} = a_{11}, A_{i_1j_2} = a_{12}, A_{i_2j_3} = a_{23}, A_{i_2j_4} = a_{24}\}) \\
\overset{\sigma(1) := i_1, \dots, \pi(4) := j_4}{=}
&\frac{1}{4!\cdot4!} 
\sum_{\sigma,\pi\in S_4}
\qinf(\{A_{\sigma(1)\pi(1)} = a_{11}, A_{\sigma(1)\pi(2)} = a_{12}, A_{\sigma(2)\pi(3)} = a_{23}, A_{\sigma(2)\pi(4)} = a_{24}\}) \\
\overset{\textup{eq. }\eqref{eq:inf sym sc}}{=} &\qinf(\{A_{11} = a_{11}, A_{12} = a_{12}, A_{23} = a_{23}, A_{24} = a_{24}\}) \\
\overset{\textup{eq. }\eqref{eq:inf q first condition}}{=} &\atargetp{1}{a_{11}}{a_{12}}\atargetp{1}{a_{23}}{a_{24}}.
\end{align}
\end{subequations}
The other cases are proven analogously, so that indeed $\big(\atargetp{1}{}{},\atargetp{2}{}{}\big) \in \appinffeasible$.

$\appinffeasible_{\textup{alt}} \supseteq \appinffeasible$: for any $\big(\atargetp{1}{}{},\atargetp{2}{}{}\big) \in \appinffeasible$, consider the probability tensors $\threedstrat{\alice}{}{}{}{}$ and $\isource{\hv}{}$ (the latter can be assumed to have finitely many outputs $\hvval$ without loss of generality --- see \cref{sec:details about the implementation}) associated to it.
Define the distribution $\qinf$ through, for all $\{a_{ij} \in \{1,2\}\}_{i,j=1}^4$,
\begin{equation}
\qinf(\{A_{ij} = a_{ij}\}_{i,j=1}^4) := \sum_{\hvval} \isource{\hv}{\hvval} \frac{1}{4!\cdot4!} \sum_{\sigma,\pi\in S_4} \prod_{i,j=1}^4 \threedstrat[2]{\alice}{a_{\sigma(i)\pi(j)}}{i}{j}{\hvval}.
\end{equation}
This $\qinf$ clearly verifies the symmetry condition \eqref{eq:inf sym sc}. It furthermore verifies the desired marginal constraints: for instance, for \eqref{eq:inf q first condition}, we have
\begin{subequations}
\begin{align}
&\qinf(\{A_{11} = a_{11}, A_{12} = a_{12}, A_{23} = a_{23}, A_{24} = a_{24}\}) \\
= &\sum_{\hvval}\isource{\hv}{\hvval} \frac{1}{4!\cdot4!}\sum_{\sigma,\pi\in S_4} \threedstrat[1]{\alice}{a_{11}}{\sigma(1)}{\pi(1)}{\hvval}
\threedstrat[1]{\alice}{a_{12}}{\sigma(1)}{\pi(2)}{\hvval}
\threedstrat[1]{\alice}{a_{23}}{\sigma(2)}{\pi(3)}{\hvval}
\threedstrat[1]{\alice}{a_{24}}{\sigma(2)}{\pi(4)}{\hvval} \\
\overset{\textup{eqs. }\eqref{eq:temp qinf computation},\eqref{eq:temp correspondence qinf 2},\eqref{eq:app inf set first eq}}{=} &\atargetp{1}{a_{11}}{a_{12}}\atargetp{1}{a_{23}}{a_{24}}.
\end{align}
\end{subequations}
The other marginal constraints are verified analogously, and indeed we see that $\big(\atargetp{1}{}{},\atargetp{2}{}{}\big) \in \appinffeasible$.
\end{proof}


\begin{figure}[h!]
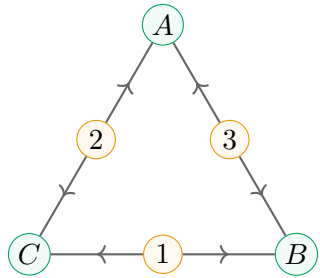

\centering
\begin{equation*}
\centertikz{
\node[agent] (a) at (0.5*100pt,0.866*100pt) {$A$};
\node[agent] (b) at (100pt,0)  {$B$};
\node[agent] (c) at (0,0)  {$C$};
\node[source] (alpha) at ($0.5*(b) + 0.5*(c)$) {$1$};
\node[source] (beta) at ($0.5*(a)+0.5*(c)$) {$2$};
\node[source] (gamma) at ($0.5*(a)+0.5*(b)$) {$3$};
\drawabsleg (beta) -- (a);
\drawabsleg (gamma) -- (a);
\drawabsleg (alpha) -- (b);
\drawabsleg (gamma) -- (b);
\drawabsleg (alpha) -- (c);
\drawabsleg (beta) -- (c);
}
\end{equation*}
\caption{The triangle network.}
\label{fig:triangle config}
\end{figure}

\subsection{Triangle network: three strategies}
\label{sec:triangle network three strats}

Let us give two additional examples. First off, consider the triangle network, sketched in \cref{fig:triangle config}:
\begin{subequations}
\begin{align}
\network^{\textup{(tr)}} &= (\pcount = 3,\npcount = 3,\scount = 3,\pmap,\cmap), \\
\forall \npindex \in \{1,2,3\},\ \pmap(\npindex) &= \npindex, &\netnote{three strategies}\\
\cmap(1) &= (2,3),\ \cmap(2) = (3,1),\ \cmap(3) = (1,2).
\end{align}
\end{subequations}

We can formulate a postselected inflation with a more minimal postselection than that of \cref{def:postselected inflation set} following the argument of \cref{lem:simple proof} (this type of postselected inflation is generic in the context of a single-network scenario with all agents using distinct strategies): let
\begin{subequations}
\begin{equation}
\mathcal I^{\textup{(tr)}}(\ninf,\ninfcons=2) := \Bigg\{ 
\targetp{}{}{} \Bigg| \exists \threedstrat{\alice}{}{}{}{},\threedstrat{\bob}{}{}{}{},\threedstrat{\charlie}{}{}{}{},\isource{\hv}{} \textup{ s.t. }
\end{equation}
\begin{equation}
\centertikz{
\node[detnode] (a1) {$A$};
\node[detnode] (b1) [right=20pt of a1] {$B$};
\node[detnode] (c1) [right=20pt of b1] {$C$};
\node[detnode] (a2) [right=20pt of c1] {$A$};
\node[detnode] (b2) [right=20pt of a2] {$B$};
\node[detnode] (c2) [right=20pt of b2] {$C$};
\foreach \x in {a,b,c}
{
	\foreach \y in {1,2}
	{
		\node[voidnode] (o\x\y) [above=\outcomevspace of \x\y] {};
		\drawleg (\x\y.north) -- (o\x\y.south);
	}
}
\node[tensornode] (gamma1) [below=50pt of a1] {$\dsunif$};
\node[tensornode] (beta1) [below=100pt of b1] {$\dsunif$};
\node[tensornode] (alpha1) [below=150pt of c1] {$\dsunif$};
\node[tensornode] (gamma2) [below=50pt of a2] {$\dsunif$};
\node[tensornode] (beta2) [below=100pt of b2] {$\dsunif$};
\node[tensornode] (alpha2) [below=150pt of c2] {$\dsunif$};
\foreach \x in {alpha,beta,gamma}
{
	\foreach \y in {1,2}
	{
		\node[copynode] (copy\x\y) [above=\outcomevspace of \x\y] {};
		\drawleg (\x\y.north) -- (copy\x\y);
	}
}
\foreach \x in {1,2}
{
	\drawleg (copyalpha\x) -- \AnchorTwoThree{b\x}{south};
	\drawleg (copyalpha\x) -- \AnchorOneThree{c\x}{south};
	\drawleg (copybeta\x) -- \AnchorTwoThree{c\x}{south};
	\drawleg (copybeta\x) -- \AnchorOneThree{a\x}{south};
	\drawleg (copygamma\x) -- \AnchorTwoThree{a\x}{south};
	\drawleg (copygamma\x) -- \AnchorOneThree{b\x}{south};
}
%
%
\foreach \x in {alpha,beta,gamma}
{
	\node[psnode] (ps\x) [above left=20pt and 40pt of \x1] {\psdiffname{2}};
	\foreach \y/\z in {1/west,2/east}
		\drawdashedleg (copy\x\y) to [out=190,in=340] (ps\x.south \z);
}
%
%
\node[tensornode] (hv) [below right=20pt and 30pt of c2] {$\hv$};
\node[copynode] (copyhv) [above=\outcomevspace of hv] {};
\drawleg (hv.north) -- (copyhv);
\foreach \x in {a,b,c}
{
	\foreach \y in {1,2}
		\drawdashedleg (copyhv) to [out=180,in=340] ($0.2*(\x\y.south west)+0.8*(\x\y.south east)$);
}
}
\quad=\quad 
\targetp{}{}{}\targetp{}{}{}\Bigg\}. \label{eq:psinf triangle condition}
\end{equation}
\end{subequations}
%
%
The above characterization coincides with the fanout inflation of the triangle network as in 
\cite{navascues_inflation_2020}. Leaving the number of outcomes implicit, let, for any $\ninf \geq 2$,
\begin{subequations}
\begin{align}
\mathcal I^{(\textup{tr})}_{\textup{alt}}(\ninf,\ninfcons=2) := \Bigg\{ \targetp{}{}{} \Bigg| &\exists\text{ a probability distribution }\qinf\textup{ over the outcomes of the agents }\{A_{ij}\}_{i,j=1}^\ninf,\nonumber \\ &\{B_{kl}\}_{k,l=1}^\ninf,\{C_{pq}\}_{p,q=1}^\ninf \textup{ such that }\forall \{a_{ij}\}_{i,j=1}^\ninf,\{b_{kl}\}_{k,l=1}^\ninf,\{c_{pq}\}_{p,q=1}^\ninf,
\end{align}
\vspace{-30pt}
\begin{multline}
\forall \sigma,\pi,\tau\in S_\ninf\st 
\qinf(
\{A_{ij} = a_{ij},
B_{kl} = b_{kl},
C_{pq} = c_{pq}\}_{i,j,k,l,p,q=1}^\ninf
) \\
= \qinf(
\{A_{ij} = a_{\tau(i)\sigma(j)},
B_{kl} = b_{\sigma(k)\pi(l)},
C_{pq} = c_{\pi(p)\tau(q)}\}_{i,j,k,l,p,q=1}^\ninf
)
\label{eq:fanout inf triangle 1}
\end{multline}
\vspace{-10pt}
\begin{equation}
\hspace{4.2cm} \textup{ and }\ \qinf(\{A_{ii} = a_{ii},B_{ii} = b_{ii}, C_{ii} = c_{ii}\}_{i=1,2}) = \targetp[1]{a_{11}}{b_{11}}{c_{11}}\targetp[1]{a_{22}}{b_{22}}{c_{22}} \Bigg\}.
\label{eq:fanout inf triangle 2}
\end{equation}
\end{subequations}
The corresponding fanout inflation graph is shown in \cref{fig:fanout inflation triangle} for the case of $\ninf = 2$.

\begin{figure}[h!]
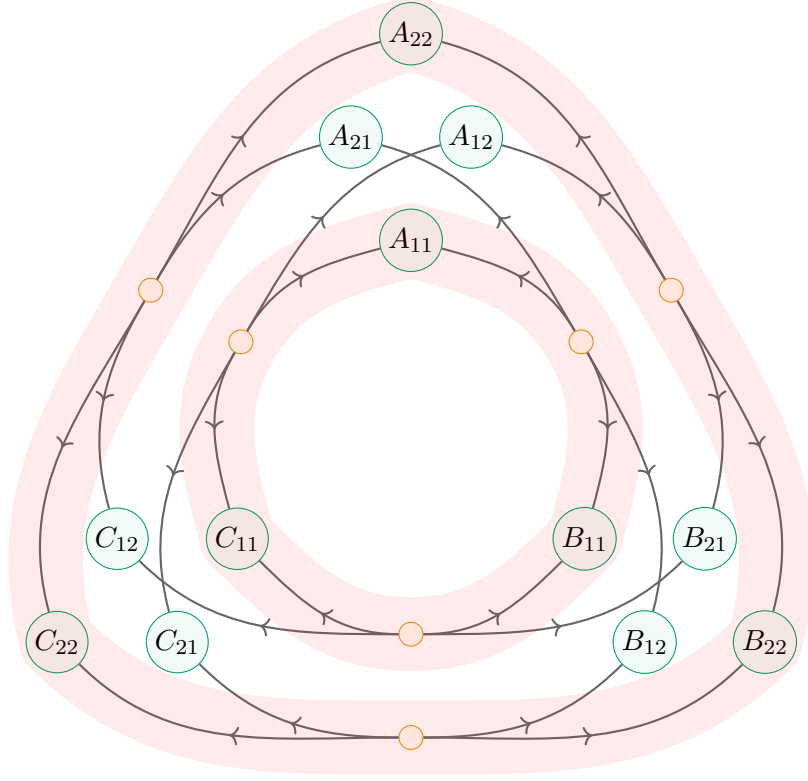

\centering
\begin{equation*}
\centertikz{
%
%
\foreach \x/\i in {0/1,1/2}
{
	\foreach \y/\j in {0/1,1/2}
	{
		\node[agent] (a\i\j) at ($(0.5*\TrBigD,0.866*\TrBigD) + \x*(-0.5*\TrSmallD,0.866*\TrSmallD) +  \y*(0.5*\TrSmallD,0.866*\TrSmallD)$) {$A_{\i\j}$};
	}
}
%
%
\foreach \x/\i in {0/1,1/2}
{
	\foreach \y/\j in {0/1,1/2}
	{
		\node[agent] (b\i\j) at ($(\TrBigD,0) + \x*(\TrSmallD,0) +  \y*(0.5*\TrSmallD,-0.866*\TrSmallD)$) {$B_{\i\j}$};
	}
}
%
%
\foreach \x/\i in {0/1,1/2}
{
	\foreach \y/\j in {0/1,1/2}
	{
		\node[agent] (c\i\j) at ($(0,0) + \y*(-\TrSmallD,0) +  \x*(-0.5*\TrSmallD,-0.866*\TrSmallD)$) {$C_{\i\j}$};
	}
}
%
%
\foreach \i in {1,2}
{
	\node[source] (alpha\i) at ($0.5*(b1\i) + 0.5*(c\i1) - 0.8*(0,\TrSmallD)$) {};
	\node[source] (beta\i) at ($0.5*(c1\i) + 0.5*(a\i1) + 0.8*(-0.866*\TrSmallD,0.5*\TrSmallD)$) {};
	\node[source] (gamma\i) at ($0.5*(a1\i) + 0.5*(b\i1) + 0.8*(0.866*\TrSmallD,0.5*\TrSmallD)$) {};
	\foreach \j in {1,2}
	{
		\drawabsleg (alpha\i) to [out=0,in=225] (b\j\i);
		\drawabsleg (alpha\i) to [out=180,in=315] (c\i\j);
		\drawabsleg (beta\i) to [out=60,in=195] (a\i\j);
		\drawabsleg (beta\i) to [out=240,in=105] (c\j\i);
		\drawabsleg (gamma\i) to [out=120,in=345] (a\j\i);
		\drawabsleg (gamma\i) to [out=300,in=75] (b\i\j);
	}
}
\draw[opacity=0.08,cap=round,join=round,line width=28pt,draw=red] (a11.center) to [out=345,in=120] (gamma1.center) to [out=300,in=75] (b11.center) to [out=225,in=0] (alpha1.center) to [out=180,in=315] (c11.center) to [out=105,in=240] (beta1.center) to [out=60,in=195] cycle;
\draw[opacity=0.08,cap=round,join=round,line width=28pt,draw=red] (a22.center) to [out=345,in=120] (gamma2.center) to [out=300,in=75] (b22.center) to [out=225,in=0] (alpha2.center) to [out=180,in=315] (c22.center) to [out=105,in=240] (beta2.center) to [out=60,in=195] cycle;
}
\end{equation*}
\caption{The fanout inflation graph for of $\mathcal I^{(\textup{tr})}_{\textup{alt}}(\ninf,\ninfcons=2)$ for $\ninf = 2$. The red hyperedge represents the known marginal of equation \eqref{eq:fanout inf triangle 2}.}
\label{fig:fanout inflation triangle}
\end{figure}

\newpage

Let us briefly sketch the equality proof (which can easily be generalized to arbitrary $\ninfcons$).
\begin{lemma}
It holds that $\mathcal I^{(\textup{tr})}_{\textup{alt}}(\ninf,\ninfcons=2) = \mathcal I^{(\textup{tr})}(\ninf,\ninfcons=2)$.
\end{lemma}
\begin{proof}
Given $\targetp{}{}{} \in \mathcal I^{\textup{(tr)}}_{\textup{alt}}(\ninf,\ninfcons=2)$ and the $\qinf$ distribution associated to it, write $\qinf$ as a mixture of deterministic distributions $\qinf^{(\lambda)}$ with convex weights $\isource{\hv}{\hvval}$:
\begin{equation}
\qinf =: \sum_{\hvval} \isource{\hv}{\hvval}\qinf^{(\lambda)}.
\end{equation}
Then, define the deterministic tensors, for all $i,j,k,l,p,q\in\{1,\dots,\ninf\}$ and $\outputa,\outputb,\outputc,\hvval$:
\begin{equation}
\threedstrat{\alice}{\outputa}{i}{j}{\hvval} := \qinf^{(\lambda)}(A_{ij} = \outputa), \quad
\threedstrat{\bob}{\outputb}{k}{l}{\hvval} := \qinf^{(\lambda)}(B_{kl} = \outputb), \quad
\threedstrat{\charlie}{\outputc}{p}{q}{\hvval} := \qinf^{(\lambda)}(C_{pq} = \outputc).
\end{equation}
The probability tensors $\threedstrat{\alice}{}{}{}{},\threedstrat{\bob}{}{}{}{},\threedstrat{\charlie}{}{}{}{},\isource{\hv}{}$ verify equation \eqref{eq:psinf triangle condition} thanks to equations \eqref{eq:fanout inf triangle 1} and \eqref{eq:fanout inf triangle 2} --- the manipulations are analogous to those of the proof of \cref{lem:connection sc} --- so that $\targetp{}{}{}\in\mathcal I^{\textup{(tr)}}(\ninf,\ninfcons=2)$.

Conversely, given $\targetp{}{}{} \in \mathcal I^{\textup{(tr)}}(\ninf,\ninfcons=2)$ and the tensors $\threedstrat{\alice}{}{}{}{},\threedstrat{\bob}{}{}{}{},\threedstrat{\charlie}{}{}{}{},\isource{\hv}{}$ associated to it, define the probability distribution $\qinf$ such that, for all $ \{a_{ij}\}_{i,j=1}^\ninf,\{b_{kl}\}_{k,l=1}^\ninf,\{c_{pq}\}_{p,q=1}^\ninf$,
\begin{multline}
\qinf(
\{A_{ij} = a_{ij},
B_{kl} = b_{kl},
C_{pq} = c_{pq}\}_{i,j,k,l,p,q=1}^\ninf
) 
:= \\
\sum_{\hvval} \isource{\hv}{\hvval} 
\frac{1}{n!^3} \sum_{\sigma,\pi,\tau\in S_\ninf} 
\left(\prod_{i,j=1}^\ninf 
\threedstrat[2]{\alice}{a_{\tau(i)\sigma(j)}}{i}{j}{\hvval}
\right)
\left(\prod_{k,l=1}^\ninf 
\threedstrat[2]{\bob}{b_{\sigma(k)\pi(l)}}{k}{l}{\hvval}
\right)
\left(\prod_{p,q=1}^\ninf 
\threedstrat[2]{\charlie}{c_{\pi(p)\tau(q)}}{p}{q}{\hvval}
\right).
\end{multline}
This distribution $\qinf$ verifies the symmetry condition \eqref{eq:fanout inf triangle 1} by design, and verifies \eqref{eq:fanout inf triangle 2} thanks to \eqref{eq:psinf triangle condition} after some manipulations similar to those of \cref{lem:connection sc}. This implies $\targetp{}{}{}\in\mathcal I^{(\textup{tr})}_{\textup{alt}}(\ninf,\ninfcons=2)$.
\end{proof}

Notice that as $\ninf$ grows, the outer approximation $\mathcal I^{(\textup{tr})}(\ninf,\ninfcons=2)$ will converge to the set $\outputdistribs{\network^{\textup{(tr)}}}$: this is a special case of the proof of \cite{navascues_inflation_2020}. It can also be seen intuitively from equation \eqref{eq:psinf triangle condition} together with the arguments of \cref{sec:outer approx}, and formally from \cref{th:convergence} after adapting the proof to the more minimal postselection of equation \eqref{eq:psinf triangle condition}.

\subsection{Triangle network: one strategy}
\label{sec:triangle network one strat}

Now, what about the case of the triangle network, but with only one strategy for the agents? This defines the network
\begin{subequations}
\begin{align}
\bar\network^{\textup{(tr)}} &= (\pcount = 1,\npcount = 3,\scount = 3,\pmap,\cmap), \\
\forall \npindex \in \{1,2,3\},\ \pmap(\npindex) &= 1, &\netnote{one strategy}\\
\cmap(1) &= (2,3),\ \cmap(2) = (3,1),\ \cmap(3) = (1,2),
\end{align}
\end{subequations}
which we sketch in \cref{fig:triangle config one strat}.
\begin{figure}[h!]
\centering
\begin{equation*}
\centertikz{
\node[agent] (a) at (0.5*100pt,0.866*100pt) {$A$};
\node[agent] (b) at (100pt,0)  {$A$};
\node[agent] (c) at (0,0)  {$A$};
\node[source] (alpha) at ($0.5*(b) + 0.5*(c)$) {$1$};
\node[source] (beta) at ($0.5*(a)+0.5*(c)$) {$2$};
\node[source] (gamma) at ($0.5*(a)+0.5*(b)$) {$3$};
\drawabsleg (beta) -- (a);
\drawabsleg (gamma) -- (a);
\drawabsleg (alpha) -- (b);
\drawabsleg (gamma) -- (b);
\drawabsleg (alpha) -- (c);
\drawabsleg (beta) -- (c);
}
\end{equation*}
\caption{The network $\bar\network^{\textup{(tr)}}$, corresponding to the triangle network with only one strategy for the three agents.}
\label{fig:triangle config one strat}
\end{figure}

A set of outer approximations that converges to the set $\outputdistribs{\bar\network^{\textup{(tr)}}}$ would be those defined in \cref{def:postselected inflation set} --- in this case, it does not seem to be possible to reduce the amount of postselection. Let us choose $\ninfcons = 2$; we then need $\ninf \geq 6$ to have a feasible postselection on the $\ninfcons \cdot \scount = 6$ sources present in the postselected inflation. The general construction of \cref{def:postselected inflation set} directly yields the following outer approximation:
\begin{multline}
\bar{\mathcal I}^{\textup{(tr)}}(\ninf=6,\ninfcons=2) := \Bigg\{ 
\targetp{}{}{} \Bigg| \exists \threedstrat{\alice}{}{}{}{},\isource{\hv}{} \textup{ s.t. } \\
\centertikz{
\node[detnode] (a1) {$A$};
\node[detnode] (b1) [right=20pt of a1] {$A$};
\node[detnode] (c1) [right=20pt of b1] {$A$};
\node[detnode] (a2) [right=20pt of c1] {$A$};
\node[detnode] (b2) [right=20pt of a2] {$A$};
\node[detnode] (c2) [right=20pt of b2] {$A$};
\foreach \x in {a,b,c}
{
	\foreach \y in {1,2}
	{
		\node[voidnode] (o\x\y) [above=\outcomevspace of \x\y] {};
		\drawleg (\x\y.north) -- (o\x\y.south);
	}
}
\node[tensornode] (gamma1) [below=75pt of a1] {$\macrodsunif{6}$};
\node[tensornode] (beta1) [below=75pt of b1] {$\macrodsunif{6}$};
\node[tensornode] (alpha1) [below=75pt of c1] {$\macrodsunif{6}$};
\node[tensornode] (gamma2) [below=75pt of a2] {$\macrodsunif{6}$};
\node[tensornode] (beta2) [below=75pt of b2] {$\macrodsunif{6}$};
\node[tensornode] (alpha2) [below=75pt of c2] {$\macrodsunif{6}$};
\foreach \x in {alpha,beta,gamma}
{
	\foreach \y in {1,2}
	{
		\node[copynode] (copy\x\y) [above=\outcomevspace of \x\y] {};
		\drawleg (\x\y.north) -- (copy\x\y);
	}
}
\foreach \x in {1,2}
{
	\drawleg (copyalpha\x) -- \AnchorTwoThree{b\x}{south};
	\drawleg (copyalpha\x) -- \AnchorOneThree{c\x}{south};
	\drawleg (copybeta\x) -- \AnchorTwoThree{c\x}{south};
	\drawleg (copybeta\x) -- \AnchorOneThree{a\x}{south};
	\drawleg (copygamma\x) -- \AnchorTwoThree{a\x}{south};
	\drawleg (copygamma\x) -- \AnchorOneThree{b\x}{south};
}
%
%
\node[psnode] (ps) [above left=45pt and 20pt of gamma1] {\psdiffname{6}};
\foreach \x/\yl/\yr in {gamma1/1.0/0.0,beta1/0.8/0.2,alpha1/0.6/0.4,gamma2/0.4/0.6,beta2/0.2/0.8,alpha2/0.0/1.0}
	\drawdashedleg (copy\x) to [out=170,in=330] ($\yl*(ps.south west)+\yr*(ps.south east)$);
%
%
\node[tensornode] (hv) [below right=20pt and 30pt of c2] {$\hv$};
\node[copynode] (copyhv) [above=\outcomevspace of hv] {};
\drawleg (hv.north) -- (copyhv);
\foreach \x in {a,b,c}
{
	\foreach \y in {1,2}
		\drawdashedleg (copyhv) to [out=180,in=340] ($0.2*(\x\y.south west)+0.8*(\x\y.south east)$);
}
}
\ =\ 
\targetp{}{}{}\targetp{}{}{}\Bigg\}. \label{eq:psinf triangle one strat condition}
\end{multline}
Notice that the agent $A$ will never see the two inputs from the $\isource{\macrodsunif{6}}{}$ sources being equal in the network of equation \eqref{eq:psinf triangle one strat condition}: computationally speaking, one can thus safely assume that the agent output a fixed, default outcome in this case, or even avoid storing that information entirely. This is apparent in the fanout inflation formulation.

\begin{figure}[h!]
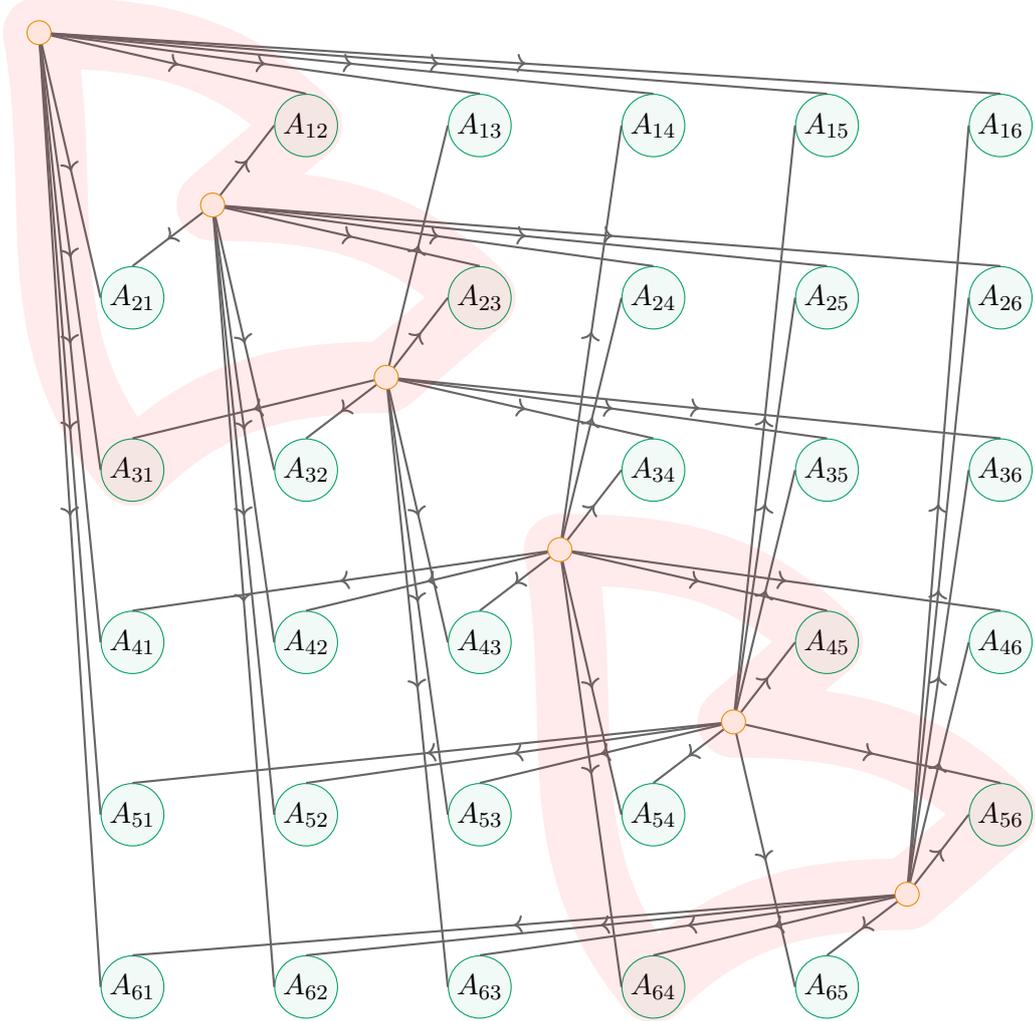

\centering
\begin{equation*}
\centertikz{
\foreach \x in {1,2,3,4,5,6}
{
	\foreach \y in {1,2,3,4,5,6}
	{
		\ifnum\x=\y
		\relax
		\else
		\node[agent] (a\x\y) at (\y*65pt,-\x*65pt) {$A_{\x\y}$};
		\fi
	}
}
\foreach \x in {1,2,3,4,5,6}
{
	\node[source] (alpha\x) at (-65pt+\x*65pt+30pt,65pt-\x*65pt-30pt) {};
	\foreach \y in {1,2,3,4,5,6}
	{ 
		\ifnum\x=\y
		\relax
		\else
		\drawabsleg (alpha\x) -- (a\x\y.north);
		\drawabsleg (alpha\x) -- (a\y\x.west);
		\fi
	}
}
\draw[join=round,cap=round,line width=27pt,draw=red,opacity=0.08] (a12.center) -- (alpha2.center) to [out=0,in=135] (a23.center) -- (alpha3.center) to [out=180,in=45] (a31.center) to [out=130,in=280] (alpha1.center) to [out=0,in=135] cycle;
\draw[join=round,cap=round,line width=27pt,draw=red,opacity=0.08] (a45.center) -- (alpha5.center) to [out=0,in=135] (a56.center) -- (alpha6.center) to [out=180,in=45] (a64.center) to [out=130,in=280] (alpha4.center) to [out=0,in=135] cycle;
}
\end{equation*}
\caption{The fanout inflation graph of $\bar{\mathcal I}^{(\textup{tr})}_{\textup{alt}}(\ninf=6,\ninfcons=2)$. The red hyperedge represents the known marginal of equation \eqref{eq:inf triangle one strat q first condition}.}
\label{fig:connection triangle one strat graph}
\end{figure}

The corresponding fanout inflation, whose graph is shown in \cref{fig:connection triangle one strat graph}, is the following:
\begin{subequations}
\label{eq:def fanout inf triangle one strat}
\begin{align}
\bar{\mathcal I}^{(\textup{tr})}_{\textup{alt}}(\ninf=6,\ninfcons=2)
:=  
\Bigg\{ \targetp{}{}{}
\Bigg|\exists 
&\text{ a probability distribution }\qinf \text{ over the outcomes of the agents }\centertikz{\node[agent] {$A_{ij}$};}\nonumber\\
&\text{ of the graph of \cref{fig:connection triangle one strat graph} such that }
\forall \{a_{ij}\}_{i\neq j \in \{1,\dots,6\}}\st \\
\forall \sigma\in S_6\st &q(\{A_{ij} = a_{ij}\}_{i\neq j \in \{1,\dots,6\}}) = q(\{A_{ij} = a_{\sigma(i)\sigma(j)}\}_{i\neq j \in \{1,\dots,6\}}), \label{eq:inf triangle one strat sym sc}\\
q(\{A_{12}\!=\!a_{12}, A_{23}\!=\!a_{23}, A_{31}\!&=\!a_{31}, A_{45}\!=\!a_{45}, A_{56}\!=\!a_{56}, A_{64}\!=\!a_{64}\})\!=\!\! \targetp[1]{a_{12}}{a_{23}}{a_{31}}\!\!\targetp[1]{a_{45}}{a_{56}}{a_{64}}\! \label{eq:inf triangle one strat q first condition}\Bigg\}.
\end{align}
\end{subequations}

The equality between the two characterizations works as in the previous two examples:
\begin{lemma}
It holds that $\bar{\mathcal I}^{(\textup{tr})}_{\textup{alt}}(\ninf=6,\ninfcons=2) = \bar{\mathcal I}^{(\textup{tr})}(\ninf=6,\ninfcons=2)$.
\end{lemma}
\begin{proof}
Given $\targetp{}{}{}\in\bar{\mathcal I}^{(\textup{tr})}_{\textup{alt}}(\ninf=6,\ninfcons=2)$ and the associated $\qinf$, decompose $\qinf$ as a mixture of deterministic behaviors
\begin{equation}
\qinf =: \sum_{\hvval} \isource{\hv}{\hvval} \qinf^{(\lambda)}
\end{equation}
and define the deterministic tensor $\threedstrat{\alice}{}{}{}{}$ for all $i\neq j\in\{1,\dots,6\}$ and $\outputa,\hvval$ through\footnote{The behavior whenever $i = j$ is irrelevant.}
\begin{equation}
\threedstrat{\alice}{\outputa}{i}{j}{\hvval} := \qinf^{(\lambda)}(A_{ij} = \outputa).
\end{equation}
This will show that $\targetp{}{}{}\in\bar{\mathcal I}^{\textup{(tr)}}(\ninf=6,\ninfcons=2)$.

Conversely, given $\targetp{}{}{} \in\bar{\mathcal I}^{\textup{(tr)}}(\ninf=6,\ninfcons=2)$ and the associated $\threedstrat{\alice}{}{}{}{},\isource{\hv}{}$, define the probability distribution $\qinf$ for all $\{a_{ij}\}_{i\neq j \in \{1,\dots,6\}}$ through
\begin{equation}
\qinf(\{A_{ij} = a_{ij}\}_{i\neq j \in \{1,\dots,6\}}) :=
\sum_{\hvval} \isource{\hv}{\hvval} \frac{1}{6!} \sum_{\sigma\in S_6}
\prod_{i\neq j \in \{1,\dots,6\}}
\threedstrat[2]{\alice}{\outputa_{\sigma(i)\sigma(j)}}{i}{j}{\hvval}.
\end{equation}
This will show that $\targetp{}{}{}\in\bar{\mathcal I}^{\textup{(tr)}}_{\textup{alt}}(\ninf=6,\ninfcons=2)$.
\end{proof}

The above set $\bar{\mathcal I}^{(\textup{tr})}(\ninf=6,\ninfcons=2)$, whose generalization to arbitrary $\ninf$ is clear from \cref{def:postselected inflation set}, would indeed converge to the set $\mathcal L\big( \bar{\network}^{(\textup{tr})}\big)$ thanks to \cref{th:convergence}. However, at the finite order of $\ninf=6$ that we are considering here, we are missing one constraint: this is the constraint that the postselected inflation of equation \eqref{eq:psinf triangle one strat condition} should additionally verify
\begin{equation}
\centertikz{
\foreach \x in {0,1,2}
{	
	\node[detnode] (a\x) at (\x*60pt,0pt) {$A$};
	\node[voidnode] (o\x) [above=\outcomevspace of a\x] {};
	\drawleg (a\x.north) -- (o\x.south);
}
\foreach \x/\y in {00/0,01/1,10/2,11/3,20/4,21/5}
{
	\node[tensornode] (s\x) at (-15pt+\y*30pt,-60pt) {$\macrodsunif{6}$};
	\node[copynode] (copys\x) [above=\outcomevspace of s\x]  {};
	\drawleg (s\x.north) -- (copys\x);
}
\foreach \x in {0,1,2}
{
	\drawleg (copys\x0) -- \AnchorOneThree{a\x}{south};
	\drawleg (copys\x1) -- \AnchorTwoThree{a\x}{south};
}
%
%
\node[psnode] (ps) [above left=45pt and 20pt of s00] {\psdiffname{6}};
\foreach \x/\yl/\yr in {00/1.0/0.0,01/0.8/0.2,10/0.6/0.4,11/0.4/0.6,20/0.2/0.8,21/0.0/1.0}
	\drawdashedleg (copys\x) to [out=170,in=330] ($\yl*(ps.south west)+\yr*(ps.south east)$);
%
%
\node[tensornode] (hv) [below right=20pt and 30pt of a2] {$\hv$};
\node[copynode] (copyhv) [above=\outcomevspace of hv] {};
\drawleg (hv.north) -- (copyhv);
\foreach \x in {0,1,2}
	\drawdashedleg (copyhv) to [out=180,in=340] ($0.2*(a\x.south west)+0.8*(a\x.south east)$);
}
=
\targetponemarg{}\targetponemarg{}\targetponemarg{}.
\end{equation}
The corresponding fanout inflation of equation \eqref{eq:def fanout inf triangle one strat} should verify
\begin{equation}
\forall a_{12},a_{34},a_{56}\st \qinf(\{A_{12} = a_{12},A_{34} = a_{34}, A_{56} = a_{56}\}) = \targetponemarg{a_{12}}\targetponemarg{a_{34}}\targetponemarg{a_{56}}.
\end{equation}

\newpage
\section{Outlook}

In this work, we introduced the postselected inflation framework that can be seen as a reformulation of the fanout inflation framework as exemplified in \cref{sec:fanout inflation}.
Despite the mathematical equivalence, the postselected inflation framework allows to conveniently devise converging outer approximations of the set of distributions causally compatible with a given classical multi-network scenario, in particular in the case where several agents are using the same strategy.
The general idea behind the convergence of these outer approximations was presented in \cref{sec:outer approx} and formally proven in \cref{sec:post-selected inflation}.

\paragraph{Further developments?}

Certain basic problems of causal compatibility remain open to this day. An interesting example is the outcome distribution of \cite{gisin_elegant_2017} that is causally compatible with the quantum triangle network. There, although it is believed that the distribution is not causally compatible with the classical triangle network, the inflation technique is not able to certify this causal incompatibility with modern computing power.
Successfully proving this incompatibility may involve the formulation of efficient outer approximation schemes to supplement the inflation framework.

\paragraph{Quantum analogues?}

The quantum analogues of fanout inflation in the context of networks featuring quantum sources are very natural to formulate, and were extensively studied \cite{wolfe_quantum_2021,ligthart_convergent_2021}. Whether the postselected inflation formulation may open the door to alternative outer approximation schemes in the quantum case is open.

\newpage
\section*{Acknowledgments}
\addcontentsline{toc}{section}{Acknowledgments}

I am thankful to Renato Renner, Marc-Olivier Renou, Raffaele Salvia, V. Vilasini and Elie Wolfe for their support and interest in this work.

\section*{Software}
\addcontentsline{toc}{section}{Software}

The numerical simulations were run in Python v3.8.10 (\textit{Python Software Foundation}, \href{python.org}{python.org}). The linear programming library used is MOSEK's Python Optimizer API v9.2.45 (\textit{MOSEK ApS}, \href{mosek.com}{mosek.com}). Please contact \href{mailto:vgitton@ethz.ch}{vgitton@ethz.ch} to gain access to the code and the detailed data.
The tensor networks, the network graphs and the plots were generated thanks to the TikZ and PGF packages v3.1.9a (\textit{The TikZ and PGF packages}, \href{pgf-tikz.github.io}{pgf-tikz.github.io}).

\addcontentsline{toc}{section}{References}
\bibliographystyle{my_style}
\bibliography{general_convex_out_approx}

\newpage
\appendix

\section{Deterministic strategies are sufficient for the Correlated Sleeper}
\label{app:det strat sc}

Here we state explicitly the definition of the norms we shall use.

\begin{definition}[$p$-norms]
\label{def:p norms}
Let $p\in\mathbb{R}$, $p\geq 1$, $k\in\mathbb{N}$, and $x = (x_1,\dots,x_k) \in\mathbb{R}^k$. We define
\begin{equation}
\macronorm{x}{p} = \left( \sum_{i=1}^k |x_i|^p \right)^{\frac{1}{p}}.
\end{equation}
\end{definition}

This definition will be extended in the obvious way to linear combinations of probability tensors that share the same finite output domain and that have no inputs. For instance, we can write
\begin{equation}
\onenorm{\atargetp{1}{}{} - \atargetp{2}{}{}} = \sum_{\outputa_1,\outputa_2} \left| \atargetp{1}{\outputa_1}{\outputa_2} - \atargetp{2}{\outputa_1}{\outputa_2}\right|.
\end{equation}

\subsection{Deterministic approximation}

We first prove a result regarding the set of feasible distributions in the context of the multi-network scenario described by the two networks $\scnetone$ and $\scnettwo$ of equations \eqref{eq:sc config 1} and \eqref{eq:sc config 2}; namely, that the set of feasible distributions allowing for non-deterministic strategies is the closure of the set $\outputdistribs{\scnetone,\scnettwo}$, which only allows for deterministic strategies (see \cref{def:causal compat}). The proof is based on the idea that a deterministic strategy taking the sum$\mod 2$ of the two inputs (discretized into bits) can generate local randomness. This lemma does not yet take into account the uniform-marginal constraint: this will be covered in \cref{lem:sc approx with marg}.

\begin{lemma}
\label{lem:sc approx no marg}
Let $\Big(\atargetp{1}{}{},\atargetp{2}{}{}\Big)$ be such that there exists a (Riemann integrable, as usual) probability tensor $\twostrat{\alice_0}{}{}{}$ with
\begin{equation}
\label{eq:non det strat}
\atargetp{1}{}{} = \acaseone{\alice_0}{}{}{0},\quad 
\atargetp{2}{}{} = \acasetwo{\alice_0}{}{}{0}.
\end{equation}
Then, for all $\epsilon > 0$, there exists a \emph{deterministic} probability tensor $\scdetstrat$ such that
\begin{align}
\label{eq:target det strat}
\onenorm{\atargetp{1}{}{} - \acaseone{\alice}{}{}{1}} \leq \epsilon, \quad
\onenorm{\atargetp{2}{}{} - \acasetwo{\alice}{}{}{1}} \leq \epsilon.
\end{align}
\end{lemma}
\begin{proof}
We can always rewrite $\twostrat{\alice_0}{}{}{}$ as a deterministic strategy $\macrothreestrat{1}{\alice_1}{}{}{}{}$ with an extra input connected to a local source of randomness:
\begin{equation}
\label{eq:def a1}
\twostrat{\alice_0}{}{}{} =:
\centertikz{
\node[detnode] (a) {$\alice_1$};
\node[voidnode] (o) [above=\outcomevspace of a] {};
\drawleg (a.north) -- (o.south);
\node[tensornode] (u) [below left=10pt and 0pt of a] {$\maxunif$};
\drawleg (u.north) -- \AnchorOneThree{a}{south};
\node[voidnode] (i1) [below=15pt of a] {};
\node[voidnode] (i2) [below right=15pt and 0pt of a] {};
\drawleg (i1.north) -- \AnchorTwoThree{a}{south};
\drawleg (i2.north) -- \AnchorThreeThree{a}{south};
}.
\end{equation}
We will make use of our assumption of Riemann integrability to approximate this newly introduced $\isource{\maxunif}{}$ source. In particular, there exists $\ninf \in\mathbb{N}$ such that\footnote{Technically, going from \eqref{eq:def a1} to \eqref{eq:approx a}, the tensor $\macrothreestrat{1}{\alice_1}{}{}{}{}{}$ needs to apply a rescaling of the input $\dsvalalpha\in\{1,\dots,\ninf\}$ coming from the $\isource{\dsunif}{}$ to map it to $\dsvalalpha/\ninf \in [0,1]$, but we leave this $\ninf$-dependence implicit.}
\begin{subequations}
\label{eq:approx a}
\begin{align}
\onenorm{\atargetp{1}{}{} - 
\centertikz{
\node[detnode] (a1) at (-5pt,10pt) {$\alice_1$};
\node[detnode] (a2) at (35pt,10pt) {$\alice_1$};
\node[voidnode] (o1) [above=\outcomevspace of a1] {};
\node[voidnode] (o2) [above=\outcomevspace of a2] {};
\drawleg (a1.north) -- (o1.south);
\drawleg (a2.north) -- (o2.south);
\node[tensornode] (un1) at (-45pt,-45pt) {$\dsunif$};
\drawleg (un1.north) -- \AnchorOneThree{a1}{south};
\node[tensornode] (u1) at (-15pt,-45pt) {$\maxunif$};
\node[copynode] (copy) [above=\outcomevspace of u1] {};
\drawleg (u1.north) -- (copy);
\drawleg (copy) -- \AnchorTwoThree{a1}{south};
\drawleg (copy) -- \AnchorTwoThree{a2}{south};
\node[tensornode] (u2) at (15pt,-45pt) {$\maxunif$};
\drawleg (u2.north) -- \AnchorThreeThree{a1}{south};
\node[tensornode] (un2) at (45pt,-45pt) {$\dsunif$};
\drawleg (un2.north) -- \AnchorOneThree{a2}{south};
\node[tensornode] (u3) at (75pt,-45pt) {$\maxunif$};
\drawleg (u3.north) -- \AnchorThreeThree{a2}{south};
}
} &\leq \epsilon,
\end{align}
\begin{align}
\onenorm{\atargetp{2}{}{} - 
\centertikz{
\node[detnode] (a1) at (-5pt,10pt) {$\alice_1$};
\node[detnode] (a2) at (35pt,10pt) {$\alice_1$};
\node[voidnode] (o1) [above=\outcomevspace of a1] {};
\node[voidnode] (o2) [above=\outcomevspace of a2] {};
\drawleg (a1.north) -- (o1.south);
\drawleg (a2.north) -- (o2.south);
\node[tensornode] (un1) at (-45pt,-45pt) {$\dsunif$};
\drawleg (un1.north) -- \AnchorOneThree{a1}{south};
\node[tensornode] (u1) at (-15pt,-45pt) {$\maxunif$};
\drawleg (u1.north) -- \AnchorTwoThree{a1}{south};
\node[tensornode] (un2) at (15pt,-45pt) {$\dsunif$};
\drawleg (un2.north) -- \AnchorOneThree{a2}{south};
\node[tensornode] (u2) at (45pt,-45pt) {$\maxunif$};
\drawleg (u2.north) -- \AnchorTwoThree{a2}{south};
\node[tensornode] (u3) at (75pt,-45pt) {$\maxunif$};
\node[copynode] (copy) [above=\outcomevspace of u3] {};
\drawleg (u3.north) -- (copy);
\drawleg (copy) -- \AnchorThreeThree{a1}{south};
\drawleg (copy) -- \AnchorThreeThree{a2}{south};
}
} &\leq \epsilon.
\end{align}
\end{subequations}
The deterministic tensor $\scdetstrat$ that achieves \eqref{eq:target det strat} will be one that is such that
\begin{subequations}
\label{eq:condition for det a}
\begin{align}
\acaseone{\alice}{}{}{1} &= \centertikz{
\node[detnode] (a1) at (-5pt,10pt) {$\alice_1$};
\node[detnode] (a2) at (35pt,10pt) {$\alice_1$};
\node[voidnode] (o1) [above=\outcomevspace of a1] {};
\node[voidnode] (o2) [above=\outcomevspace of a2] {};
\drawleg (a1.north) -- (o1.south);
\drawleg (a2.north) -- (o2.south);
\node[tensornode] (un1) at (-45pt,-45pt) {$\dsunif$};
\drawleg (un1.north) -- \AnchorOneThree{a1}{south};
\node[tensornode] (u1) at (-15pt,-45pt) {$\maxunif$};
\node[copynode] (copy) [above=\outcomevspace of u1] {};
\drawleg (u1.north) -- (copy);
\drawleg (copy) -- \AnchorTwoThree{a1}{south};
\drawleg (copy) -- \AnchorTwoThree{a2}{south};
\node[tensornode] (u2) at (15pt,-45pt) {$\maxunif$};
\drawleg (u2.north) -- \AnchorThreeThree{a1}{south};
\node[tensornode] (un2) at (45pt,-45pt) {$\dsunif$};
\drawleg (un2.north) -- \AnchorOneThree{a2}{south};
\node[tensornode] (u3) at (75pt,-45pt) {$\maxunif$};
\drawleg (u3.north) -- \AnchorThreeThree{a2}{south};
}, \label{eq:condition for det a caseone} \\
\acasetwo{\alice}{}{}{1} &= \centertikz{
\node[detnode] (a1) at (-5pt,10pt) {$\alice_1$};
\node[detnode] (a2) at (35pt,10pt) {$\alice_1$};
\node[voidnode] (o1) [above=\outcomevspace of a1] {};
\node[voidnode] (o2) [above=\outcomevspace of a2] {};
\drawleg (a1.north) -- (o1.south);
\drawleg (a2.north) -- (o2.south);
\node[tensornode] (un1) at (-45pt,-45pt) {$\dsunif$};
\drawleg (un1.north) -- \AnchorOneThree{a1}{south};
\node[tensornode] (u1) at (-15pt,-45pt) {$\maxunif$};
\drawleg (u1.north) -- \AnchorTwoThree{a1}{south};
\node[tensornode] (un2) at (15pt,-45pt) {$\dsunif$};
\drawleg (un2.north) -- \AnchorOneThree{a2}{south};
\node[tensornode] (u2) at (45pt,-45pt) {$\maxunif$};
\drawleg (u2.north) -- \AnchorTwoThree{a2}{south};
\node[tensornode] (u3) at (75pt,-45pt) {$\maxunif$};
\node[copynode] (copy) [above=\outcomevspace of u3] {};
\drawleg (u3.north) -- (copy);
\drawleg (copy) -- \AnchorThreeThree{a1}{south};
\drawleg (copy) -- \AnchorThreeThree{a2}{south};
}. \label{eq:condition for det a casetwo}
\end{align}
\end{subequations}
How can $\scdetstrat$ simulate a local source of randomness using only a deterministic function of the two inputs, $\alpha$ and $\beta$? Surely, using, say, the left input $\alpha$ only as a tentative source of local randomness will not do the trick when trying to reproduce $\atargetp{1}{}{}$.
However, suppose that we define an extractor function labeled $E_\ninf$ defined through:  
for all $\dsvalalpha\in\{1,\dots,\ninf\},$ for all $\alpha,\alpha'\in[0,1]$,
\begin{equation}
\centertikz{
\node[detnode] (g) {$E_\ninf$};
\node[voidnode] (o1) [above left=\outcomevspace and \outcomehspacetargetp of g] {\indexstyle{i}};
\node[voidnode] (o2) [above right=\outcomevspace and \outcomehspacetargetp of g] {\indexstyle{\alpha'}};
\drawleg \AnchorOneTwo{g}{north} -- (o1.south);
\drawleg \AnchorTwoTwo{g}{north} -- (o2.south);
\node[voidnode] (in) [below=\outcomevspace of g] {\indexstyle{\alpha}};
\drawleg (in.north) -- (g.south);
}
:=
\ddelta{\dsvalalpha}{\left\lfloor 1 + \ninf \alpha \right\rfloor}\delta\left(\alpha' - \left(\ninf \alpha - \left\lfloor \ninf\alpha \right\rfloor\right)\right).
\end{equation}
For instance, in the case $\ninf = 3$,
\begin{equation}
\centertikz{
\node[detnode] (g) {$E_3$};
\node[voidnode] (o1) [above left=\outcomevspace and \outcomehspacetargetp of g] {\indexstyle{i}};
\node[voidnode] (o2) [above right=\outcomevspace and \outcomehspacetargetp of g] {\indexstyle{\alpha'}};
\drawleg \AnchorOneTwo{g}{north} -- (o1.south);
\drawleg \AnchorTwoTwo{g}{north} -- (o2.south);
\node[voidnode] (in) [below=\outcomevspace of g] {\indexstyle{\frac{2}{3} + \frac{1}{9}}};
\drawleg (in.north) -- (g.south);
}
= 
\ddelta{\dsvalalpha}{3}\delta\left(\alpha' - \frac{1}{3}\right).
\end{equation}
We furthermore define a sum modulo $\ninf$ function as, for all $\dsvalalpha_1,\dsvalalpha_2,\dsvalalpha_3 \in \{1,\dots,\ninf\}:$
\begin{equation}
\centertikz{
\node[detnode] (sum) {$\oplus_\ninf$};
\node[voidnode] (o) [above=\outcomevspace of sum] {\indexstyle{\dsvalalpha_3}};
\drawleg (sum.north) -- (o.south);
\node[voidnode] (i1) [below left=\outcomevspace and \outcomehspacetargetp of sum] {\indexstyle{\dsvalalpha_1}};
\node[voidnode] (i2) [below right=\outcomevspace and \outcomehspacetargetp of sum] {\indexstyle{\dsvalalpha_2}};
\drawleg (i1.north) -- \AnchorOneTwo{sum}{south};
\drawleg (i2.north) -- \AnchorTwoTwo{sum}{south};
}
:=
\delta\left(\dsvalalpha_3 - \big(\dsvalalpha_1 + \dsvalalpha_2 \mod \ninf \big)\right).
\end{equation}
With these new tensors at hand, let us define our desired tensor $\scdetstrat$, corresponding to $\alice$ extracting two discrete values from her continuous inputs and taking their sum$\mod \ninf$, before forwarding the resulting three values into the strategy $\macrothreestrat{1}{\alice_1}{}{}{}{}$ of \eqref{eq:approx a}:
\begin{equation}
\label{eq:winning choice det a}
\scdetstratargs
:=
\centertikz{
\node[detnode] (sum) {$\oplus_\ninf$};
\node[detnode] (e1) [below right=20pt and -10pt of sum] {$E_\ninf$};
\node[detnode] (e2) [below right=20pt and 20pt of sum] {$E_\ninf$};
\drawleg \AnchorOneTwo{e1}{north} -- \AnchorOneTwo{sum}{south};
\drawleg \AnchorOneTwo{e2}{north} -- \AnchorTwoTwo{sum}{south};
\node[voidnode] (alpha) [below=\outcomevspace of e1] {\indexstyle{\alpha}};
\node[voidnode] (beta) [below=\outcomevspace of e2] {\indexstyle{\beta}};
\drawleg (alpha.north) -- (e1.south);
\drawleg (beta.north) -- (e2.south);
\node[detnode] (a) [above right=20pt and 3pt of sum] {$\alice_1$};
\node[voidnode] (o) [above=\outcomevspace of a] {\indexstyle{\outputa}};
\drawleg (a.north) -- (o.south);
\drawleg (sum.north) -- \AnchorOneThree{a}{south};
\drawleg \AnchorTwoTwo{e1}{north} -- \AnchorTwoThree{a}{south};
\drawleg \AnchorTwoTwo{e2}{north} -- \AnchorThreeThree{a}{south};
}
\end{equation}
It now remains to verify \eqref{eq:condition for det a}. To do so, we will make use of three useful tensor identities.

\begin{claim}
it holds that
\begin{subequations}
\label{eq:claim det a}
\begin{align}
\centertikz{
\node[detnode] (e) {$E_\ninf$};
\node[voidnode] (o1) [above left=\outcomevspace and \outcomehspacetargetp of e] {};
\node[voidnode] (o2) [above right=\outcomevspace and \outcomehspacetargetp of e] {};
\drawleg \AnchorOneTwo{e}{north} -- (o1.south);
\drawleg \AnchorTwoTwo{e}{north} -- (o2.south);
\node[tensornode] (u) [below=10pt of e] {$\maxunif$};
\drawleg (u.north) -- (e.south);
} &= 
\isource{\dsunif}{} \isource{\maxunif}{}, \label{eq:en factorization} \\
\centertikz{
\node[detnode] (sum) {$\oplus_\ninf$};
\node[voidnode] (o) [above=\outcomevspace of sum] {};
\drawleg (sum.north) -- (o.south);
\node[tensornode] (u) [below left=15pt and -3pt of sum] {$\dsunif$};
\drawleg (u.north) -- \AnchorOneTwo{sum}{south};
\node[voidnode] (i) [below right=15pt and -3pt of sum] {};
\drawleg (i.north) -- \AnchorTwoTwo{sum}{south};
}
&=
\centertikz{
\node[margnode] (marg) {};
\node[voidnode] (i) [below=\outcomevspace of marg] {};
\drawleg (i.north) -- (marg.south);
\node[tensornode] (u) [above=\outcomevspace of marg] {$\dsunif$};
\node[voidnode] (o) [above=\outcomevspace of u] {};
\drawleg (u.north) -- (o.south);
}, \label{eq:random modulo} \\
\centertikz{
\node[detnode] (e1) at (0,0) {$E_\ninf$};
\node[margnode] (marg1) [above left=\outcomevspace and \outcomehspacetargetp of e1] {};
\node[voidnode] (o1) [above right=\outcomevspace and \outcomehspacetargetp of e1] {};
\drawleg \AnchorOneTwo{e1}{north} -- (marg1.south);
\drawleg \AnchorTwoTwo{e1}{north} -- (o1.south);
\node[detnode] (e2) at (30pt,0) {$E_\ninf$};
\node[margnode] (marg2) [above left=\outcomevspace and \outcomehspacetargetp of e2] {};
\node[voidnode] (o2) [above right=\outcomevspace and \outcomehspacetargetp of e2] {};
\drawleg \AnchorOneTwo{e2}{north} -- (marg2.south);
\drawleg \AnchorTwoTwo{e2}{north} -- (o2.south);
\node[tensornode] (u) at (15pt,-45pt) {$\maxunif$};
\node[copynode] (copy) [above=\outcomevspace of u] {};
\drawleg (u.north) -- (copy);
\drawleg (copy) -- (e1.south);
\drawleg (copy) -- (e2.south);
}
&=
\centertikz{
\node[tensornode] (u)  {$\maxunif$};
\node[copynode] (copy) [above=\outcomevspace of u] {};
\drawleg (u.north) -- (copy);
\node[voidnode] (o1) [above left=10pt and 5pt of copy] {};
\node[voidnode] (o2) [above right=10pt and 5pt of copy] {};
\drawleg (copy) -- (o1.south);
\drawleg (copy) -- (o2.south);
}. \label{eq:easy copy en}
\end{align}
\end{subequations}
\begin{claimproof}
To prove \eqref{eq:en factorization}, it suffices to realize that the map $E_\ninf$ is invertible, and so it must map the uniform distribution over $[0,1]$ to the uniform distribution over $\{1,\dots,\ninf\} \times [0,1]$. To prove \eqref{eq:random modulo}, it suffices to see that for all $\dsvalalpha_2,\dsvalalpha_3 \in \{1,\dots,\ninf\}$,
\begin{equation}
\centertikz{
\node[detnode] (sum) {$\oplus_\ninf$};
\node[voidnode] (o) [above=\outcomevspace of sum] {\indexstyle{\dsvalalpha_3}};
\drawleg (sum.north) -- (o.south);
\node[tensornode] (u) [below left=15pt and -3pt of sum] {$\dsunif$};
\drawleg (u.north) -- \AnchorOneTwo{sum}{south};
\node[voidnode] (i) [below right=15pt and -3pt of sum] {\indexstyle{\dsvalalpha_2}};
\drawleg (i.north) -- \AnchorTwoTwo{sum}{south};
}
=
\frac{1}{\ninf} \sum_{\dsvalalpha_1 = 1}^\ninf \delta\left(\dsvalalpha_3 - (\dsvalalpha_1 + \dsvalalpha_2 \mod \ninf)\right)
= \frac{1}{\ninf}
= \centertikz{
\node[margnode] (marg) {};
\node[voidnode] (i) [below=\outcomevspace of marg] {\indexstyle{\dsvalalpha_2}};
\drawleg (i.north) -- (marg.south);
\node[tensornode] (u) [above=\outcomevspace of marg] {$\dsunif$};
\node[voidnode] (o) [above=\outcomevspace of u] {\indexstyle{\dsvalalpha_3}};
\drawleg (u.north) -- (o.south);
}.
\end{equation}
Equation \eqref{eq:easy copy en} is true since, using the invertibility of $E_\ninf$ and then \eqref{eq:en factorization}, it holds that for all $\alpha_1,\alpha_2 \in [0,1]$,
\newpage
\begin{equation}
\centertikz{
\node[detnode] (e1) at (0,0) {$E_\ninf$};
\node[margnode] (marg1) [above left=\outcomevspace and \outcomehspacetargetp of e1] {};
\node[voidnode] (o1) [above right=\outcomevspace and \outcomehspacetargetp of e1] {\indexstyle{\alpha_1}};
\drawleg \AnchorOneTwo{e1}{north} -- (marg1.south);
\drawleg \AnchorTwoTwo{e1}{north} -- (o1.south);
\node[detnode] (e2) at (50pt,0) {$E_\ninf$};
\node[margnode] (marg2) [above left=\outcomevspace and \outcomehspacetargetp of e2] {};
\node[voidnode] (o2) [above right=\outcomevspace and \outcomehspacetargetp of e2] {\indexstyle{\alpha_2}};
\drawleg \AnchorOneTwo{e2}{north} -- (marg2.south);
\drawleg \AnchorTwoTwo{e2}{north} -- (o2.south);
\node[tensornode] (u) at (25pt,-45pt) {$\maxunif$};
\node[copynode] (copy) [above=\outcomevspace of u] {};
\drawleg (u.north) -- (copy);
\drawleg (copy) -- (e1.south);
\drawleg (copy) -- (e2.south);
}
=
\ddelta{\alpha_1}{\alpha_2} 
\centertikz{
\node[detnode] (e) {$E_\ninf$};
\node[margnode] (o1) [above left=\outcomevspace and \outcomehspacetargetp of e] {};
\node[voidnode] (o2) [above right=\outcomevspace and \outcomehspacetargetp of e] {\indexstyle{\alpha_1}};
\drawleg \AnchorOneTwo{e}{north} -- (o1.south);
\drawleg \AnchorTwoTwo{e}{north} -- (o2.south);
\node[tensornode] (u) [below=10pt of e] {$\maxunif$};
\drawleg (u.north) -- (e.south);
}
=
\ddelta{\alpha_1}{\alpha_2}
=
\centertikz{
\node[tensornode] (u)  {$\maxunif$};
\node[copynode] (copy) [above=\outcomevspace of u] {};
\drawleg (u.north) -- (copy);
\node[voidnode] (o1) [above left=10pt and 5pt of copy] {\indexstyle{\alpha_1}};
\node[voidnode] (o2) [above right=10pt and 5pt of copy] {\indexstyle{\alpha_2}};
\drawleg (copy) -- (o1.south);
\drawleg (copy) -- (o2.south);
}.
\end{equation}
\end{claimproof}
\end{claim}

Let us now establish equation \eqref{eq:condition for det a caseone} using the choice \eqref{eq:winning choice det a} and the identities \eqref{eq:claim det a}:
\begin{subequations}
\begin{align}
\acaseone{\alice}{}{}{1} =
\centertikz{
\node[tensornode] (alpha) at (0,0) {$\maxunif$};
\node[copynode] (alphacopy) [above=\outcomevspace of alpha] {};
\drawleg (alpha.north) -- (alphacopy);
\node[tensornode] (beta) at (30pt,0) {$\maxunif$};
\node[tensornode] (gamma) at (60pt,0) {$\maxunif$};
\node[detnode] (e1) at (-15pt,50pt) {$E_\ninf$};
\node[detnode] (e2) at (15pt,50pt) {$E_\ninf$};
\node[detnode] (e3) at (45pt,50pt) {$E_\ninf$};
\node[detnode] (e4) at (75pt,50pt) {$E_\ninf$};
\drawleg (alphacopy) -- (e1.south);
\drawleg (alphacopy) -- (e3.south);
\drawleg (beta.north) -- (e2.south);
\drawleg (gamma.north) -- (e4.south);
\node[detnode] (sum1) at (-25pt,100pt) {$\oplus_\ninf$};
\node[detnode] (sum2) at (35pt,100pt) {$\oplus_\ninf$};
\drawleg \AnchorOneTwo{e1}{north} -- \AnchorOneTwo{sum1}{south};
\drawleg \AnchorOneTwo{e2}{north} -- \AnchorTwoTwo{sum1}{south};
\drawleg \AnchorOneTwo{e3}{north} -- \AnchorOneTwo{sum2}{south};
\drawleg \AnchorOneTwo{e4}{north} -- \AnchorTwoTwo{sum2}{south};
\node[detnode] (a1) at (0pt, 150pt) {$\alice_1$};
\node[detnode] (a2) at (60pt, 150pt) {$\alice_1$};
\drawleg (sum1.north) -- \AnchorOneThree{a1}{south};
\drawleg (sum2.north) -- \AnchorOneThree{a2}{south};
\drawleg \AnchorTwoTwo{e1}{north} -- \AnchorTwoThree{a1}{south};
\drawleg \AnchorTwoTwo{e2}{north} -- \AnchorThreeThree{a1}{south};
\drawleg \AnchorTwoTwo{e3}{north} -- \AnchorTwoThree{a2}{south};
\drawleg \AnchorTwoTwo{e4}{north} -- \AnchorThreeThree{a2}{south};
\node[voidnode] (o1) [above=\outcomevspace of a1] {};
\node[voidnode] (o2) [above=\outcomevspace of a2] {};
\drawleg (a1.north) -- (o1.south);
\drawleg (a2.north) -- (o2.south);
}
%
%
%
\overset{\eqref{eq:en factorization}}{=}
\centertikz{
\node[tensornode] (alpha) at (0,0) {$\maxunif$};
\node[copynode] (alphacopy) [above=\outcomevspace of alpha] {};
\drawleg (alpha.north) -- (alphacopy);
\node[detnode] (e1) at (-15pt,50pt) {$E_\ninf$};
\node[tensornode] (e2) at (15pt,50pt) {$\dsunif$};
\node[detnode] (e3) at (45pt,50pt) {$E_\ninf$};
\node[tensornode] (e4) at (75pt,50pt) {$\dsunif$};
\drawleg (alphacopy) -- (e1.south);
\drawleg (alphacopy) -- (e3.south);
\node[detnode] (sum1) at (-25pt,100pt) {$\oplus_\ninf$};
\node[detnode] (sum2) at (35pt,100pt) {$\oplus_\ninf$};
\drawleg \AnchorOneTwo{e1}{north} -- \AnchorOneTwo{sum1}{south};
\drawleg (e2.north) -- \AnchorTwoTwo{sum1}{south};
\drawleg \AnchorOneTwo{e3}{north} -- \AnchorOneTwo{sum2}{south};
\drawleg (e4.north) -- \AnchorTwoTwo{sum2}{south};
\node[detnode] (a1) at (0pt, 150pt) {$\alice_1$};
\node[detnode] (a2) at (60pt, 150pt) {$\alice_1$};
\drawleg (sum1.north) -- \AnchorOneThree{a1}{south};
\drawleg (sum2.north) -- \AnchorOneThree{a2}{south};
\drawleg \AnchorTwoTwo{e1}{north} -- \AnchorTwoThree{a1}{south};
\drawleg \AnchorTwoTwo{e3}{north} -- \AnchorTwoThree{a2}{south};
\node[voidnode] (o1) [above=\outcomevspace of a1] {};
\node[voidnode] (o2) [above=\outcomevspace of a2] {};
\drawleg (a1.north) -- (o1.south);
\drawleg (a2.north) -- (o2.south);
\node[tensornode] (beta) [below right=10pt and 0pt of a1] {$\maxunif$};
\node[tensornode] (gamma) [below right=10pt and 0pt of a2] {$\maxunif$};
\drawleg (beta.north) -- \AnchorThreeThree{a1}{south};
\drawleg (gamma.north) -- \AnchorThreeThree{a2}{south};
}  \\
%
%
%
\overset{\eqref{eq:random modulo}}{=}
\centertikz{
\node[tensornode] (alpha) at (30pt,40pt) {$\maxunif$};
\node[copynode] (alphacopy) [above=\outcomevspace of alpha] {};
\drawleg (alpha.north) -- (alphacopy);
\node[detnode] (e1) at (-10pt,85pt) {$E_\ninf$};
\node[detnode] (e3) at (70pt,85pt) {$E_\ninf$};
\node[margnode] (m1) [above left=\outcomevspace and \outcomehspacetargetp of e1] {};
\node[margnode] (m2) [above left=\outcomevspace and \outcomehspacetargetp of e3] {};
\drawleg \AnchorOneTwo{e1}{north} -- (m1);
\drawleg \AnchorOneTwo{e3}{north} -- (m2);
\drawleg (alphacopy) -- (e1.south);
\drawleg (alphacopy) -- (e3.south);
\node[detnode] (a1) at (-10pt, 150pt) {$\alice_1$};
\node[detnode] (a2) at (70pt, 150pt) {$\alice_1$};
\drawleg \AnchorTwoTwo{e1}{north} -- \AnchorTwoThree{a1}{south};
\drawleg \AnchorTwoTwo{e3}{north} -- \AnchorTwoThree{a2}{south};
\node[voidnode] (o1) [above=\outcomevspace of a1] {};
\node[voidnode] (o2) [above=\outcomevspace of a2] {};
\drawleg (a1.north) -- (o1.south);
\drawleg (a2.north) -- (o2.south);
\node[tensornode] (beta) [below right=10pt and 0pt of a1] {$\maxunif$};
\node[tensornode] (gamma) [below right=10pt and 0pt of a2] {$\maxunif$};
\drawleg (beta.north) -- \AnchorThreeThree{a1}{south};
\drawleg (gamma.north) -- \AnchorThreeThree{a2}{south};
\node[tensornode] (un1) [below left=10pt and 0pt of a1] {$\dsunif$};
\node[tensornode] (un2) [below left=10pt and 0pt of a2] {$\dsunif$};
\drawleg (un1.north) -- \AnchorOneThree{a1}{south};
\drawleg (un2.north) -- \AnchorOneThree{a2}{south};
} 
%
%
\overset{\eqref{eq:easy copy en}}{=}
\centertikz{
\node[detnode] (a1) at (-5pt,10pt) {$\alice_1$};
\node[detnode] (a2) at (35pt,10pt) {$\alice_1$};
\node[voidnode] (o1) [above=\outcomevspace of a1] {};
\node[voidnode] (o2) [above=\outcomevspace of a2] {};
\drawleg (a1.north) -- (o1.south);
\drawleg (a2.north) -- (o2.south);
\node[tensornode] (un1) at (-45pt,-45pt) {$\dsunif$};
\drawleg (un1.north) -- \AnchorOneThree{a1}{south};
\node[tensornode] (u1) at (-15pt,-45pt) {$\maxunif$};
\node[copynode] (copy) [above=\outcomevspace of u1] {};
\drawleg (u1.north) -- (copy);
\drawleg (copy) -- \AnchorTwoThree{a1}{south};
\drawleg (copy) -- \AnchorTwoThree{a2}{south};
\node[tensornode] (u2) at (15pt,-45pt) {$\maxunif$};
\drawleg (u2.north) -- \AnchorThreeThree{a1}{south};
\node[tensornode] (un2) at (45pt,-45pt) {$\dsunif$};
\drawleg (un2.north) -- \AnchorOneThree{a2}{south};
\node[tensornode] (u3) at (75pt,-45pt) {$\maxunif$};
\drawleg (u3.north) -- \AnchorThreeThree{a2}{south};
}.
\end{align}
\end{subequations}
This establishes \eqref{eq:condition for det a caseone}, and the case of \eqref{eq:condition for det a casetwo} is completely analogous.
\end{proof}

\subsection{Deterministic approximation, exact marginal}

We will make use of the following lemma relating the trace distance between the marginals of arbitrary distributions:
\begin{lemma}
\label{lem:onenorm marg}
For all $\macroatargetp{\targetpname}{}{}$, $\macroatargetp{\targetptildename}{}{}$, it holds that
\begin{equation}
\onenorm{\macroatargetpmarg{\targetpname}{} - \macroatargetpmarg{\targetptildename}{}}
\leq
\onenorm{\macroatargetp{\targetpname}{}{} - \macroatargetp{\targetptildename}{}{}}.
\end{equation}
\end{lemma}
\begin{proof}
Using the triangle inequality and the definition of the $\onenorm{\cdot}$ norm (see \cref{def:p norms}), we get:
\begin{subequations}
\begin{align}
\onenorm{\macroatargetpmarg{\targetpname}{} - \macroatargetpmarg{\targetptildename}{}}
&= 
\sum_{\outputa}
\left|
\macroatargetpmarg{\targetpname}{\outputa} - \macroatargetpmarg{\targetptildename}{\outputa}
\right| \\
&= 
\sum_{\outputa}
\left|
\sum_{\outputb}
\left(
\macroatargetp{\targetpname}{\outputa}{\outputb} - \macroatargetp{\targetptildename}{\outputa}{\outputb}
\right)
\right| \\
&\leq
\sum_{\outputa,\outputb}
\left| 
\macroatargetp{\targetpname}{\outputa}{\outputb} - \macroatargetp{\targetptildename}{\outputa}{\outputb}
\right| \\
&= \onenorm{\macroatargetp{\targetpname}{}{} - \macroatargetp{\targetptildename}{}{}},
\end{align}
\end{subequations}
which concludes the proof.
\end{proof}

We now prove a slightly stronger result, which takes into account the marginal constraint of the Correlated Sleeper task. The idea is to slightly deform the strategy obtained in \cref{lem:sc approx no marg} to maintain closeness with the target output distributions while \emph{exactly} achieving the desired marginal constraint in the network $\scnetthree$ of equation \eqref{eq:sc config 3}.

\begin{lemma}
\label{lem:sc approx with marg}
Let $\Big(\atargetp{1}{}{},\atargetp{2}{}{}\Big)$ be such that there exists a (Riemann integrable, as usual) probability tensor $\twostrat{\alice_0}{}{}{}$ with
\begin{equation}
\label{eq:non det strat 2}
\atargetp{1}{}{} = \acaseone{\alice_0}{}{}{0},\quad 
\atargetp{2}{}{} = \acasetwo{\alice_0}{}{}{0},\quad \isource{\macrodsunif{2}}{} = \amargconstraint{\alice_0}{}{0}.
\end{equation}
Then, for all $\epsilon > 0$, there exists a \emph{deterministic} probability tensor $\scdetstrat$ such that
\begin{equation}
\label{eq:target marg constraint}
\isource{\macrodsunif{2}}{} = \amargconstraint{\alice}{}{1}
\end{equation}
and
\begin{equation}
\label{eq:target det strat 2}
\frac{1}{2}\onenorm{\atargetp{1}{}{} - \acaseone{\alice}{}{}{1}}
+
\frac{1}{2}\onenorm{\atargetp{2}{}{} - \acasetwo{\alice}{}{}{1}} \leq \epsilon.
\end{equation}
\end{lemma}
\begin{proof}
Thanks to \cref{lem:sc approx no marg}, we know that for all $\epsilon > 0$, there exists $\dettwostrat{\alice_1}{}{}{}$ such that (the factor of 1/3 is chosen for later convenience):
\begin{align}
\label{eq:lemma a1 close}
\onenorm{\atargetp{1}{}{} - \acaseone{\alice_1}{}{}{1}} \leq \frac{\epsilon}{3}, \quad
\onenorm{\atargetp{2}{}{} - \acasetwo{\alice_1}{}{}{1}} \leq \frac{\epsilon}{3}.
\end{align}
This $\dettwostrat{\alice_1}{}{}{}$ does not have the desired marginal of \eqref{eq:target marg constraint}, but is close to it: using \cref{lem:onenorm marg}, the last equation of \eqref{eq:non det strat 2} and the first inequality of equation \eqref{eq:lemma a1 close}, we get
\begin{multline}
\label{eq:lemma a1 u2}
\onenorm{\amargconstraint{\alice_1}{}{1} - \isource{\macrodsunif{2}}{}} 
\overset{\textup{eq.}\ \eqref{eq:non det strat 2}}{=}
\onenorm{
\centertikz{
\node[detnode] (a1) at (0,0) {$A_1$};
\node[detnode] (a2) at (30pt,0) {$A_1$};
\node[voidnode] (o1) [above=\outcomevspace of a1] {\indexstyle{}};
\node[margnode] (o2) [above=\outcomevspace of a2] {\indexstyle{}};
\drawleg (a1.north) -- (o1.south);
\drawleg (a2.north) -- (o2.south);
\node[tensornode] (u1) at (-15pt,-45pt) {$\maxunif$};
\node[copynode] (copy) [above=\outcomevspace of u1] {};
\drawleg (u1.north) -- (copy);
\drawleg (copy) -- \AnchorOneTwo{a1}{south};
\drawleg (copy) -- \AnchorOneTwo{a2}{south};
\node[tensornode] (u2) at (15pt,-45pt) {$\maxunif$};
\drawleg (u2.north) -- \AnchorTwoTwo{a1}{south};
\node[tensornode] (u3) at (45pt,-45pt) {$\maxunif$};
\drawleg (u3.north) -- \AnchorTwoTwo{a2}{south};
}
-
\macroatargetpmarg{\targetpname_1}{}
} \\
\overset{\textup{lem.}\ \ref{lem:onenorm marg}}{\leq}
\onenorm{\acaseone{\alice_1}{}{}{1} - \atargetp{1}{}{}} \overset{\textup{eq.}\ \eqref{eq:lemma a1 close}}{\leq} \frac{\epsilon}{3}.
\end{multline}
If we parametrize the marginal distribution of $\dettwostrat{\alice_1}{}{}{}$ with some $\nu\in\mathbb R$ such that
\begin{subequations}
\label{eq:lemma param a1 nu}
\begin{align}
\amargconstraint{\alice_1}{1}{1} &= \frac{1}{2} + \nu, \label{eq:lemma nu1} \\
\amargconstraint{\alice_1}{2}{1} &= \frac{1}{2} - \nu,
\end{align}
\end{subequations}
then \eqref{eq:lemma a1 u2} implies that 
\begin{equation}
\label{eq:relation nu epsilon}
|\nu| \leq \epsilon/6.
\end{equation}
Suppose that $\nu \geq 0$ (otherwise, swap the role of the outcomes 1 and 2 in the following argument). 
Let $\mathcal Y$ be a subset of area $\nu$ of the unit square $[0,1]^{\times 2}$ such that, for all $(\alpha,\beta) \in \mathcal Y$, we have
\begin{equation}
\dettwostrat{\alice_1}{\outputa}{\alpha}{\beta} = \kdelta{\outputa}{1},
\end{equation}
that is, $\mathcal Y$ corresponds to an input range where $\alice_1$ always outputs $1$.
Such a $\mathcal Y$ always exists: \eqref{eq:lemma nu1} states that there exists a subset of area $1/2 + \nu$ of the unit square which has this property, so we may simply choose a subset thereof with the right area.

Now, consider the modified strategy $\dettwostrat{\alice}{}{}{}$ defined as follows:
\begin{equation}
\dettwostrat{\alice}{\outputa}{\alpha}{\beta} := \left\{\begin{aligned}
\dettwostrat{\alice_1}{\outputa}{\alpha}{\beta}, &\textup{ if } (\alpha,\beta) \in [0,1]^{\times 2} \setminus \mathcal Y, \\
\kdelta{\outputa}{2}, &\textup{ if } (\alpha,\beta) \in \mathcal Y,
\end{aligned}\right.
\end{equation}
that is, we reverse the output in the region $\mathcal Y$ but otherwise leave the strategy unchanged. This new strategy clearly verifies the marginal constraint \eqref{eq:target marg constraint}: for all $\outputa \in \{1,2\}$,
\begin{subequations}
\begin{align}
\amargconstraint{\alice}{\outputa}{1} &= \int_{[0,1]^{\times 2}} \dd \alpha \dd\beta \dettwostrat{\alice}{\outputa}{\alpha}\beta \\
&=
\int_{[0,1]^{\times 2} \setminus \mathcal Y} \dd \alpha \dd\beta \dettwostrat{\alice}{\outputa}{\alpha}\beta 
+ \int_{\mathcal Y}\dd \alpha \dd\beta \dettwostrat{\alice}{\outputa}{\alpha}\beta \\
&=
\int_{[0,1]^{\times 2} \setminus \mathcal Y} \dd \alpha \dd\beta \dettwostrat{\alice_1}{\outputa}{\alpha}\beta 
+ \nu \kdelta{\outputa}{2} &\netnote{def.\ of $\dettwostrat{\alice}{}{}{}$, area of $\mathcal Y$ is $\nu$}\\
&=
\int_{[0,1]^{\times 2}} \dd \alpha \dd\beta \dettwostrat{\alice_1}{\outputa}{\alpha}\beta - \int_{\mathcal Y} \dd\alpha\dd\beta \dettwostrat{\alice_1}{\outputa}{\alpha}{\beta}
+ \nu \kdelta{\outputa}{2} \\
&= \amargconstraint{\alice_1}{\outputa}{1} - \nu \kdelta{\outputa}{1} + \nu \kdelta{\outputa}{2} \\
&= \frac{1}{2}. &\netnote{see eq.\ \eqref{eq:lemma param a1 nu}}
\end{align}
\end{subequations}
Furthermore, we can show that this small modification of the strategy has a small impact on the output correlations.

\begin{claim}
it holds that 
\begin{subequations}
\label{eq:lemma small output change}
\begin{align}
\onenorm{\atargetp{1}{}{} - \acaseone{\alice}{}{}{1}}
&\leq \epsilon, \label{eq:lemma small output change a}\\
\onenorm{\atargetp{2}{}{} - \acasetwo{\alice}{}{}{1}} &\leq \epsilon. \label{eq:lemma small output change b}
\end{align}
\end{subequations}
\begin{claimproof}
Using in particular the triangle inequality, we verify explicitly \eqref{eq:lemma small output change a}:
\begin{subequations}
\begin{align}
&\onenorm{\atargetp{1}{}{} - \acaseone{\alice}{}{}{1}} \\
\leq\ 
&\onenorm{\atargetp{1}{}{} - \acaseone{\alice_1}{}{}{1}}
+
\onenorm{\acaseone{\alice_1}{}{}{1} - \acaseone{\alice}{}{}{1}} \\
\leq\ 
&\frac{\epsilon}{3} +
\sum_{\outputa,\outputb \in \{1,2\}} \left|
\int \dd\alpha\dd{\beta_1}\dd{\beta_2}
\dettwostrat{\alice_1}{\outputa}{\alpha}{\beta_1}
\dettwostrat{\alice_1}{\outputb}{\alpha}{\beta_2}
-
\dettwostrat{\alice}{\outputa}{\alpha}{\beta_1}
\dettwostrat{\alice}{\outputb}{\alpha}{\beta_2}
\right| \\
=\ 
&\frac{\epsilon}{3}
+
\sum_{\outputa,\outputb}
\left|
\int \dd\alpha\dd{\beta_1}\dd{\beta_2}
\dettwostrat{\alice_1}{\outputa}{\alpha}{\beta_1}
\left(
\dettwostrat{\alice_1}{\outputb}{\alpha}{\beta_2}
-
\dettwostrat{\alice}{\outputb}{\alpha}{\beta_2}
\right)
+
\left(
\dettwostrat{\alice_1}{\outputa}{\alpha}{\beta_1}
-
\dettwostrat{\alice}{\outputa}{\alpha}{\beta_1}
\right)
\dettwostrat{\alice}{\outputb}{\alpha}{\beta_2}
\right|  \\
\leq\ 
&\frac{\epsilon}{3} + \sum_{\outputb}
\int\dd\alpha\dd{\beta_2} \left|
\dettwostrat{\alice_1}{\outputb}{\alpha}{\beta_2}
-
\dettwostrat{\alice}{\outputb}{\alpha}{\beta_2}
\right|
+
\sum_{\outputa}\int\dd\alpha\dd{\beta_1} \left|
\dettwostrat{\alice_1}{\outputa}{\alpha}{\beta_1}
-
\dettwostrat{\alice}{\outputa}{\alpha}{\beta_1}
\right| \\
=\ 
&\frac{\epsilon}{3} + 2\sum_{\outputa}\int_{[0,1]^{\times 2}\setminus\mathcal Y}\dd\alpha\dd\beta \left|
\dettwostrat{\alice_1}{\outputa}{\alpha}{\beta}
-
\dettwostrat{\alice}{\outputa}{\alpha}{\beta}
\right|
+
2\sum_{\outputa}\int_{\mathcal Y}\dd\alpha\dd\beta \left|
\dettwostrat{\alice_1}{\outputa}{\alpha}{\beta}
-
\dettwostrat{\alice}{\outputa}{\alpha}{\beta}
\right| \\
=\ 
&\frac{\epsilon}{3} + 2 \cdot 0 + 2 \sum_{\outputa\in\{1,2\}} \nu |\kdelta{\outputa}{1} - \kdelta{\outputa}{2}| \\
=\ &\epsilon + 4 \nu \\
\leq\ &\frac{\epsilon}{3} + 4 \frac{\epsilon}{6} \qquad \netnote{see eq.\ \eqref{eq:relation nu epsilon}}\\
=\ &\epsilon.
\end{align}
\end{subequations}
The case of \eqref{eq:lemma small output change b} can be obtained by applying the same argument to the strategies $\alice_1$ and $\alice$ under the exchange of the role of the two inputs $\alpha$ and $\beta$.
\end{claimproof}
\end{claim}
This then implies the desired property \eqref{eq:target det strat 2}, which concludes the proof.
\end{proof}

\subsection{Conclusion}

The last lemma that we need is a basic topological result that we will use in the following proof.

\begin{lemma}
\label{lem:basic topology}
Let $X$ be a subset of a metric space $M$ with metric $d : M \times M \to \mathbb{R}$. Let $\phi : M \rightarrow \mathbb R$ be a continuous function. Then, it holds that the closure of the image of $X$ equals the closure of the image of the closure of $X$:
\begin{equation}
\closure{\phi(X)} = \closure{\phi(\closure{X})}.
\end{equation}
\end{lemma}
\begin{proof}
Since $\phi(X) \subseteq \phi(\closure{X})$, we also have $\closure{\phi(X)} \subseteq \closure{\phi(\closure{X})}$. To show the other direction, let $u \in \closure{\phi(\closure{X})}$. In particular, there exists a sequence $(x_k \in \closure{X})_{k\in\mathbb N}$ such that
\begin{equation}
\label{eq:uxk}
u = \lim_{k\rightarrow \infty} \phi(x_k).
\end{equation}
For each $k$, since $x_k \in \closure{X}$, it holds that for all $\delta_k > 0$ there exists $y_k \in X$ such that
\begin{equation}
d(x_k,y_k) \leq \delta_k.
\end{equation}
Choose $\delta_k$ such that $|\phi(x_k) - \phi(y_k)| \leq 1/k$ (this is possible since $\phi$ is continuous).

\begin{claim}
It holds that
\begin{equation}
u = \lim_{k\to\infty} \phi(y_k).
\end{equation}
\begin{claimproof}
We compute:
\begin{align}
|u - \phi(y_k)| &= |u - \phi(x_k) + \phi(x_k) - \phi(y_k)| \\
&\leq |u - \phi(x_k)| + |\phi(x_k) - \phi(y_k)| \\
&\leq |u - \phi(x_k)| + \frac{1}{k} \overset{k\to\infty}{\to} 0,
\end{align}
where we used \cref{eq:uxk} in the last step.
\end{claimproof}
\end{claim}
This proves that $u \in \closure{\phi(X)}$.
\end{proof}

\PropDetApp*

\begin{proof}
Consider the metric space $M$ with metric $d : M \times M \to \mathbb R$ such that
\begin{subequations}
\begin{align}
M := \left\{ \Big(\atargetp{1}{}{},\atargetp{2}{}{}\Big) \middle| \atargetp{1}{}{}, \atargetp{2}{}{} \in \mathbb{R}^{2\times 2} \textup{ are probability tensors}\right\}, \\
\metric\left[ \Big(\atargetp{1}{}{}, \atargetp{2}{}{}\Big), 
\Big(\atargetpbis{1}{}{}, \atargetpbis{2}{}{}\Big)
\right]
:=
\frac{1}{2}\onenorm{\atargetp{1}{}{} - \atargetpbis{1}{}{}}
+
\frac{1}{2}\onenorm{\atargetp{2}{}{} - \atargetpbis{2}{}{}}.
\end{align}
\end{subequations}
Consider the two subsets $\mathcal R \subset M$ (standing for ``allowing the use of local Randomness'') and $\mathcal L \subset M$ (a special case of \cref{def:causal compat}) defined as
\begin{subequations}
\begin{align}
\mathcal R &:= \left\{
\Big(\atargetp{1}{}{}, \atargetp{2}{}{}\Big)
\middle|
\exists \twostrat{\alice}{}{}{} \textup{ s.t. } \atargetp{1}{}{} = \acaseone{\alice}{}{}{0}, \nonumber
\right. \\
&\hspace{3.4cm} \left. 
\atargetp{2}{}{} = \acaseone{\alice}{}{}{0},
\isource{\macrodsunif{2}}{} = \amargconstraint{\alice}{}{0}
\right\}, \\
\mathcal L &:= \left\{
\Big(\atargetp{1}{}{}, \atargetp{2}{}{}\Big)
\middle|
\exists \scdetstrat \textup{ s.t. } \atargetp{1}{}{} = \acaseone{\alice}{}{}{1},
\right. \nonumber\\
&\hspace{3.4cm}
\left.
\atargetp{2}{}{} = \acasetwo{\alice}{}{}{1},
\isource{\macrodsunif{2}}{} = \amargconstraint{\alice}{}{1}
\right\}.
\end{align}
\end{subequations}
\begin{claim}
$\closure{\mathcal R} = \closure{\mathcal L}$.
\begin{claimproof}
\Cref{lem:sc approx with marg} proved that $\mathcal R \subseteq \closure{\mathcal L}$. Indeed, in this notation, \cref{lem:sc approx with marg} reads: ``for all $r\in\mathcal R$, for all $\epsilon > 0$, there exists $l \in\mathcal L$ such that $\metric(r,l) < \epsilon$''. This implies in turn that $\closure{\mathcal R} \subseteq \closure{\mathcal L}$. Further more, it is clear that $\mathcal L \subseteq \mathcal R$ (deterministic probability tensors are a special case of probability tensors), which implies $\closure{\mathcal L} \subseteq \closure{\mathcal R}$. Altogether we see that $\closure{\mathcal R} = \closure{\mathcal L}$.
\end{claimproof}
\end{claim}

Let us now introduce the map $\phi : M \to \mathbb R$ defined as
\begin{equation}
\phi\left[ \Big(\atargetp{1}{}{},\atargetp{2}{}{} \Big) \right] := \frac{1}{2}\sum_{\outputa\in\{1,2\}} \left(\atargetp{1}{\outputa}{\outputa} + \atargetp{2}{\outputa}{\outputa}\right).
\end{equation}

\begin{claim}
$\phi$ is a continuous map.
\begin{claimproof}
It suffices to verify that
\begin{align}
\left|\phi\left[ \Big(\atargetp{1}{}{},\atargetp{2}{}{} \Big) \right]
- 
\phi\left[ \Big(\atargetpbis{1}{}{},\atargetpbis{2}{}{} \Big) \right] \right|
&\leq \frac{1}{2}\sum_{\outputa\in\{1,2\}} \left| \atargetp{1}{\outputa}{\outputa} - \atargetpbis{1}{\outputa}{\outputa} \right| +  \left| \atargetp{2}{\outputa}{\outputa} - \atargetpbis{2}{\outputa}{\outputa} \right| \\
&\leq \metric\left[ \Big(\atargetp{1}{}{},\atargetp{2}{}{} \Big),
\Big(\atargetpbis{1}{}{},\atargetpbis{2}{}{} \Big) \right]
\end{align}
which implies the continuity of $\phi$.
\end{claimproof}
\end{claim}
Now, using \cref{lem:basic topology}, we know that $\closure{\phi(\mathcal R)} = \closure{\phi(\closure{\mathcal R})}$ and $\closure{\phi(\mathcal L)} = \closure{\phi(\closure{\mathcal L})}$. Since $\closure{\mathcal R} = \closure{\mathcal L}$, we in fact have $\closure{\phi(\mathcal R)} = \closure{\phi(\mathcal L)}$. The result now follows easily: the original definition of $p^*$ in \eqref{eq:sleeping correlations tensor} can be rewritten in this notation as:
\begin{equation}
p^* = \sup_{u \in \phi(\mathcal R)} u = \sup_{u\in\closure{\phi(\mathcal R)}} u = \sup_{u\in\closure{\phi(\mathcal L)}} u = \sup_{u\in \phi(\mathcal L)} u,
\end{equation}
where the last supremum is the same optimization problem as equation \eqref{eq:sc det tensor}.
\end{proof}

\newpage
\section{Postselected inflation: proofs}
\label{app:post-selected inflation proofs}

\subsection{Main lemma}

We now restate and prove \cref{lem:average twonorm} (recall the definitions of the $p$-norms in \cref{def:p norms}).

\LemmaMainLemma*
\begin{proof}
Let us label the outcomes of the probability tensors with $\dsvalalpha = 1,\dots,k$.
\begin{subequations}
\begin{align}
\int\dd\hvval \isource{\hv}{\hvval} \twonorm{\isource{\targetpname}{} - \onestrat{\targetptildename}{}{\hvval}}^2
&=
\left| \int\dd\hvval \isource{\hv}{\hvval} \sum_{\dsvalalpha=1}^k \left[ \isource{\targetpname}{\dsvalalpha}^2 - 2 \isource{\targetpname}{\dsvalalpha}\onestrat{\targetptildename}{\dsvalalpha}{\hvval} + \onestrat{\targetptildename}{\dsvalalpha}{\hvval}^2 \right] \right|\\
&=
\left| \int\dd\hvval \isource{\hv}{\hvval} \sum_{\dsvalalpha=1}^k \left[
2 \isource{\targetpname}{\dsvalalpha}\left( \isource{\targetpname}{\dsvalalpha} - \onestrat{\targetptildename}{\dsvalalpha}{\hvval}\right) - \left(\isource{\targetpname}{\dsvalalpha}^2 - \onestrat{\targetptildename}{\dsvalalpha}{\hvval}^2 \right) \right] \right|\\
&=
\left| \sum_{\dsvalalpha=1}^k \left[
2 \isource{\targetpname}{\dsvalalpha}\left( \isource{\targetpname}{\dsvalalpha} - 
\centertikz{
\node[tensornode] (hv) {$\hv$};
\node[tensornode] (q) [above=\outcomevspace of hv] {$\targetptildename$};
\node[voidnode] (o) [above=\outcomevspace of q] {\indexstyle{\dsvalalpha}};
\drawleg (hv.north) -- (q.south);
\drawleg (q.north) -- (o.south);
}
\right) - \left(\isource{\targetpname}{\dsvalalpha}^2 - \cpdoubleq{\dsvalalpha}{\dsvalalpha} \right) \right] \right|\\
&\leq
2 \sum_{\dsvalalpha=1}^k \left| \isource{\targetpname}{\dsvalalpha} - 
\centertikz{
\node[tensornode] (hv) {$\hv$};
\node[tensornode] (q) [above=\outcomevspace of hv] {$\targetptildename$};
\node[voidnode] (o) [above=\outcomevspace of q] {\indexstyle{\dsvalalpha}};
\drawleg (hv.north) -- (q.south);
\drawleg (q.north) -- (o.south);
}
\right|
+ 
\sum_{i=1}^k \left|\isource{\targetpname}{\dsvalalpha}^2 - \cpdoubleq{\dsvalalpha}{\dsvalalpha}  \right| \\
&\leq
2 
\onenorm{
\centertikz{
\node[tensornode] (p1) {$p$};
\node[voidnode] (o1) [above=\outcomevspace of p1] {};
\drawleg (p1.north) -- (o1);
\node[tensornode] (p2) [right=5pt of p1] {$p$};
\node[margnode] (m) [above=\outcomevspace of p2] {};
\drawleg (p2.north) -- (m);
}
-
\centertikz{
\node[tensornode] (q1) at (0,0) {$\targetptildename$};
\node[voidnode] (o1) [above=\outcomevspace of q1] {};
\drawleg (q1.north) -- (o1.south);
\node[tensornode] (q2) at (20pt,0) {$\targetptildename$};
\node[margnode] (o2) [above=\outcomevspace of q2] {};
\drawleg (q2.north) -- (o2);
\node[tensornode] (hv) at (10pt,-30pt) {$\hv$};
\node[copynode] (hvcopy) [above=\outcomevspace of hv] {};
\drawleg (hvcopy) -- (q1.south);
\drawleg (hvcopy) -- (q2.south);
\drawleg (hv.north) -- (hvcopy);
}
}
+ 
\onenorm{
\isource{\targetpname}{} \isource{\targetpname}{
} - \cpdoubleq{}{}
} \\
&\hspace{-10pt}\overset{\textup{lem.}\ \ref{lem:onenorm marg}}{\leq} 3 \onenorm{
\isource{\targetpname}{}\isource{\targetpname}{}
-
\cpdoubleq{}{}
},
\end{align}
\end{subequations}
where we used the definition of the $1$- and $2$-norms (\cref{def:p norms}), the triangle inequality, and \cref{lem:onenorm marg}.
\end{proof}

\subsection{Certification}

Let us now prove \cref{th:certifying ps}.

\ThCertifyingPs*
\begin{proof}
Suppose that $\genmultopotargetps \in \outputdistribs{\genmultopo}$. Then, there exist deterministic probability tensors 
\begin{equation}
\label{eq:original strategies}
\left\{\gonestratdet{\alice^{(0)}_\pindex}{}{}\right\}_{\pindex=1}^{\pcount}
\end{equation}
which verify the condition of equation \eqref{eq:causal compat multiple config condition}. 
We added the superscript ``0'' to make clear that these are the original agent strategies, solving the causal compatibility problem, not to be confused with the agent strategies of \eqref{eq:prob gen inf set primitives} that are to solve the postselected inflation problem.

It remains to write down the probability tensors that solve the postselected inflation problem \eqref{eq:prob gen inf set primitives} to establish the inclusion relation \eqref{eq:causal compat certified}. The random variable $\isource{\hv}{}$ should actually distribute a tuple $\vec\hvval$ of $\ninf$ independent values sampled from $\ninf$ sources $\isource{\maxunif}{}$, which we can represent as, for any sequence $\vec\hvval$ of $\ninf$ real numbers in $[0,1]$,
\begin{equation}
\label{eq:solution lambda}
\isource{\hv}{\vec\hvval} \quad = \quad 
\centertikz{
\node[tensornode] (u1) {$\maxunif$};
\node[tensornode] (u2) [right=10pt of u1] {$\maxunif$};
\node[voidnode] (dots) [right=10pt of u2] {$\overset{(\ninf - 3)}{\dots}$};
\node[tensornode] (u3) [right=10pt of dots] {$\maxunif$};
\node[tuplenode] (tuple) [above right=30pt and 4pt of u2] { };
\node[voidnode] (tupledots) [below right=5pt and -5pt of tuple] {\indexstyle{\dots}};
\foreach \source in {u1,u2,u3}
	\drawleg (\source.north) -- (tuple.south);
\node[voidnode] (o) [above=15pt of tuple] {\indexstyle{\vec\hvval}};
\drawtupleleg (tuple.north) -- (o.south);
}.
\end{equation}
Then, the agents should use their tuple of values $\vec\dsvalalpha$, which is received from the sources $\isource{\dsunif}{}$ that they have access to, to \emph{select} which of the $\ninf$ source distributed by $\isource{\hv}{}$ they should use as inputs to the \emph{original} agent strategies. This can be represented as, for each $\pindex \in\{1,\dots,\pcount\}$,
\begin{equation}
\label{eq:solution a}
\geninfstratargs{\alice_\pindex}{}{\vec\dsvalalpha}{\vec\hvval} \quad := \quad 
\centertikz{
\node[detnode] (a) {$\alice^{(0)}_\pindex$};
\node[voidnode] (o) [above=\outcomevspace of a] { };
\drawleg (a.north) -- (o.south);
\node[selectnode] (select) [below=10pt of a] { };
\node[voidnode] (in) [left=10pt of select] {\indexstyle{\vec\dsvalalpha}};
\node[voidnode] (hvin) [below=10pt of select] {\indexstyle{\vec\hvval}};
\drawtupleleg (in.east) to [out=0,in=180] (select.west);
\drawtupleleg (hvin.north) to [out=90,in=270] (select.south);
\drawtupleleg (select.north) -- (a.south);
}.
\end{equation}
Recall the definition of the selector node $\centertikz{\node[selectnode] (a) {};}$ in \cref{sec:special tensors}.
It remains to argue why this construction implies the desired inclusion of equation \eqref{eq:causal compat certified}, that is, why this construction solves the postselected inflation problem of equation \eqref{eq:geninfcondition}. Let us fix $\topoindex\in\{1,\dots,\topocount\}$. Consider the left-hand side diagram of \eqref{eq:geninfcondition} but for a fixed value assignment for each of the outputs of the $\scount_\topoindex\cdot\ninfcons$ sources $\isource{\dsunif}{}$, such that this value assignment is compatible with the postselection, i.e., such that all the values therein are pairwise distinct. Let us refer to this assignment as a ``$\isource{\dsunif}{}$-conditioning'' for brevity.
Furthermore, let us refer to the left-hand side diagram as being formed from $\ninfcons$ ``groups'', where a group consists of $\scount_\topoindex$ sources $\isource{\dsunif}{}$ and $\npcount_\topoindex$ agents. 

\begin{claim}
with equations \eqref{eq:solution lambda} and \eqref{eq:solution a}, the left-hand side diagram of \eqref{eq:geninfcondition} under any $\isource{\dsunif}{}$-conditioning factorizes in the same way as the right-hand side diagram of \eqref{eq:geninfcondition} factorizes.
\begin{claimproof}
There is no value of the sources $\isource{\dsunif}{}$ of one group in common with any of the value of the sources $\isource{\dsunif}{}$ of any other group under any $\isource{\dsunif}{}$-conditioning thanks to the postselection.
Given \eqref{eq:solution a}, this means that the agents of one group are \emph{guaranteed} to only look at sources $\isource{\maxunif}{}$ distributed through $\isource{\hv}{}$ which are independent from those that the agents of any other group look at.
Since the only potential correlations between two groups, after $\isource{\dsunif}{}$-conditioning, would come from the shared source $\isource{\hv}{}$, the claim follows.
\end{claimproof}
\end{claim}
\begin{claim}
with equations \eqref{eq:solution lambda} and \eqref{eq:solution a}, the marginal of the left-hand side diagram of \eqref{eq:geninfcondition} under any $\isource{\dsunif}{}$-conditioning where all groups except one are ignored equals to $\macrogtargetp{\targetpname_\topoindex}$.
\begin{claimproof}
In this marginal, each agent will be sampling several sources $\isource{\maxunif}{}$ sent out by $\isource{\hv}{}$. Which of the $\ninf$ such sources $\isource{\maxunif}{}$ the agent samples depends on the values sent out by the subset of the sources $\isource{\dsunif}{}$ the agent has access to. Thus, all the agents that are receiving the value of a given source $\isource{\dsunif}{}$ will be sampling from the same source $\isource{\maxunif}{}$ sent out by $\isource{\hv}{}$, and so the corresponding inputs of the original strategies (equation \eqref{eq:original strategies}) will be connected to the same source $\isource{\maxunif}{}$, exactly as they should in the original causal compatibility diagram of \eqref{eq:causal compat multiple config condition}. Furthermore, because the values of the different sources $\isource{\dsunif}{}$ are all pairwise distinct under $\isource{\dsunif}{}$-conditioning (thanks to the postselection), it follows that the agents will never be using the same source $\isource{\maxunif}{}$ sent out by $\isource{\hv}{}$, except if they are connected to the same source $\isource{\dsunif}{}$.
Thus, the marginal under consideration will be equal to the tensor contraction of the left-hand side of equation \eqref{eq:causal compat multiple config condition}, which is, by assumption on the original agent strategies of \eqref{eq:original strategies}, equal to the right-hand side of equation \eqref{eq:causal compat multiple config condition}, which is just $\macrogtargetp{\targetpname_\topoindex}$.
\end{claimproof}
\end{claim}
Using the two claims together, we see that under $\isource{\dsunif}{}$-conditioning, the left-hand side of \eqref{eq:geninfcondition} equals the right-hand side. Averaging over the possible $\isource{\dsunif}{}$-conditioning, this property remains true, so that we obtain that \eqref{eq:geninfcondition} holds true, and $\topoindex$ was arbitrary in $\{1,\dots,\topocount\}$.
\end{proof}

\subsection{Hierarchy}

Let us now turn to \cref{th:inclusion relations}.

\ThInclusionRelations*

\begin{proof}
We first prove that the relation \eqref{eq:general inclusion} works in the case $\ninf' = \ninf$ and $\ninfcons' = \ninfcons - 1$, and then in the case $\ninf' < \ninf$ and $\ninfcons' = \ninfcons$. The general relation \eqref{eq:general inclusion} then follows by repeated iteration of the arguments.

\begin{claim}
the relation \eqref{eq:general inclusion} holds in the case $\ninf' = \ninf$ and $\ninfcons' = \ninfcons - 1$.
\begin{claimproof}
For any outcome distributions $\genmultopotargetps \in \geninfset{\genmultopo}{\ninf}{\ninfcons}$, consider equation \eqref{eq:geninfcondition} with $\ninfcons$ groups, as displayed. The marginal where the last group is ignored is simply the condition of equation \eqref{eq:geninfcondition} with $\ninfcons - 1$ groups. This is clear for the right-hand side. For the left-hand side, this follows from the fact that one can safely marginalize sources that are inputs to the postselection, that is, for any $k \in\mathbb{N}$ such that $k + 1 \leq n$,
\begin{equation}
\centertikz{
\node[voidnode] (i) {};
\node[tuplenode] (tuple) [below=10pt of i] {};
\drawtupleleg (tuple.north) -- (i.south);
\node[voidnode] (tupledots) [below=1pt of tuple] {\indexstyle{\dots}};
\node[voidnode] (dots) [below=30pt of tuple] {$\overset{(k-2)}{\dots}$};
\node[tensornode] (u1) [left=6pt of dots] {$\dsunif$};
\node[copynode] (u1copy) [above=\outcomevspace of u1] {};
\drawleg (u1.north) -- (u1copy);
\drawleg (u1copy) -- \AnchorOneThree{tuple}{south};
\node[tensornode] (uk) [right=6pt of dots] {$\dsunif$};
\node[copynode] (ukcopy) [above=\outcomevspace of uk] {};
\drawleg (uk.north) -- (ukcopy);
\drawleg (ukcopy) -- \AnchorThreeThree{tuple}{south};
\node[tensornode] (uk1) [right=25pt of uk] {$\dsunif$};
\node[copynode] (uk1copy) [above=\outcomevspace of uk1] {};
\drawleg (uk1.north) -- (uk1copy);
\node[psnode] (f) [above right=30pt and 3pt of uk] {\psdiffname{k+1}};
\node[voidnode] (fdots) [below left=1pt and -6pt of f] {\indexstyle{\dots}};
\drawleg (u1copy) -- (f.south west);
\drawleg (ukcopy) -- \AnchorTwoThree{f}{south};
\drawleg (uk1copy) -- \AnchorThreeThree{f}{south};
\node[margnode] (marg) [above right=10pt and 5pt of uk1copy] {};
\drawleg (uk1copy) -- (marg.south);
}
=
\centertikz{
\node[voidnode] (i) {};
\node[tuplenode] (tuple) [below=10pt of i] {};
\drawtupleleg (tuple.north) -- (i.south);
\node[voidnode] (tupledots) [below=1pt of tuple] {\indexstyle{\dots}};
\node[voidnode] (dots) [below=30pt of tuple] {$\overset{(k-2)}{\dots}$};
\node[tensornode] (u1) [left=6pt of dots] {$\dsunif$};
\node[copynode] (u1copy) [above=\outcomevspace of u1] {};
\drawleg (u1.north) -- (u1copy);
\drawleg (u1copy) -- \AnchorOneThree{tuple}{south};
\node[tensornode] (uk) [right=6pt of dots] {$\dsunif$};
\node[copynode] (ukcopy) [above=\outcomevspace of uk] {};
\drawleg (uk.north) -- (ukcopy);
\drawleg (ukcopy) -- \AnchorThreeThree{tuple}{south};
\node[psnode] (f) [above=30pt of uk] {\psdiffname{k}};
\drawleg (u1copy) -- (f.south west);
\drawleg (ukcopy) -- (f.south east);
\node[voidnode] (fdots) [below=1pt of f] {\indexstyle{\dots}};
}.
\end{equation}
Thus, the tensors of \eqref{eq:prob gen inf set primitives} that solve the problem of  \eqref{eq:prob gen inf set} with $\ninfcons$ groups also solve the problem of \eqref{eq:prob gen inf set} with $\ninfcons-1$ groups, and hence $\genmultopotargetps \in \geninfset{\genmultopo}{\ninf}{\ninfcons-1}$.
\end{claimproof}
\end{claim}
\begin{claim}
the relation \eqref{eq:general inclusion} holds in the case where $\ninf' < \ninf$ and $\ninfcons' = \ninfcons$.
\begin{claimproof}
Consider some $\genmultopotargetps \in \geninfset{\genmultopo}{\ninf}{\ninfcons}$, and let $\Big\{\geninfstrat{\alice_\pindex}\Big\}_{\pindex=1}^\pcount$ and $\isource{\hv}{}$ be some choice of tensors 
that solve equation \eqref{eq:prob gen inf set}. We will construct new tensors $\Big\{\geninfstrat{\alice'_\pindex}\Big\}_{\pindex=1}^\pcount$ and $\isource{\hv'}{}$ that solve the problem of \eqref{eq:prob gen inf set} with parameters $\ninf' < \ninf$ and $\ninfcons$, thus showing that also $\genmultopotargetps \in \geninfset{\genmultopo}{\ninf'}{\ninfcons}$. 
The construction is the following: let the new $\isource{\hv'}{}$ tensor be a tuple of the original $\isource{\hv}{}$ tensor together with a uniformly sampled permutation $\pi$ of $\ninf$ indices (such permutations are in one-to-one correspondence with the set $\{1,\dots,\ninf!\}$):
\begin{subequations}
\label{eq:primitives for smaller n}
\begin{equation}
\isource{\hv'}{\vec \hvval' = (\hvval, \pi)} = \isource{\hv}{\hvval}\isource{\macrodsunif{(\ninf!)}}{\pi}.
\end{equation}
The new agent strategies $\Big\{\geninfstrat{\alice'_\pindex}\Big\}_{\pindex=1}^\pcount$ are then obtained by letting the agents first apply the permutation $\pi$ on all the values they receive from the sources $\isource{\dsunif}{}$, and then using the original strategies:
\begin{equation}
\geninfstratargs{\alice'_\pindex}{}{\vec\dsvalalpha}{\vec\lambda'=(\lambda,\pi)}{}{}{}
=
\geninfstratargs{\alice_\pindex}{}{\pi(\vec\dsvalalpha)}{\hvval}{}{},
\end{equation}
\end{subequations}
where we introduced the notation $\pi(\vec\dsvalalpha)$ to denote the application of the permutation $\pi$ to all the components of $\vec\dsvalalpha$, e.g.\ if $\pi=(1\leftrightarrow2)(3\leftrightarrow4)$ and $\vec\dsvalalpha = (1,6,2,3)$, then $\pi(\vec\dsvalalpha) = (2,6,1,4)$.
It remains to proves that the choice of \eqref{eq:primitives for smaller n} does solve the problem of \eqref{eq:prob gen inf set} for $\ninf' < \ninf$. Using the terminology of the proof of \cref{th:certifying ps}, the diagram of the left-hand side of \eqref{eq:geninfcondition} with $\ninf$ replaced by $\ninf' < \ninf$, and under $\isource{\macrodsunif{\ninf'}}{}$-conditioning, i.e., under a choice of value assignment for all the sources $\isource{\macrodsunif{\ninf'}}{}$ where the value assignment is compatible with the postselection, does in fact already verifies the condition \eqref{eq:geninfcondition}. Thus, the average over all $\isource{\macrodsunif{\ninf'}}{}$-conditioning will also verify this condition \eqref{eq:geninfcondition}, and the claim follows.
Consider for instance the value assignment where the sources $\isource{\macrodsunif{\ninf'}}{}$ take the values $1,2,\dots, \scount_\topoindex\cdot \ninfcons$ (recall that by assumption, $\scount_\topoindex\cdot\ninfcons \leq \ninf'$). 
After all agents apply the random permutation $\pi$, the sources effectively take the values $\pi(1), \pi(2), \dots, \pi(\scount_\topoindex\cdot\ninfcons)$, which is a uniformly distributed tuple of values all pairwise distinct and in the range $\{1,\dots,\ninf\}$, exactly as those obtained from the $\scount_\topoindex\cdot\ninfcons$ sources $\isource{\dsunif}{}$ postselected with the tensor 
$\centertikz{
\node[psnode,inner sep=-3pt] (f) {\psdiffname{\scount_\topoindex\cdot\ninfcons}};
\node[voidnode] (i1) [below left=\outcomevspace and 3pt of f] {};
\node[voidnode] (i2) [right=3pt of i1] {};
\node[voidnode] (dots) [above right=-2pt and -1pt of i2] {\indexstyle{\dots}};
\node[voidnode] (i3) [below right=\outcomevspace and 3pt of f] {};
\drawdsource (i1.north) -- \AnchorOneThree{f}{south};
\drawdsource (i2.north) -- \AnchorTwoThree{f}{south};
\drawdsource (i3.north) -- \AnchorThreeThree{f}{south};
}$.
\end{claimproof}
\end{claim}
This concludes the proof.
\end{proof}

\subsection{Convergence}

The following lemma gives a useful relation between norms.
\begin{lemma}
\label{lem:onetwonorms}
For all $k\in\mathbb{N}$, $x \in \mathbb{R}^k$, it holds that
\begin{equation}
\onenorm{x} \leq \sqrt{k} \twonorm{x}.
\end{equation}
\end{lemma}
\begin{proof}
Let $y = (1,\dots,1) \in \mathbb{R}^k$, and $\tilde x = (|x_1|, \dots, |x_k|) \in\mathbb{R}^k$. We see that
\begin{equation}
\onenorm{x} = \sum_{i=1}^k y_i \tilde x_i \leq \twonorm{y}\twonorm{\tilde x} = \sqrt{k} \twonorm{x},
\end{equation}
where we used the Cauchy-Schwartz inequality for the canonical inner product of $\mathbb{R}^k$.
\end{proof}

We make formal the intuition that the postselection has almost no effect for very large $\ninf$ in the following lemma:

\begin{lemma}
\label{lem:postselection approx}
For any $\ninf, \scount \in \mathbb{N}$ with $\ninf \geq \scount$, for any probability tensor $\cpm{}{}$ with $\scount$ input legs, it holds that
\begin{equation}
\onenorm{
\centertikz{
\node[tensornode] (u1) {$\dsunif$};
\node[voidnode] (dots) [right=5pt of u1] {$\overset{(\scount-2)}{\dots}$};
\node[tensornode] (un) [right=5pt of dots]{$\dsunif$};
\node[copynode] (u1copy) [above=\outcomevspace of u1] {};
\node[copynode] (uncopy) [above=\outcomevspace of un] {};
\drawleg (u1.north) -- (u1copy);
\drawleg (un.north) -- (uncopy);
\node[tuplenode] (tuple) [above=20pt of dots] {};
\node[tensornode] (c) [above=10pt of tuple] {$M$};
\drawleg (c.north) -- ($(c.north) + (0,\outcomevspace)$);
\node[voidnode] [below=1pt of tuple] {\indexstyle{\dots}};
\drawleg (u1copy) -- \AnchorOneThree{tuple}{south};
\drawleg (uncopy) -- \AnchorThreeThree{tuple}{south};
\drawtupleleg (tuple.north) -- (c.south);
\node[psnode] (f) [above right=10pt and 5pt of uncopy] {\psdiffname{\scount}};
\node[voidnode] [below left=0pt and -12pt of f] {\indexstyle{\dots}};
\drawleg (u1copy) to [out=10,in=220] (f.south west);
\drawleg (uncopy) to [out=10,in=220] (f.south east);
}
-
\centertikz{
\node[tensornode] (u1) {$\dsunif$};
\node[voidnode] (dots) [right=5pt of u1] {$\overset{(\scount-2)}{\dots}$};
\node[tensornode] (un) [right=5pt of dots]{$\dsunif$};
\node[copynode] (u1copy) [above=\outcomevspace of u1] {};
\node[copynode] (uncopy) [above=\outcomevspace of un] {};
\drawleg (u1.north) -- (u1copy);
\drawleg (un.north) -- (uncopy);
\node[tuplenode] (tuple) [above=20pt of dots] {};
\node[tensornode] (c) [above=10pt of tuple] {$M$};
\drawleg (c.north) -- ($(c.north) + (0,\outcomevspace)$);
\node[voidnode] [below=1pt of tuple] {\indexstyle{\dots}};
\drawleg (u1copy) -- \AnchorOneThree{tuple}{south};
\drawleg (uncopy) -- \AnchorThreeThree{tuple}{south};
\drawtupleleg (tuple.north) -- (c.south);
}
}
\leq
\frac{\scount(\scount - 1)}{\ninf} + \bigo{\frac{1}{\ninf^2}}.
\end{equation}
\end{lemma}
\begin{proof}
We label the outcomes of the tensor $\cpm{}{}$ with $a$, and the combined inputs form the $\isource{\dsunif}{}$ sources with $\vec\dsvalalpha$. We further use the notation `` $\vec\dsvalalpha \allneq$ '' in case the tuple of inputs $\vec\dsvalalpha$ is compatible with the postselection, i.e., all the components of $\vec\dsvalalpha$ are pairwise distinct, and the notation `` $\vec\dsvalalpha \notallneq$ '' otherwise, i.e., if at least two components of $\vec\dsvalalpha$ are equal. Note that with $\scount$ sources outputting $\ninf$ distinct values, it holds that
\begin{equation}
|\{\vec\dsvalalpha \allneq \}| = \ninf(\ninf-1)\cdots(\ninf-(\scount - 1)) = \frac{\ninf!}{(\ninf - \scount)!}.
\end{equation}
Expanding the norm and the source contractions, we obtain
\begin{subequations}
\begin{align}
\onenorm{
\centertikz{
\node[tensornode] (u1) {$\dsunif$};
\node[voidnode] (dots) [right=5pt of u1] {$\overset{(\scount-2)}{\dots}$};
\node[tensornode] (un) [right=5pt of dots]{$\dsunif$};
\node[copynode] (u1copy) [above=\outcomevspace of u1] {};
\node[copynode] (uncopy) [above=\outcomevspace of un] {};
\drawleg (u1.north) -- (u1copy);
\drawleg (un.north) -- (uncopy);
\node[tuplenode] (tuple) [above=20pt of dots] {};
\node[tensornode] (c) [above=10pt of tuple] {$M$};
\drawleg (c.north) -- ($(c.north) + (0,\outcomevspace)$);
\node[voidnode] [below=1pt of tuple] {\indexstyle{\dots}};
\drawleg (u1copy) -- \AnchorOneThree{tuple}{south};
\drawleg (uncopy) -- \AnchorThreeThree{tuple}{south};
\drawtupleleg (tuple.north) -- (c.south);
\node[psnode] (f) [above right=10pt and 5pt of uncopy] {\psdiffname{\scount}};
\node[voidnode] [below left=0pt and -12pt of f] {\indexstyle{\dots}};
\drawleg (u1copy) to [out=10,in=220] (f.south west);
\drawleg (uncopy) to [out=10,in=220] (f.south east);
}
-
\centertikz{
\node[tensornode] (u1) {$\dsunif$};
\node[voidnode] (dots) [right=5pt of u1] {$\overset{(\scount-2)}{\dots}$};
\node[tensornode] (un) [right=5pt of dots]{$\dsunif$};
\node[copynode] (u1copy) [above=\outcomevspace of u1] {};
\node[copynode] (uncopy) [above=\outcomevspace of un] {};
\drawleg (u1.north) -- (u1copy);
\drawleg (un.north) -- (uncopy);
\node[tuplenode] (tuple) [above=20pt of dots] {};
\node[tensornode] (c) [above=10pt of tuple] {$M$};
\drawleg (c.north) -- ($(c.north) + (0,\outcomevspace)$);
\node[voidnode] [below=1pt of tuple] {\indexstyle{\dots}};
\drawleg (u1copy) -- \AnchorOneThree{tuple}{south};
\drawleg (uncopy) -- \AnchorThreeThree{tuple}{south};
\drawtupleleg (tuple.north) -- (c.south);
}
}
=
\sum_{\outputa}
\left|
\sum_{\vec\dsvalalpha \allneq}
\cpm{\outputa}{\vec\dsvalalpha}
\frac{(\ninf- \scount)!}{\ninf!}
-
\sum_{\vec\dsvalalpha}
\cpm{\outputa}{\vec\dsvalalpha}
\frac{1}{\ninf^{\scount}}
\right|
\end{align}
\begin{align}
&=
\sum_{\outputa}
\left|
\sum_{\vec\dsvalalpha \allneq}
\cpm{\outputa}{\vec\dsvalalpha}
\left(
\frac{(\ninf- \scount)!}{\ninf!}
-
\frac{1}{\ninf^{\scount}}
\right)
-
\sum_{\vec\dsvalalpha \notallneq}
\cpm{\outputa}{\vec\dsvalalpha}
\frac{1}{\ninf^{\scount}}
\right| \\
&\leq 
\sum_{\vec\dsvalalpha \allneq}
\sum_{\outputa}
\cpm{\outputa}{\vec\dsvalalpha}
\left(
\frac{(\ninf- \scount)!}{\ninf!}
-
\frac{1}{\ninf^{\scount}}
\right)
+
\sum_{\vec\dsvalalpha \notallneq}
\sum_{\outputa}
\cpm{\outputa}{\vec\dsvalalpha}
\frac{1}{\ninf^{\scount}} \\
&=
\frac{\ninf!}{(\ninf - \scount)!}\left(
\frac{(\ninf- \scount)!}{\ninf!}
-
\frac{1}{\ninf^{\scount}}
\right)
+ \left(\ninf^{\scount} - \frac{\ninf!}{(\ninf - \scount)!}\right)\frac{1}{\ninf^{\scount}} \\
&=
2\left( 1 - \frac{\ninf!}{\ninf^{\scount}(\ninf-\scount)!}\right)\\
&=
2\left( 1 - \prod_{k=1}^{\scount-1} \left(1 - \frac{k}{\ninf}\right)\right) \\
&=
\frac{2}{\ninf}\sum_{k=1}^{\scount-1} k + \bigo{\frac{1}{\ninf^2}} \\
&= \frac{\scount(\scount -1)}{\ninf} + \bigo{\frac{1}{\ninf^2}},
\end{align}
\end{subequations}
as expected.
\end{proof}

We now prove \cref{th:convergence}.

\ThConvergence*

\begin{proof}
Let us make implicit the infimum constraint for better readability. We have, using \cref{lem:onetwonorms} for $\mathbb{R}^{d_\topoindex}$, then for $\mathbb{R}^\topocount$, and then using the monotonicity of the square root:
\begin{subequations}
\begin{align}
\inf_{
\scalebox{0.95}{$
\scriptsize\big(\macrogtargetp{\targetptildename_\topoindex}\big)_{\topoindex}
$}
}
\frac{1}{\topocount}\sum_{\topoindex=1}^\topocount
\onenorm{\macrogtargetp{\targetpname_\topoindex} - \macrogtargetp{\targetptildename_\topoindex}}
&\leq 
\inf_{
\scalebox{0.95}{$
\scriptsize\big(\macrogtargetp{\targetptildename_\topoindex}\big)_{\topoindex}
$}
}
\frac{1}{\topocount}
\underbrace{
\sum_{\topoindex=1}^{\topocount} \sqrt{d_\topoindex} \twonorm{\macrogtargetp{\targetpname_\topoindex} - \macrogtargetp{\targetptildename_\topoindex}}
}_{=: \onenorm{x},\textup{ for some } x \in\mathbb R^{C}} \\
&\leq 
\inf_{
\scalebox{0.95}{$
\scriptsize\big(\macrogtargetp{\targetptildename_\topoindex}\big)_{\topoindex}
$}
}
\frac{1}{\topocount}\cdot 
\underbrace{
\sqrt{\topocount}
\sqrt{
\sum_{\topoindex=1}^{\topocount} d_\topoindex \twonorm{\macrogtargetp{\targetpname_\topoindex} - \macrogtargetp{\targetptildename_\topoindex}}^2
}}_{= \sqrt{C}\twonorm{x}} \\
&=
\sqrt{
\inf_{
\scalebox{0.95}{$
\scriptsize\big(\macrogtargetp{\targetptildename_\topoindex}\big)_{\topoindex}
$}
}
\frac{1}{\topocount}\sum_{\topoindex=1}^{\topocount} d_\topoindex \twonorm{\macrogtargetp{\targetpname_\topoindex} - \macrogtargetp{\targetptildename_\topoindex}}^2
}.\label{eq:temp convergence bound}
\end{align}
\end{subequations}
We now introduce some notation. We know that $\genmultopotargetps$ is in $\geninfset{\genmultopo}{\ninf}{\ninfcons}$ with $\ninfcons \geq 2$, so thanks to \cref{th:inclusion relations}, we have in particular $\genmultopotargetps \in \geninfset{\genmultopo}{\ninf}{2}$: consider the tensors $\Big\{\geninfstrat{\alice_\pindex}\Big\}_{\pindex=1}^\pcount$ and $\isource{\hv}{}$ of equation \eqref{eq:prob gen inf set primitives} that establish that $\genmultopotargetps \in \geninfset{\genmultopo}{\ninf}{2}$.
Define, for all $\hvval$ in the output domain of the source $\isource{\hv}{}$ and for all $\topoindex$, the distributions
\begin{equation}
\label{eq:def cpq}
\cpq
:=
\centertikz{
%
%
%
%
%
\node[voidnode] (pdots3) [right=140pt of pdots2] {$\dots$};
\node[detnode] (a3) [left=3pt of pdots3] {$\alice_{\pmap_\topoindex(1)}$};
\node[voidnode] (outa3) [above=\outcomevspace of a3] { };
\node[detnode] (b3) [right=3pt of pdots3] {$\alice_{\pmap_\topoindex(\npcount_\topoindex)}$};
\node[voidnode] (outb3) [above=\outcomevspace of b3] { };
\drawoutcomeleg (a3.north) -- (outa3.south);
\drawoutcomeleg (b3.north) -- (outb3.south); 
%
\node[selectnode] (sa3) [below=20pt of a3] { };
\node[selectnode] (sb3) [below=20pt of b3] { };
\node[voidnode] (sdots3) [right=7pt of sa3] {\indexstyle{\dots}};
\drawtupledsource (sa3.north) -- (a3.south);
\drawtupledsource (sb3.north) -- (b3.south);
%
\node[voidnode] (ia3) [left=6pt of sa3] {\indexstyle{\cmap_\topoindex(1)}};
\node[voidnode] (ib3) [left=6pt of sb3] {\indexstyle{\cmap_\topoindex(\npcount_\topoindex)}};
\drawtupleoutcomeleg (ia3.east) -- (sa3.west);
\drawtupleoutcomeleg (ib3.east) -- (sb3.west);
%
\node[copynode] (copy3) [below=55pt of pdots3] { };
\drawtupledsource (copy3) to [out=160,in=270] (sa3.south);
\drawtupledsource (copy3) to [out=20,in=270] (sb3.south);
\node[voidnode] (copydots3) [above=2pt of copy3] {\indexstyle{\dots}};
%
\node[tuplenode] (tuple3) [below=15pt of copy3] { };
\node[voidnode] (tupledots3) [below=6pt of tuple3] {\indexstyle{\dots}};
\drawtupledsource (tuple3.north) -- (copy3.south);
%
\node[voidnode] (sourcedots3) [below=20pt of tuple3] {$\overset{(\scount_\topoindex-2)}{\dots}$};
\node[tensornode] (beta3) [below left=-12pt and 10pt of sourcedots3] {$\dsunif$};
\node[tensornode] (gamma3) [below right=-12pt and 10pt of sourcedots3] {$\dsunif$};
\node[copynode] (copybeta3) [above=\outcomevspace of beta3] {};
\node[copynode] (copygamma3) [above=\outcomevspace of gamma3] {};
\drawdsource (beta3.north) -- (copybeta3);
\drawdsource (gamma3.north) -- (copygamma3);
\drawdsource (copybeta3) -- \AnchorOneThree{tuple3}{south};
\drawdsource (copygamma3) -- \AnchorThreeThree{tuple3}{south};
%
%
%
%
\node[voidnode] (hv) [above right=10pt and 70pt of tuple3] {\indexstyle{\hvval}};
\node[copynode] (copyhv) [above=\outcomevspace of hv] {};
\drawleg (hv.north) -- (copyhv);
\foreach \player/\outangle/\inangle in {a3/180/280,b3/180/270}
	\drawdashedleg (copyhv) to [out=\outangle,in=\inangle] ($(\player.south west)!\AnchorTwoTwoFactor!(\player.south east)$);
\node[voidnode] (dotshv)  [above left=0pt and 18pt of copyhv] {\indexstyle{\dots}};
}.
\end{equation}
%
By construction, for any $\hvval$, it holds that
\begin{equation}
\left( \cpq \right)_{\topoindex=1}^{\topocount} \in \outputdistribs{\genmultopo}.
\end{equation}
Thus, we can upper bound the infimum of \eqref{eq:temp convergence bound} by the convex combination
\begin{align}
\label{eq:temp convergence bound 2}
\eqref{eq:temp convergence bound} 
\leq
\sqrt{
\int \dd\hvval \isource{\hv}{\hvval} \frac{1}{\topocount}\sum_{\topoindex=1}^{\topocount} d_\topoindex 
\twonorm{\macrogtargetp{\targetpname_\topoindex} - \cpq}^2
}.
\end{align}
Now, using \cref{lem:average twonorm}:
\begin{align}
\label{eq:temp convergence bound 3}
\eqref{eq:temp convergence bound 2} \leq
\sqrt{
\frac{3}{\topocount}
\sum_{\topoindex=1}^\topocount d_\topoindex
\onenorm{
\macrogtargetp{\targetpname_\topoindex}\macrogtargetp{\targetpname_\topoindex} 
-
\centertikz{
\node[tensornode] (q1) at (0,0) {$\targetptildename_\topoindex$};
\node[voidnode] (a1) [above left=\outcomevspace and 0pt of q1] { };
\node[voidnode] (b1) [right=0pt of a1] { };
\node[voidnode] (dots1) [below right=-1.5pt and 0pt of b1] {\indexstyle{\dots}};
\node[voidnode] (c1) [above right=\outcomevspace and 0pt of q1] { };
\drawoutcomeleg \AnchorOneThree{q1}{north} -- (a1.south);
\drawoutcomeleg \AnchorTwoThree{q1}{north} -- (b1.south);
\drawoutcomeleg \AnchorThreeThree{q1}{north} -- (c1.south);
\node[tensornode] (q2) at (25pt,0pt) {$\targetptildename_\topoindex$};
\node[voidnode] (a2) [above left=\outcomevspace and 0pt of q2] { };
\node[voidnode] (b2) [right=0pt of a2] { };
\node[voidnode] (dots2) [below right=-1.5pt and 0pt of b2] {\indexstyle{\dots}};
\node[voidnode] (c2) [above right=\outcomevspace and 0pt of q2] { };
\drawoutcomeleg \AnchorOneThree{q2}{north} -- (a2.south);
\drawoutcomeleg \AnchorTwoThree{q2}{north} -- (b2.south);
\drawoutcomeleg \AnchorThreeThree{q2}{north} -- (c2.south);
\node[tensornode] (hv) at (12.5pt,-35pt) {$\hv$};
\node[copynode] (hvcopy) [above=\outcomevspace of hv] {};
\drawleg (hv.north) -- (hvcopy);
\drawleg (hvcopy) -- (q1.south);
\drawleg (hvcopy) -- (q2.south);
}
}
}.
\end{align}
Let us now introduce, for all $\topoindex =1,\dots,\topocount$, a new probability tensor $\cpmc{}{}$ that allows us to rewrite the constraint of equation \eqref{eq:geninfcondition} as
\begin{equation}
\label{eq:def mc}
\centertikz{
\node[tensornode] (u1) {$\dsunif$};
\node[voidnode] (dots) [right=5pt of u1] {$\overset{(2\scount_\topoindex-2)}{\dots}$};
\node[tensornode] (un) [right=5pt of dots]{$\dsunif$};
\node[copynode] (u1copy) [above=\outcomevspace of u1] {};
\node[copynode] (uncopy) [above=\outcomevspace of un] {};
\drawleg (u1.north) -- (u1copy);
\drawleg (un.north) -- (uncopy);
\node[tuplenode] (tuple) [above=20pt of dots] {};
\node[tensornode] (c) [above=10pt of tuple] {$M_\topoindex$};
\drawtupleleg (c.north) -- ($(c.north) + (0,\outcomevspace)$);
\node[voidnode] [below=1pt of tuple] {\indexstyle{\dots}};
\drawleg (u1copy) -- \AnchorOneThree{tuple}{south};
\drawleg (uncopy) -- \AnchorThreeThree{tuple}{south};
\drawtupleleg (tuple.north) -- (c.south);
\node[psnode] (f) [above right=10pt and 5pt of uncopy] {\psdiffname{2\scount_\topoindex}};
\node[voidnode] [below left=0pt and -12pt of f] {\indexstyle{\dots}};
\drawleg (u1copy) to [out=10,in=220] (f.south west);
\drawleg (uncopy) to [out=10,in=220] (f.south east);
}
=
\macrogtargetp{\targetpname_\topoindex}\macrogtargetp{\targetpname_\topoindex},
\end{equation}
and looking back at the definition of $\cpq$ in equation \eqref{eq:def cpq}, we also have
\begin{equation}
\label{eq:def mc 2}
\centertikz{
\node[tensornode] (u1) {$\dsunif$};
\node[voidnode] (dots) [right=5pt of u1] {$\overset{(2\scount_\topoindex-2)}{\dots}$};
\node[tensornode] (un) [right=5pt of dots]{$\dsunif$};
\node[copynode] (u1copy) [above=\outcomevspace of u1] {};
\node[copynode] (uncopy) [above=\outcomevspace of un] {};
\drawleg (u1.north) -- (u1copy);
\drawleg (un.north) -- (uncopy);
\node[tuplenode] (tuple) [above=20pt of dots] {};
\node[tensornode] (c) [above=10pt of tuple] {$M_\topoindex$};
\drawtupleleg (c.north) -- ($(c.north) + (0,\outcomevspace)$);
\node[voidnode] [below=1pt of tuple] {\indexstyle{\dots}};
\drawleg (u1copy) -- \AnchorOneThree{tuple}{south};
\drawleg (uncopy) -- \AnchorThreeThree{tuple}{south};
\drawtupleleg (tuple.north) -- (c.south);
}
=
\centertikz{
\node[tensornode] (q1) at (0,0) {$\targetptildename_\topoindex$};
\node[voidnode] (a1) [above left=\outcomevspace and 0pt of q1] { };
\node[voidnode] (b1) [right=0pt of a1] { };
\node[voidnode] (dots1) [below right=-1.5pt and 0pt of b1] {\indexstyle{\dots}};
\node[voidnode] (c1) [above right=\outcomevspace and 0pt of q1] { };
\drawoutcomeleg \AnchorOneThree{q1}{north} -- (a1.south);
\drawoutcomeleg \AnchorTwoThree{q1}{north} -- (b1.south);
\drawoutcomeleg \AnchorThreeThree{q1}{north} -- (c1.south);
\node[tensornode] (q2) at (25pt,0pt) {$\targetptildename_\topoindex$};
\node[voidnode] (a2) [above left=\outcomevspace and 0pt of q2] { };
\node[voidnode] (b2) [right=0pt of a2] { };
\node[voidnode] (dots2) [below right=-1.5pt and 0pt of b2] {\indexstyle{\dots}};
\node[voidnode] (c2) [above right=\outcomevspace and 0pt of q2] { };
\drawoutcomeleg \AnchorOneThree{q2}{north} -- (a2.south);
\drawoutcomeleg \AnchorTwoThree{q2}{north} -- (b2.south);
\drawoutcomeleg \AnchorThreeThree{q2}{north} -- (c2.south);
\node[tensornode] (hv) at (12.5pt,-35pt) {$\hv$};
\node[copynode] (hvcopy) [above=\outcomevspace of hv] {};
\drawleg (hv.north) -- (hvcopy);
\drawleg (hvcopy) -- (q1.south);
\drawleg (hvcopy) -- (q2.south);
}.
\end{equation}
Using equations \eqref{eq:def mc} and \eqref{eq:def mc 2} together with \cref{lem:postselection approx}, we have that
\begin{equation}
\onenorm{
\macrogtargetp{\targetpname_\topoindex}\macrogtargetp{\targetpname_\topoindex} 
-
\centertikz{
\node[tensornode] (q1) at (0,0) {$\targetptildename_\topoindex$};
\node[voidnode] (a1) [above left=\outcomevspace and 0pt of q1] { };
\node[voidnode] (b1) [right=0pt of a1] { };
\node[voidnode] (dots1) [below right=-1.5pt and 0pt of b1] {\indexstyle{\dots}};
\node[voidnode] (c1) [above right=\outcomevspace and 0pt of q1] { };
\drawoutcomeleg \AnchorOneThree{q1}{north} -- (a1.south);
\drawoutcomeleg \AnchorTwoThree{q1}{north} -- (b1.south);
\drawoutcomeleg \AnchorThreeThree{q1}{north} -- (c1.south);
\node[tensornode] (q2) at (25pt,0pt) {$\targetptildename_\topoindex$};
\node[voidnode] (a2) [above left=\outcomevspace and 0pt of q2] { };
\node[voidnode] (b2) [right=0pt of a2] { };
\node[voidnode] (dots2) [below right=-1.5pt and 0pt of b2] {\indexstyle{\dots}};
\node[voidnode] (c2) [above right=\outcomevspace and 0pt of q2] { };
\drawoutcomeleg \AnchorOneThree{q2}{north} -- (a2.south);
\drawoutcomeleg \AnchorTwoThree{q2}{north} -- (b2.south);
\drawoutcomeleg \AnchorThreeThree{q2}{north} -- (c2.south);
\node[tensornode] (hv) at (12.5pt,-35pt) {$\hv$};
\node[copynode] (hvcopy) [above=\outcomevspace of hv] {};
\drawleg (hv.north) -- (hvcopy);
\drawleg (hvcopy) -- (q1.south);
\drawleg (hvcopy) -- (q2.south);
}
}
\leq
\frac{2\scount_\topoindex(2\scount_\topoindex -1)}{\ninf} + \bigo{\frac{1}{\ninf^2}}.
\end{equation}
Inserting this bound into equation \eqref{eq:temp convergence bound 3} yields the desired result.
\end{proof}

The \cref{corollary:convergence} is now easy to obtain:

\CorollaryConvergence*
\begin{proof}
\Cref{th:certifying ps} already implies that
\begin{equation}
\outputdistribs{\genmultopo} \subseteq \bigcap_{\ninf=\ninf_0}^\infty \geninfset{\genmultopo}{\ninf}{\ninfcons}.
\end{equation}
Now, let
\begin{equation}
\genmultopotargetps \in \bigcap_{\ninf=\ninf_0}^\infty \geninfset{\genmultopo}{\ninf}{\ninfcons}.
\end{equation}
For all $\epsilon > 0$, choose $\ninf$ sufficiently large such that the right-hand side of equation \eqref{eq:convergence rate} is less than or equal to $\epsilon$. Since we have in particular that
\begin{equation}
\genmultopotargetps \in \geninfset{\genmultopo}{\ninf}{\ninfcons},
\end{equation}
\cref{th:convergence} implies that
\begin{equation}
\label{eq:corollary temp 1}
\inf_{
\scalebox{0.95}{$
\scriptsize\big(\macrogtargetp{\targetptildename_\topoindex}\big)_{\topoindex=1}^\topocount\in\outputdistribs{\genmultopo}
$}
}
\metric\left[
\genmultopotargetps,
\left(\macrogtargetp{\targetptildename_\topoindex}\right)_{\topoindex=1}^\topocount
\right]
\leq \epsilon.
\end{equation}
Since equation \eqref{eq:corollary temp 1} holds for all $\epsilon > 0$, and since the metric $\metric$ is positive definite, we must have in fact
\begin{equation}
\inf_{
\scalebox{0.95}{$
\scriptsize\big(\macrogtargetp{\targetptildename_\topoindex}\big)_{\topoindex=1}^\topocount\in\outputdistribs{\genmultopo}
$}
}
\metric\left[
\genmultopotargetps,
\left(\macrogtargetp{\targetptildename_\topoindex}\right)_{\topoindex=1}^\topocount
\right]
= 0.
\end{equation}
This is equivalent to the statement that there exists a sequence $\left(l_k \in \outputdistribs{\genmultopo}\right)_{k\in\mathbb N}$ such that
\begin{equation}
\lim_{k\to\infty} l_k = \genmultopotargetps
\end{equation}
in the metric $\metric$, so that
\begin{equation}
\bigcap_{\ninf=\ninf_0}^\infty \geninfset{\genmultopo}{\ninf}{\ninfcons}
\subseteq \closure{\outputdistribs{\genmultopo}}
\end{equation}
holds.
\end{proof}

\newpage
\section{Correlated Sleeper: additional material}

\subsection{Parametrization}
\label{sec:sc parametrization}


\LemmaScLambdas*

\begin{proof}
First, the symmetry of each $\atargetp{\topoindex}{}{}$, apparent from \eqref{eq:sc det tensor param}, implies
\begin{equation}
\label{eq:proof lambda sym}
\atargetp{\topoindex}{1}{2} = \atargetp{\topoindex}{2}{1}.
\end{equation}
Additionally, the marginal constraint implied by the last equality of \eqref{eq:sc det tensor param} yields
\begin{subequations}
\begin{align}
\label{eq:proof lambda marg}
\atargetp{\topoindex}{1}{1} + \atargetp{\topoindex}{1}{2} &= \frac{1}{2}, \\
\atargetp{\topoindex}{2}{1} + \atargetp{\topoindex}{2}{2} &= \frac{1}{2},
\end{align}
\end{subequations}
which together with \eqref{eq:proof lambda sym} implies that
\begin{equation}
\atargetp{\topoindex}{1}{1} = \atargetp{\topoindex}{2}{2}.
\end{equation}
Thus, the only free parameter in the distribution $\atargetp{\topoindex}{}{}$ is $\lambda_\topoindex := \atargetp{\topoindex}{1}{1}$. It is clear from \eqref{eq:proof lambda marg} that
\begin{equation}
0 \leq \atargetp{\topoindex}{1}{1} \leq \frac{1}{2},
\end{equation}
which concludes the proof.
\end{proof}

\LemmaBoundLambdas*
\begin{proof}
To see this, consider the $1,2$ component of the distribution:
\begin{align}
\label{eq:proof lambda pre inner product}
\atargetp{\topoindex}{1}{2} &= \int \dd \alpha f_1^{(\topoindex)}(\alpha) f_2^{(\topoindex)}(\alpha),
\end{align}
where we defined
\begin{equation}
f_i^{(\topoindex)}(\alpha) := \left\{
\begin{aligned}
\centertikz{
\node[detnode] (a) {$\alice$};
\node[voidnode] (o) [above=\outcomevspace of a] {\indexstyle{i}};
\drawleg (a.north) -- (o.south);
\node[tensornode] (u) [below right=10pt and -5pt of a] {$\maxunif$};
\drawleg (u.north) -- \AnchorTwoTwo{a}{south};
\node[voidnode] (i) [below left=10pt and 0pt of a] {\indexstyle{\alpha}};
\drawleg (i.north) -- \AnchorOneTwo{a}{south};
}, 
\quad &\textup{ if }\topoindex = 1,\\
\centertikz{
\node[detnode] (a) {$\alice$};
\node[voidnode] (o) [above=\outcomevspace of a] {\indexstyle{i}};
\drawleg (a.north) -- (o.south);
\node[tensornode] (u) [below left=10pt and -5pt of a] {$\maxunif$};
\drawleg (u.north) -- \AnchorOneTwo{a}{south};
\node[voidnode] (i) [below right=10pt and 0pt of a] {\indexstyle{\alpha}};
\drawleg (i.north) -- \AnchorTwoTwo{a}{south};
}, \quad &\textup{ if }\topoindex = 2.
\end{aligned}
\right.
\end{equation}
The right-hand side of equation \eqref{eq:proof lambda pre inner product} is an inner product on the function space $L^2([0,1])$, so that we can apply the Cauchy-Schwartz inequality:
\begin{equation}
\atargetp{\topoindex}{1}{2} = \left< f_1^{(\topoindex)} \middle| f_2^{(\topoindex)} \right >_{L^2([0,1])} \leq \twonorm{f_1^{(\topoindex)}}\twonorm{f_2^{(\topoindex)}}.
\end{equation}
Now, these norms can be evaluated explicitly:
\begin{align}
\twonorm{f_i^{(\topoindex)}}^2 = \int \dd \alpha f_i^{(\topoindex)}(\alpha) f_i^{(\topoindex)}(\alpha) = \atargetp{\topoindex}{i}{i} = \atargetp{\topoindex}{1}{1}.
\end{align}
We hence have
\begin{equation}
\atargetp{\topoindex}{1}{2} \leq \atargetp{\topoindex}{1}{1},
\end{equation}
so that
\begin{equation}
\atargetp{\topoindex}{1}{1} = \frac{1}{2} - \atargetp{\topoindex}{1}{2} \geq \frac{1}{2} - \atargetp{\topoindex}{1}{1}
\end{equation}
which implies
\begin{equation}
\atargetp{\topoindex}{1}{1} \geq \frac{1}{4}.
\end{equation}
This inequality expresses the intuitive fact that the best $A$ can do to obtain opposite output bits across the two rounds is to ignore her inputs and output a random bit instead.
\end{proof}

\subsection{Explicit linear programs}
\label{sec:app linear program explicit}

\paragraph{Explicit feasibility problem.} Continuing with the notation of \cref{sec:applications}, we say that $(\lambda_1,\lambda_2) \in \appinffeasibleparam$ if and only if there exists
\begin{equation}
\left\{\isource{\hv}{\appmat}\right\}_{\appmat\in\appmatsetreduced}
\end{equation}
such that (the distributions $\atargetp{\topoindex}{}{}$ for $\topoindex\in\{1,2\}$ are obtained from $(\lambda_1,\lambda_2)$ as in \cref{lem:sc lambdas}, the indices $\dsvalalpha_1,\dots,\dsvalbeta_4$ are always in the set $\{1,\dots,4\}$)
\begin{subequations}
\begin{align}
\forall \appmat&\in\appmatsetreduced \st \isource{\hv}{\appmat} \geq 0, \quad \sum_{\appmat' \in \appmatsetreduced} \isource{\hv}{\appmat'} = 1, \\
\forall (\outputa_1,\outputa_2,\outputa_3,\outputa_4) &\in \left\{(1,1,1,1),(1,1,1,2),(1,1,2,2),(1,2,1,2),(1,2,2,2)\right\} \st \nonumber \\
\sum_{\appmat\in\appmatsetreduced} \isource{\hv}{\appmat} \frac{1}{12\cdot 4!} 
&\sum_{\substack{
\dsvalalpha_1,\dsvalalpha_2 \allneq \\
\dsvalbeta_1,\dots,\dsvalbeta_4 \allneq
}}
\delta(\appmat_{\dsvalalpha_1,\dsvalbeta_1} = \outputa_1)
\delta(\appmat_{\dsvalalpha_1,\dsvalbeta_2} = \outputa_2)
\delta(\appmat_{\dsvalalpha_2,\dsvalbeta_3} = \outputa_3)
\delta(\appmat_{\dsvalalpha_2,\dsvalbeta_4} = \outputa_4) \nonumber\\
&\qquad= \atargetp{1}{\outputa_1}{\outputa_2}\atargetp{1}{\outputa_3}{\outputa_4}, \\
\forall (\outputa_1,\outputa_2,\outputa_3,\outputa_4) &\in \left\{(1,1,1,1),(1,1,1,2),(1,1,2,2),(1,2,1,2),(1,2,2,2)\right\} \st \nonumber \\
\sum_{\appmat\in\appmatsetreduced} \isource{\hv}{\appmat} \frac{1}{4!\cdot 12} 
&\sum_{\substack{
\dsvalalpha_1,\dots,\dsvalalpha_4 \allneq \\
\dsvalbeta_1,\dsvalbeta_2 \allneq
}}
\delta(\appmat_{\dsvalalpha_1,\dsvalbeta_1} = \outputa_1)
\delta(\appmat_{\dsvalalpha_2,\dsvalbeta_1} = \outputa_2)
\delta(\appmat_{\dsvalalpha_3,\dsvalbeta_2} = \outputa_3)
\delta(\appmat_{\dsvalalpha_4,\dsvalbeta_2} = \outputa_4) \nonumber\\
&\qquad= \atargetp{2}{\outputa_1}{\outputa_2}\atargetp{2}{\outputa_3}{\outputa_4}, \\
\forall (\outputa_1,\outputa_2,\outputa_3,\outputa_4) &\in \left\{(1,1,1,1),(1,1,1,2),(1,1,2,2),(1,2,2,2)\right\} \st \nonumber\\
\sum_{\appmat\in\appmatsetreduced} \isource{\hv}{\appmat} \frac{1}{4!\cdot 4!} 
&\sum_{\substack{
\dsvalalpha_1,\dots,\dsvalalpha_4 \allneq \\
\dsvalbeta_1,\dots,\dsvalbeta_4 \allneq
}}
\delta(\appmat_{\dsvalalpha_1,\dsvalbeta_1} = \outputa_1)
\delta(\appmat_{\dsvalalpha_2,\dsvalbeta_2} = \outputa_2)
\delta(\appmat_{\dsvalalpha_3,\dsvalbeta_3} = \outputa_3)
\delta(\appmat_{\dsvalalpha_4,\dsvalbeta_4} = \outputa_4) \nonumber\\
&\qquad= \frac{1}{16}, \\
\forall (\outputa_1,\outputa_2),(\outputa_3,\outputa_4) &\in \big\{(1,1),(1,2),(2,2)\big\},\outputa_5 \in \{1,2\}, (\outputa_1,\outputa_2,\outputa_3,\outputa_4,\outputa_5) \neq (2,2,2,2,2) \st \nonumber\\
\sum_{\appmat\in\appmatsetreduced} \isource{\hv}{\appmat} \frac{1}{4!\cdot 4!} 
&\sum_{\substack{
\dsvalalpha_1,\dots,\dsvalalpha_4 \allneq \\
\dsvalbeta_1,\dots,\dsvalbeta_4 \allneq
}}
\delta(\appmat_{\dsvalalpha_1,\dsvalbeta_1} = \outputa_1)
\delta(\appmat_{\dsvalalpha_1,\dsvalbeta_2} = \outputa_2)
\delta(\appmat_{\dsvalalpha_2,\dsvalbeta_3} = \outputa_3)
\delta(\appmat_{\dsvalalpha_3,\dsvalbeta_3} = \outputa_4)
\delta(\appmat_{\dsvalalpha_4,\dsvalbeta_4} = \outputa_5) \nonumber\\
&\qquad= \frac{1}{2}\atargetp{1}{\outputa_1}{\outputa_2}\atargetp{2}{\outputa_3}{\outputa_4}.
\end{align}
\end{subequations}

\paragraph{Explicit dual optimization problem.} 
The problem of equation \eqref{eq:app inf dual} can be rewritten as
\begin{subequations}
\begin{align}
\adual = &\min_{z_{11},z_{12},z_{21},z_{22}\in\mathbb{R}} \frac{1}{4} \sum_{\outputa,\outputb\in\{1,2\}} z_{\outputa\outputb} \\
&\textup{s.t. }\forall \appmat \in \appmatset\ (\textup{or }\appmatsetreduced): \nonumber
\end{align}
\vspace{-25pt}
\begin{multline}
\sum_{\outputa,\outputb} z_{\outputa\outputb} \frac{1}{12\cdot12}\sum_{\substack{\dsvalalpha_1,\dsvalalpha_2 \allneq\\\dsvalbeta_1,\dsvalbeta_2\allneq}}
\delta(\appmat_{\dsvalalpha_1,\dsvalbeta_1} = \outputa)\delta(\appmat_{\dsvalalpha_2,\dsvalbeta_2} = \outputb)
\\
\geq
\frac{1}{2}
\sum_{\outputa}
\left[
\frac{1}{4\cdot 12} \sum_{\substack{\dsvalalpha\\\dsvalbeta_1\neq\dsvalbeta_2}}
\delta(\appmat_{\dsvalalpha,\dsvalbeta_1} = \outputa)
\delta(\appmat_{\dsvalalpha,\dsvalbeta_2} = \outputa)
+
\frac{1}{12\cdot4} \sum_{\substack{\dsvalalpha_1\neq\dsvalalpha_2\\\dsvalbeta}}
\delta(\appmat_{\dsvalalpha_1,\dsvalbeta} = \outputa)
\delta(\appmat_{\dsvalalpha_2,\dsvalbeta} = \outputa)
\right].
\end{multline}
\end{subequations}

\newpage
\subsection{Extended plot of the feasible region}
\label{sec:extended plot}

In \cref{fig:bound on lambdas}, we show how the inflation described in \cref{def:app inf set} behaves on the whole region $[0,1/2]^{\times 2}$ which corresponds, thanks to \cref{lem:sc lambdas}, to the set of pairs of distributions with $2\times 2$ outcomes that are symmetric under the exchange of the two outcomes, and that have uniform marginals.
We saw in \cref{lem:bound lambdas} that if this pair of distributions is in $\appfeasible$ (see e.g.\ equation \eqref{eq:app feasible}), then it must further verify $\lambda_1,\lambda_2 \geq 1/4$. This is not the case according to our outer approximation of \cref{def:app inf set}: as we see in \cref{fig:bound on lambdas}, the postselected inflation only enforces that $\lambda_1,\lambda_2$ are greater than $\approx 0.17$.

\begin{figure}[h!]
	\centering
	\includegraphics{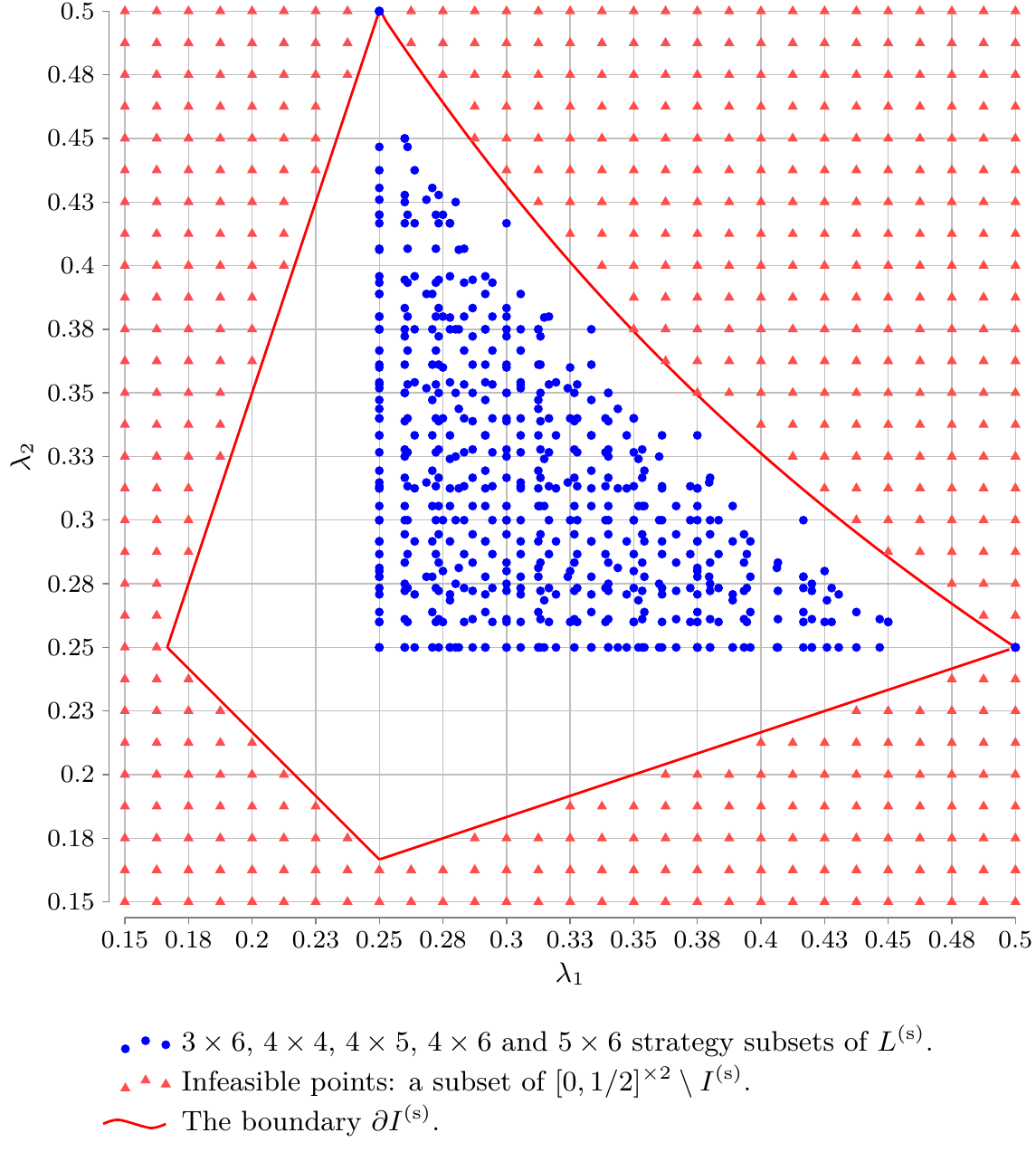}
	\caption{Behavior of the postselected inflation of \cref{def:app inf set} on the whole $[0,1/2]^{\times 2}$ region. In fact, if either of the two $\lambda_1,\lambda_2$ is less than $0.15$, we find that the corresponding distribution does not admit a postselected inflation, so we restrict the range of the plot to better visualize the boundary $\setboundary\appinffeasibleparam$. This boundary is obtained similarly to that of \cref{fig:main plot}: we ran a dichotomic search of the threshold radius (with respect to feasibility of the corresponding linear program) in polar coordinates centered around $(\lambda_1=0.25,\lambda_2=0.25)$, which makes sense given the overall shape of $\appinffeasibleparam$ shown with the scan of infeasible points.}
	\label{fig:bound on lambdas}
\end{figure}

\end{document}